\documentclass[aps,prd,preprint,nofootinbib,superscriptaddress]{revtex4-2}

\usepackage{amsmath,amssymb,mathtools,bm}
\usepackage{booktabs}
\usepackage{hyperref}

\newcommand{\C}{\mathbb C}
\newcommand{\R}{\mathbb R}
\newcommand{\Tr}{\operatorname{Tr}}
\newcommand{\rank}{\operatorname{rank}}
\newcommand{\Span}{\operatorname{span}}
\newcommand{\ii}{\mathrm{i}}
\newcommand{\Ga}{\mathbb G_{\mathrm a}}
\newcommand{\Tf}{T_{\mathrm f}}
\newcommand{\cI}{\mathcal I}
\newcommand{\cU}{\mathcal U}
\newcommand{\cD}{\mathcal D}
\newcommand{\cR}{\mathcal R}

\newtheorem{theorem}{Theorem}[section]
\newtheorem{proposition}[theorem]{Proposition}
\newtheorem{corollary}[theorem]{Corollary}

\begin{document}

\title{Invariant Algebras of Low-Rank Neutrino Matter Flows:\\
Krylov--Pl\"ucker Invariants, the $3+1$ Toric Ring, and Controlled Deformations}

\author{Jianlong Lu}
\email{jianlong@nus.edu.sg}
\affiliation{
Department of Mathematics, National University of Singapore,
Singapore 119076
}

\date{\today}

\begin{abstract}
Matter changes the eigenvalues and mixing matrix of the neutrino propagation Hamiltonian while entering the flavor-basis Hamiltonian through a low-dimensional additive deformation.  We develop a unified algebraic description of the quantities that are insensitive to this deformation.  The invariant algebra is formulated as the joint kernel of commuting matter-translation derivations inside the flavor-rephasing invariant ring.  For diagonal matter potentials, this gives a general decomposition into unaffected diagonal directions and an off-diagonal cycle algebra, together with an exact dimension formula for an arbitrary number of flavors and independent matter parameters.  For matter spurions with common low-rank support, block-Krylov rows generate exterior-power covariants whose maximal minors are matter independent.  Their mass-basis form is a Pl\"ucker expansion: rank one produces the familiar Vandermonde factorization, whereas higher rank generically produces a sum of independent Pl\"ucker terms.  The standard $3+1$ system provides our principal nontrivial example.  Its global off-diagonal algebra is the toric cycle ring of the complete four-vertex flavor graph, generated by edge moduli, triangle cycles, and quadrangle cycles; the commonly used eleven invariants constitute a generic local coordinate system rather than a global polynomial generating set.  We further formulate independently varying matter compositions and off-diagonal nonstandard interactions as, respectively, commuting exact flows and controlled deformations of the invariant algebra.  This framework separates the existence of exact invariants from the stronger and exceptional property of single-monomial spectral factorization.
\end{abstract}

\makeatletter
\@booleanfalse\preprintsty@sw
\maketitle
\@booleantrue\preprintsty@sw
\makeatother

\newpage

\section{Introduction}
\label{sec:introduction}

Neutrino oscillations in matter are governed by a simple tension between two descriptions of the same Hamiltonian.  In the instantaneous mass basis, the eigenvalues and mixing vectors depend nonlinearly on the matter density and exhibit resonant behavior.  In the flavor basis, however, the standard matter effect is an additive deformation of a vacuum Hamiltonian.  For an ultrarelativistic neutrino of energy $E$, it is convenient to write
\begin{equation}
 M(\bm a) \equiv 2E H(\bm a)
 =M_0+\sum_{A=1}^{k}a_A P_A ,
 \label{eq:intro-master-flow}
\end{equation}
where the Hermitian matrices $P_A$ specify the flavor structure of the medium and the real parameters $a_A$ contain its density and energy dependence.  The mismatch between the complicated spectral response and the elementary flavor-basis translation is the source of exact matter identities.

Throughout this work the flavor-universal part is quotiented out.  Explicitly, we set
\begin{equation}
 \overline M=M-\frac{\Tr M}{n}\bm 1,
 \qquad
 \overline P_A=P_A-\frac{\Tr P_A}{n}\bm 1,
 \label{eq:intro-traceless-representatives}
\end{equation}
and use $M$ and $P_A$ for these traceless representatives unless an identity component is restored explicitly.  This removes the unobservable common oscillation phase.

Several important instances have long been known.  The Toshev relation and the Naumov--Harrison--Scott identity isolate combinations of three-flavor mixing parameters and mass splittings that are unchanged by the standard electron potential~\cite{Toshev1991,Naumov1992,HarrisonScott2000}.  Hamiltonian and Vandermonde methods also lead to the exact KTY relations~\cite{KimuraTakamuraYokomakura2002}.  Harrison and Scott identified a complete set of five three-flavor matter-invariant observables~\cite{HarrisonScott2002}.  Evolution equations with the induced matter mass as the independent variable
recast the matter problem as a flow of effective masses and rephasing-invariant
mixing variables~\cite{ChiuKuoLiu2010,ChiuKuo2018}. Their covariance under
mass-eigenstate relabelings follows from permutation symmetry
~\cite{KuoChiu2020}, while the continuous Lie symmetries of the flow equations
were classified in Ref.~\cite{Zhou2022}; the RGE treatment has also been extended to a light sterile neutrino~\cite{ZengXu2022}.  Eigenvector--eigenvalue and adjugate methods provide complementary exact descriptions~\cite{DentonParkeTaoZhang2022,AbdullahiParke2022}; in particular, generalized CP-even and Naumov--Harrison--Scott identities are available for four or more flavors.  Four-neutrino mass and CP sum rules were obtained from commutator identities and direct mixing relations~\cite{Xing2001,Zhang2006}.  More recently, the matter-flow, permutation-covariant, and spurion viewpoints have been combined to organize five three-flavor first integrals and an eleven-dimensional generic local transcendence basis for a standard $3+1$ flow~\cite{WangZhou2026}.

These results leave three structural questions open.  First, the observed contrast between rank-one and rank-two matter potentials suggests a general low-rank theorem, but parameter counting or cancellation of individual spectral poles does not by itself identify the higher-rank invariants.  Second, a set of algebraically independent local coordinates is not the same object as a global polynomial invariant ring: rational reconstruction may divide by quantities that vanish on physically relevant strata.  Third, realistic extensions introduce several independently varying diagonal potentials, small flavor-nonuniversal radiative contributions, or off-diagonal nonstandard interactions (NSI).  A systematic treatment should determine which old invariants survive, how a new translation direction changes the generators, and when a correction ceases to exist within a precisely declared factorized or bounded-degree class.

The purpose of this work is to address these questions in one framework.  The central principle is that matter invariants are joint invariants of two actions: flavor rephasing and additive translation along the matter-spurion directions.  This viewpoint separates three notions that are often conflated:
\begin{enumerate}
 \item a \emph{first integral}, which is constant along a specified matter trajectory;
 \item a \emph{global polynomial invariant}, which is regular on the full algebraic quotient;
 \item a \emph{factorized spectral invariant}, which has the stronger form of a single product of mixing and eigenvalue factors.
\end{enumerate}
The third property is substantially more restrictive than the first two.

\subsection{Translation--rephasing invariant algebra}

Let
\begin{equation}
 \Tf \equiv (\C^\times)^n/\C^\times
 \label{eq:intro-rephasing-torus}
\end{equation}
denote the complexified flavor-rephasing torus.  It acts by simultaneous diagonal conjugation on $M$ and on all matter spurions.  Treating the $P_A$ as spurions is essential when off-diagonal NSI are present, because it makes the construction basis covariant.  On the polynomial coordinate algebra we introduce the commuting translation derivations
\begin{equation}
 \cD_A
 =\sum_{\alpha,\beta}(P_A)_{\alpha\beta}
 \frac{\partial}{\partial M_{\alpha\beta}} .
 \label{eq:intro-derivations}
\end{equation}
Before a physical matter potential is fixed, the universal spurion-covariant algebra is
\begin{equation}
 \cU_{n,k}
 =
 \C[M,P_1,\ldots,P_k]^{\Tf}
 \cap\bigcap_{A=1}^{k}\ker\cD_A .
 \label{eq:intro-master-algebra}
\end{equation}
This ring includes invariants formed entirely from the spurions.  Those quantities are external coefficient data rather than neutrino observables.  For a fixed diagonal tuple $\bm P=(P_1,\ldots,P_k)$, which is preserved by $\Tf$, the relevant fiber is instead
\begin{equation}
 \cI_{n,k}(\bm P)
 =\C[M]^{\Tf}\cap\bigcap_{A=1}^{k}\ker\cD_A .
 \label{eq:intro-fixed-fiber}
\end{equation}
The physical Hermitian algebra is the real form defined by $M_{\beta\alpha}=M_{\alpha\beta}^*$.  All dimensions below refer to the fixed-spurion fiber~\eqref{eq:intro-fixed-fiber}.  For off-diagonal NSI, a fixed numerical spurion is not preserved by the full rephasing torus; the universal relative formulation~\eqref{eq:intro-master-algebra} is retained until after invariant formation and only then specialized.

For $k$ independent traceless diagonal matter directions, we show that the invariant algebra decomposes into an unaffected diagonal polynomial algebra and an off-diagonal cycle algebra,
\begin{equation}
 \cI^{\mathrm{diag}}_{n,k}
 \simeq
 \C[d_1,\ldots,d_{n-1-k}]
 \otimes
 \C[z_{\alpha\beta}\,|\,\alpha\ne\beta]^{\Tf},
 \qquad z_{\alpha\beta}=M_{\alpha\beta} .
 \label{eq:intro-diagonal-factorization}
\end{equation}
The generic dimension is therefore
\begin{equation}
 \dim \cI^{\mathrm{diag}}_{n,k}
 =(n-1-k)+(n-1)^2
 =n(n-1)-k .
 \label{eq:intro-dimension}
\end{equation}
This result covers fixed-composition matter, independently varying charged- and neutral-current densities, and arbitrary diagonal flavor potentials within a single statement.

\subsection{Low-rank Krylov invariants}

Oscillation observables are insensitive to an additive multiple of the identity.  We therefore define the effective rank of one matter direction by
\begin{equation}
 r_{\mathrm{eff}}(P)
 \equiv
 \min_{c\in\R}\rank(P-c\bm 1),
 \label{eq:intro-effective-rank}
\end{equation}
and write, at a minimizing value of $c$,
\begin{equation}
 P=c\bm 1+SDS^\dagger,
 \label{eq:intro-low-rank-spurion}
\end{equation}
where $S$ has $r=r_{\mathrm{eff}}(P)$ columns and $D$ is nonsingular.  More generally, several matter directions have common effective support of dimension $r$ when, after independent scalar shifts, they admit a simultaneous representation
\begin{equation}
 P_A=c_A\bm 1+SD_AS^\dagger,
 \qquad A=1,\ldots,k,
 \label{eq:intro-common-support}
\end{equation}
with one fixed $n\times r$ matrix $S$.  We take $r$ minimal among such representations.  Introduce the block rows
\begin{equation}
 R_p=S^\dagger M^p,
 \qquad
 K_q=
 \begin{pmatrix}
 R_0\\R_1\\ \vdots\\R_{q-1}
 \end{pmatrix}.
 \label{eq:intro-block-krylov}
\end{equation}
We prove for every matter direction and $p\geq1$ that
\begin{equation}
 \frac{\partial R_p}{\partial a_A}
 =pc_A R_{p-1}
 +\sum_{j=0}^{p-1}
 (S^\dagger M^jS)D_A R_{p-1-j},
 \label{eq:intro-krylov-flow}
\end{equation}
with $\partial R_0/\partial a_A=0$, so that
$\partial K_q/\partial a_A=L_{A,q}K_q$ with $L_{A,q}$ strictly
block-lower triangular.  It follows that the exterior product of all rows of
$K_q$ is constant along every flow.  When $qr\le n$, the maximal column
minors of $K_q$ are components of a matter-constant exterior covariant; when
$qr=n$, $\det K_q$ is matter independent in a fixed frame.  These components
are not generally flavor-rephasing scalars.  Physical polynomial invariants
are obtained from norm squares, more general rephasing-balanced products, or
invariant contractions.  The triangular row-span identity is dual to the
classical control-theory invariance of observability under output
injection~\cite{KarcaniasLimantsevaHalikias2021}; the results developed here
concern its exact neutrino-matter covariants, spectral Pl\"ucker representation,
and factorization consequences.

For $(n,r,q)=(3,1,3)$, this construction gives the determinant underlying the
electron-row KTY identity.  For the standard $3+1$ potential, with
$P_\alpha=|\alpha\rangle\langle\alpha|$,
\begin{equation}
 P=P_e+\eta P_s,
 \label{eq:intro-standard-31-spurion}
\end{equation}
the effective rank is two for $\eta\ne0$ (with a nonunique minimizing support
at $\eta=1$), whereas it drops to one at $\eta=0$.  Choosing the
electron--sterile support, $(n,r,q)=(4,2,2)$ gives, up to ordering conventions,
\begin{equation}
 \det K_2
 =\pm
 \det\begin{pmatrix}
 M_{e\mu}&M_{e\tau}\\
 M_{s\mu}&M_{s\tau}
 \end{pmatrix}.
 \label{eq:intro-rank-two-minor}
\end{equation}
Its rephasing-invariant norm is already a polynomial in the cycle notation introduced below,
\begin{equation}
 |\det K_2|^2
 =Q_{e\mu}Q_{s\tau}+Q_{e\tau}Q_{s\mu}
 -2\operatorname{Re}C_{e\mu s\tau}.
 \label{eq:intro-krylov-cycle-bridge}
\end{equation}
Thus rank two does not eliminate compact matter invariants, and the Krylov construction does not add an independent observable to the global ring.  It supplies a distinguished determinantal representative whose spectral structure differs from rank one.

Indeed, if $M=V\Lambda V^\dagger$ and $X=S^\dagger V$, the square
Krylov determinant is a block alternant built from $X\Lambda^p$.  Its Laplace
expansion is a sum over ordered partitions into $r$-element subsets, with
coefficients given by Pl\"ucker coordinates of $X$.  At rank one the alternant
resums and factorizes into the product of one mixing entry from every column
and the ordinary spectral Vandermonde.  At generic higher rank, several
independent Pl\"ucker terms remain, and equality of two eigenvalues need not
make the corresponding block columns proportional.  For the square case
$n=qr$ with $r,q>1$, we prove that there is no universal nonconstant spectral
divisor and no rank-one-type separation into a Pl\"ucker-dependent factor
times a spectral polynomial; in particular, no pairwise eigenvalue difference
is a generic divisor.  We then give a complete spectral reducibility
classification at fixed Pl\"ucker point for the physically central
$4\times4$, rank-two case.  At general $(n,r)$ we identify mechanisms forcing
vanishing or special factorization, including rank reduction, matroid
basis-packing failure, parallel classes, invariant subspaces, and
high-multiplicity spectral degeneracy.  Lower-multiplicity degeneracies are
analyzed separately and need not force vanishing.  We do not claim
unrestricted irreducibility for all $(n,r)$ or an exhaustive factorization
classification on every matroid stratum.

\subsection{The global standard $3+1$ ring}

The off-diagonal part of Eq.~\eqref{eq:intro-diagonal-factorization} is a toric cycle ring.  Under flavor rephasing,
\begin{equation}
 z_{\alpha\beta}\longmapsto
 t_\alpha t_\beta^{-1}z_{\alpha\beta}.
 \label{eq:intro-edge-weight}
\end{equation}
An invariant monomial is therefore a nonnegative integral circulation on the directed flavor graph.  Such a circulation decomposes into directed simple cycles.  For four flavors, the complex Hilbert basis contains six two-cycles, eight oriented triangles, and six oriented Hamiltonian quadrangles.  On the Hermitian real form these become
\begin{align}
 Q_{\alpha\beta}&=|M_{\alpha\beta}|^2,
 &&6\ \text{real generators},
 \label{eq:intro-edge-generator}\\
 T_{\alpha\beta\gamma}
 &=M_{\alpha\beta}M_{\beta\gamma}M_{\gamma\alpha},
 &&4\ \text{complex generators},
 \label{eq:intro-triangle-generator}\\
 C_{\alpha\beta\gamma\delta}
 &=M_{\alpha\beta}M_{\beta\gamma}
 M_{\gamma\delta}M_{\delta\alpha},
 &&3\ \text{complex generators}.
 \label{eq:intro-quadrangle-generator}
\end{align}
Together with two diagonal invariants for the one-parameter standard flow, these give 22 real global generators of an eleven-dimensional invariant algebra.  The surplus generators are tied by binomial cycle relations and their real and imaginary consequences.  Representative identities are
\begin{align}
 T_{\alpha\beta\gamma}T_{\alpha\gamma\beta}
 &=Q_{\alpha\beta}Q_{\beta\gamma}Q_{\gamma\alpha},
 \label{eq:intro-triangle-relation}\\
 T_{\alpha\beta\gamma}T_{\alpha\gamma\delta}
 &=Q_{\alpha\gamma}C_{\alpha\beta\gamma\delta}.
 \label{eq:intro-quadrangle-relation}
\end{align}
Equation~\eqref{eq:intro-quadrangle-relation} makes the global-versus-local distinction explicit: a quadrangle can be reconstructed rationally from two triangles only on the open set $Q_{\alpha\gamma}\ne0$.  We determine the toric ideal and its multigraded Hilbert series, translate the result into CP-even and CP-odd real generators, and stratify the quotient by the graph of nonzero edges.  Loss of cycles can reduce the number of physical loop phases even while the support remains connected; the rephasing stabilizer is enhanced only when the nonzero support graph becomes disconnected.  The CP-conserving locus and the spectral-discriminant locus are treated separately: the former belongs to the real cycle algebra, whereas the latter marks failure of spectral coordinates and need not coincide with a toric singularity.  The use of cycle generators is consistent with general quiver invariant theory~\cite{CrawYamagishi2023}, while Hilbert-series and graph algorithms have proved effective in substantially larger flavor-invariant problems~\cite{GrojeanKleyLeflotYao2024}.

\subsection{Independent matter directions and controlled breaking}

For independently varying diagonal densities, the derivations commute and the common exact algebra is given by Eq.~\eqref{eq:intro-diagonal-factorization}.  In particular, allowing the electron and sterile or neutron contributions in a $3+1$ system to vary independently reduces the number of invariant diagonal directions from two to one, while leaving the nine-dimensional off-diagonal cycle algebra unchanged.

Off-diagonal NSI are qualitatively different: they translate edge variables and hence deform the cycle algebra itself.  Let
\begin{equation}
 P_{\mathrm{eff}}=P_0+\epsilon Q,
 \qquad
 \cD_{\mathrm{eff}}=\cD_0+\epsilon\cD_Q .
 \label{eq:intro-deformed-flow}
\end{equation}
For $I\in\ker\cD_0$, a first-order corrected invariant
\begin{equation}
 I_\epsilon=I+\epsilon J+\mathcal O(\epsilon^2)
 \label{eq:intro-deformed-invariant}
\end{equation}
must satisfy
\begin{equation}
 \cD_0J=-\cD_QI
 \label{eq:intro-cohomological-equation}
\end{equation}
at first order.  On the unrestricted polynomial algebra a standard diagonal translation admits a rephasing-invariant slice, so Eq.~\eqref{eq:intro-cohomological-equation} has a polynomial solution.  Nontrivial obstructions arise only after the admissible correction space is specified---for example, the old edge-only ring, a fixed edge multidegree, a prescribed Laurent-monomial span, or a factorized spectral ansatz.  We therefore distinguish generators that remain exactly unchanged, generators that require a correction, and generators that cannot be corrected within a declared restricted class; we do not claim a cohomological obstruction in the unrestricted polynomial ring or under a total $M$-degree bound that contains the explicit slice solution.  Off-diagonal NSI are already known to induce matter dependence in otherwise invariant combinations~\cite{AbdullahiParke2022}.

Electroweak radiative violation of neutral-current flavor universality provides a different, controlled Standard Model example.  The corresponding matter correction remains diagonal, so the off-diagonal edge, triangle, and quadrangle invariants remain exact.  For a fixed-composition corrected potential the flow is still one dimensional: the number of invariant diagonal directions is unchanged, although the surviving diagonal combinations are rotated and the effective support rank can increase.  The diagonal invariant count decreases only if the radiative contribution is allowed to vary as an independent matter direction.  The rank lift changes the special spectral factorization underlying
parametrization-dependent identities such as the Toshev relation
~\cite{Toshev1991}. The flavor-nonuniversal one-loop matter potential
responsible for this lift was derived in Ref.~\cite{BotellaLimMarciano1987},
its consequences for effective mixing in three-flavor and $3+1$ schemes were
analyzed in Ref.~\cite{Zhu2020}, and the resulting modification of the Toshev
relation was worked out explicitly in Ref.~\cite{XingZhu2022}.  This separation between an unchanged cycle algebra and a changed spectral factorization is one of the applications of the present framework.

\subsection{Organization and scope}

Section II develops the joint translation--rephasing algebra and fixes the complex and real forms used throughout.  Section III proves the structure theorem for multiple diagonal matter directions.  Sections IV and V establish the block-Krylov construction and its Pl\"ucker expansion, including the precisely delimited factorization theorem.  Section VI gives the complete standard $3+1$ toric presentation, and Sec. VII analyzes its syzygies, CP structure, and exceptional strata.  Section VIII treats exact additional matter directions and controlled NSI and radiative deformations.  Section IX reports independent symbolic and numerical checks, and Sec. X summarizes the results.

The paper is restricted to finite-dimensional Hermitian propagation Hamiltonians with unitary instantaneous diagonalization.  General density profiles affect the path-ordered evolution operator, but not the local algebraic statements when the relevant matter directions are included among the $P_A$.  Non-Hermitian decay, decoherence, collective neutrino self-interactions, and full experimental sensitivity forecasts lie outside the present scope.  This restriction keeps the distinction between an invariant of the local Hamiltonian family and a property of the complete propagation history explicit.

\section{Translation--rephasing invariant algebra}
\label{sec:translation-algebra}

The additive form of the matter Hamiltonian suggests taking a quotient by
translations.  Physical quantities, however, must also be independent of the
arbitrary phases assigned to flavor states.  This section puts these two
requirements into a single algebraic action.  We work over $\C$ to use the
coordinate ring of an affine space and return to the Hermitian real form at the
end.  No diagonalization is used, so all statements remain regular at spectral
degeneracies.

\subsection{Trace quotient and complexified Hamiltonian space}
\label{subsec:trace-complexification}

Let
\begin{equation}
 \mathfrak h_n^0
 =\{H\in\operatorname{Mat}_n(\C):H^\dagger=H,\ \Tr H=0\}
 \label{eq:hermitian-traceless-space}
\end{equation}
be the real vector space of traceless Hermitian matrices.  Quotienting a
Hermitian Hamiltonian by the unobservable identity direction is canonically
implemented by
\begin{equation}
 \pi_0(X)=X-\frac{\Tr X}{n}\bm 1.
 \label{eq:trace-projection}
\end{equation}
Thus a physical matter direction $P_A^{\rm phys}$ induces
$P_A=\pi_0(P_A^{\rm phys})$; a flavor-universal direction becomes zero.  The
complexification of Eq.~\eqref{eq:hermitian-traceless-space} is
\begin{equation}
 V_n\equiv\mathfrak h_n^0\otimes_{\R}\C
 \simeq\mathfrak{sl}_n(\C).
 \label{eq:complexified-hamiltonian-space}
\end{equation}
Its matrix coordinates are algebraically independent apart from the trace
relation.  For one Hamiltonian and $k$ universal spurions, we therefore use
\begin{align}
 \cR_{n,k}
 &\equiv \C[V_n\oplus V_n^{\oplus k}] \nonumber\\
 &\simeq
 \frac{\C[m_{\alpha\beta},p_{A,\alpha\beta}\;|\;
          1\leq\alpha,\beta\leq n,\ 1\leq A\leq k]}
 {\left(\sum_\alpha m_{\alpha\alpha},
         \sum_\alpha p_{A,\alpha\alpha}\ (1\leq A\leq k)\right)}.
 \label{eq:universal-coordinate-ring}
\end{align}
The variables $m_{\alpha\beta}$ and $p_{A,\alpha\beta}$ are independent in
this complex coordinate ring.  Hermiticity is not imposed by treating one as
the complex conjugate of another until the real form is selected below.

There are two distinct uses of scalar shifts in this paper.  Equation
\eqref{eq:trace-projection} selects the unique traceless representative used by
the translation algebra.  In Sec.~\ref{sec:krylov-exterior} we temporarily
restore a scalar representative $P_A-c_A\bm 1$ in order to minimize its matrix
rank.  That later operation measures a property of the scalar-shift equivalence
class and does not undo the physical trace quotient.

\subsection{Flavor rephasing and the weight lattice}
\label{subsec:rephasing-weights}

Diagonal conjugation by $(\C^\times)^n$ has a one-dimensional kernel, giving
the effective flavor torus
\begin{equation}
 \Tf=(\C^\times)^n/\C^\times_{\rm diag},
 \qquad \dim\Tf=n-1.
 \label{eq:flavor-torus-section2}
\end{equation}
It acts simultaneously on the Hamiltonian and every spurion:
\begin{equation}
 (t\mathbin{\cdot}M)_{\alpha\beta}
 =t_\alpha t_\beta^{-1}m_{\alpha\beta},
 \qquad
 (t\mathbin{\cdot}P_A)_{\alpha\beta}
 =t_\alpha t_\beta^{-1}p_{A,\alpha\beta}.
 \label{eq:simultaneous-torus-action}
\end{equation}
We assign weights according to this transformation of the matrix entries.  A
pullback convention for coordinate functions reverses every displayed weight
and leaves all zero-weight conclusions unchanged.  The character lattice is
\begin{equation}
 X^*(\Tf)=\left\{q\in\mathbb Z^n:\sum_{\alpha=1}^n q_\alpha=0\right\},
 \label{eq:torus-character-lattice}
\end{equation}
and an off-diagonal entry has weight $e_\alpha-e_\beta$, whereas every
diagonal entry has weight zero.

For a monomial in any collection of variables with the same directed-edge
weights, let $N_{\alpha\beta}$ be its total exponent on the edge
$\alpha\to\beta$.  Its weight is
\begin{equation}
 \operatorname{wt}(x^N)
 =\sum_{\alpha=1}^n
 \left(\sum_{\beta\ne\alpha}N_{\alpha\beta}
       -\sum_{\beta\ne\alpha}N_{\beta\alpha}\right)e_\alpha.
 \label{eq:monomial-flow-weight}
\end{equation}
It is rephasing invariant precisely when the exponent flow is conserved at
every flavor vertex.  This elementary weight-balance condition is the origin
of the directed-cycle algebra used in Secs.~VI and VII.  More generally,
torus invariants and their orbit
structure are standard topics in computational invariant theory
\cite{DerksenKemper2015,BurgisserDoganMakamWalterWigderson2021}.

The physical rephasing group is the compact real torus
$K_{\rm f}=U(1)^n/U(1)$.  It is Zariski dense in $\Tf$.  Consequently a complex
polynomial is invariant under physical flavor rephasings if and only if it is
invariant under the complexified torus.  Complexification therefore introduces
no additional polynomial invariance condition.

\subsection{Matter translations and their infinitesimal generators}
\label{subsec:additive-matter-action}

On the universal space define an additive algebraic action by
\begin{equation}
 \tau_{\bm a}(M,P_1,\ldots,P_k)
 =\left(M+\sum_{A=1}^k a_AP_A,P_1,\ldots,P_k\right),
 \qquad \bm a\in\Ga^k.
 \label{eq:universal-translation-action}
\end{equation}
Its infinitesimal generators are the derivations
\begin{equation}
 \cD_A(m_{\alpha\beta})=p_{A,\alpha\beta},
 \qquad
 \cD_A(p_{B,\alpha\beta})=0,
 \label{eq:derivation-on-generators}
\end{equation}
or, in ambient matrix coordinates,
\begin{equation}
 \cD_A=\sum_{\alpha,\beta}p_{A,\alpha\beta}
 \frac{\partial}{\partial m_{\alpha\beta}}.
 \label{eq:universal-derivation-section2}
\end{equation}
This formula descends to the quotient~\eqref{eq:universal-coordinate-ring}
because $\cD_A(\Tr M)=\Tr P_A=0$.

\begin{theorem}[Translation--rephasing joint kernel]
\label{thm:translation-rephasing-kernel}
The $\cD_A$ are commuting locally nilpotent derivations of $\cR_{n,k}$.
Their exponentials generate Eq.~\eqref{eq:universal-translation-action}, the
$\Ga^k$ and $\Tf$ actions commute, and
\begin{equation}
 \cU_{n,k}
 =\cR_{n,k}^{\Tf\times\Ga^k}
 =\cR_{n,k}^{\Tf}\cap\bigcap_{A=1}^k\ker\cD_A .
 \label{eq:joint-kernel-theorem}
\end{equation}
\end{theorem}

\emph{Proof.}
If $F$ has degree $d$ in the $M$ coordinates, then
$\cD_A^{d+1}F=0$, proving local nilpotence.  Every $\cD_A$ annihilates every
$P_B$, so $[\cD_A,\cD_B]=0$.  The finite Taylor series gives
\begin{equation}
 F\left(M+\sum_Aa_AP_A,\bm P\right)
 =\exp\left(\sum_Aa_A\cD_A\right)F(M,\bm P).
 \label{eq:finite-translation-exponential}
\end{equation}
In characteristic zero, $F$ is invariant under $\Ga^k$ exactly when every
$\cD_AF$ vanishes.  Finally,
\begin{equation}
 t\left(M+\sum_Aa_AP_A\right)t^{-1}
 =tMt^{-1}+\sum_Aa_A(tP_At^{-1}),
 \label{eq:commuting-actions-proof}
\end{equation}
so simultaneous conjugation commutes with translation and each $\cD_A$
preserves $\cR_{n,k}^{\Tf}$.  This proves Eq.~\eqref{eq:joint-kernel-theorem}.
$\square$

The equivalence between additive-group actions and locally nilpotent
derivations is reviewed in Ref.~\cite{Freudenburg2017}.  Here it follows
directly from Eq.~\eqref{eq:finite-translation-exponential}; no general finite
generation theorem for vector-group invariants is assumed.  A physical
polynomial constant on an open set of real matter parameters has vanishing
real directional derivatives.  The resulting polynomial identity in the
matter parameters then holds over $\C$, making the polynomial constant on the
complexified $\Ga^k$ orbit and justifying the algebraic action.

The inclusion
\begin{equation}
 \C[P_1,\ldots,P_k]^{\Tf}\subset\cU_{n,k}
 \label{eq:spurion-coefficient-subring}
\end{equation}
is automatic because the translations do not alter the spurions.  After a
medium is selected, these elements are fixed coefficient data.  They should
not be counted as independent observables of the neutrino Hamiltonian.

\subsection{Fixed-spurion fibers and failure of naive specialization}
\label{subsec:universal-versus-fiber}

Let $\bm p=(p_1,\ldots,p_k)$ be a numerical traceless tuple and set
\begin{equation}
 W_{\bm p}=\Span_{\C}\{p_1,\ldots,p_k\}\subset V_n,
 \qquad \ell(\bm p)=\dim W_{\bm p}.
 \label{eq:fixed-translation-space}
\end{equation}
Only $\ell(\bm p)$, rather than the number of displayed parameters, counts
independent matter translations.  Before imposing rephasing, the fixed-fiber
translation algebra is especially simple.
For $w\in W_{\bm p}$ let
$\cD_w=\sum_{\alpha,\beta}w_{\alpha\beta}
\partial/\partial m_{\alpha\beta}$, and let $\Ga(W_{\bm p})$ denote the
additive vector group underlying $W_{\bm p}$.  Then
\begin{proposition}[Linear translation quotient]
\label{prop:linear-translation-quotient}
\begin{equation}
 \C[V_n]^{\Ga(W_{\bm p})}
 =\bigcap_{w\in W_{\bm p}}\ker\cD_w
 \simeq \operatorname{Sym}(W_{\bm p}^{\circ})
 \simeq \C[V_n/W_{\bm p}],
 \label{eq:translation-quotient-linear}
\end{equation}
\end{proposition}
where $W_{\bm p}^{\circ}\subset V_n^*$ is the annihilator.  To see this,
choose a vector-space complement $V_n=W_{\bm p}\oplus U$.  The independent
matter translations become ordinary translations in the first $\ell$
coordinates, and their joint kernel consists exactly of polynomials in the
$U$ coordinates.  Thus the translation quotient itself is linear.  Nontrivial
cycle relations arise when a rephasing subgroup preserving $W_{\bm p}$ is
subsequently imposed; in particular, the full $\Tf$ acts when $W_{\bm p}$ is
diagonal.

If every $p_A$ is diagonal, $\Tf$ fixes the tuple pointwise and the two actions
continue to commute on the fiber.  In that case
\begin{equation}
 \cI_{n,k}(\bm p)
 =\left[\operatorname{Sym}(W_{\bm p}^{\circ})\right]^{\Tf}
 =\C[V_n/W_{\bm p}]^{\Tf}.
 \label{eq:fixed-diagonal-quotient}
\end{equation}
Section~III computes this ring explicitly.  If the
projected diagonal directions are independent, then $\ell=k$; otherwise every
dimension formula must use $\ell$.

Evaluation at a fixed tuple defines
\begin{equation}
 \operatorname{ev}_{\bm p}:\cU_{n,k}\longrightarrow\C[V_n],
 \qquad F(M,\bm P)\longmapsto F(M,\bm p).
 \label{eq:universal-specialization-map}
\end{equation}
For a diagonal tuple its image lies in Eq.~\eqref{eq:fixed-diagonal-quotient},
but it need not generate the fiber.  Invariant formation does not generally
commute with specialization.  The one-dimensional example
\begin{equation}
 \cD=p\frac{\partial}{\partial m},
 \qquad \ker_{\C[m,p]}\cD=\C[p]
 \label{eq:specialization-counterexample}
\end{equation}
makes the point: at $p=0$ the specialized derivation vanishes and the fiber
kernel is $\C[m]$, while evaluation of the universal kernel gives only
constants.  Rank-deficient and symmetry-enhanced matter fibers can similarly
acquire additional invariants.  Equality with a universal specialization will
therefore be asserted only after a separate generic-localization or
fixed-fiber proof.

For an off-diagonal numerical tuple, the full torus does not in general
preserve the fiber.  The residual pointwise rephasing group is
\begin{equation}
 T_{\mathrm f,\bm p}
 =\{t\in\Tf:t p_A t^{-1}=p_A\ \text{for every }A\}.
 \label{eq:fixed-spurion-stabilizer}
\end{equation}
Accordingly, the algebra obtained after fixing this coordinate representative
and quotienting by its residual pointwise rephasings is
\begin{equation}
 \cI^{\rm fix}(\bm p)
 =\left[\operatorname{Sym}(W_{\bm p}^{\circ})\right]^{T_{\mathrm f,\bm p}}.
 \label{eq:offdiagonal-fixed-fiber-ring}
\end{equation}
If only the translation subspace $W_{\bm p}$, rather than the ordered tuple,
is regarded as fixed, the relevant rephasing group is the normalizer of
$W_{\bm p}$ in $\Tf$ rather than the pointwise stabilizer in
Eq.~\eqref{eq:fixed-spurion-stabilizer}.
By contrast, a universal scalar obeys the relative covariance relation
\begin{equation}
 F(tMt^{-1},t\bm p\,t^{-1})=F(M,\bm p),
 \label{eq:relative-spurion-covariance}
\end{equation}
which does not imply invariance under $M\mapsto tMt^{-1}$ with $\bm p$ held
fixed.  This is why off-diagonal NSI are kept as transforming spurions until
after invariant formation.  The formal intersection
$\C[V_n]^{\Tf}\cap\ker\cD_p$ for fixed off-diagonal $p$ is not the invariant
ring of two commuting actions and can overimpose invariance along the entire
torus orbit of the intended translation direction.

\subsection{Hermitian real structure and orientation reversal}
\label{subsec:real-form-cp}

The Hermitian real form is selected by the antilinear involution $\rho$ on
$\cR_{n,k}$,
\begin{equation}
 \rho(c)=\overline c,
 \qquad
 \rho(m_{\alpha\beta})=m_{\beta\alpha},
 \qquad
 \rho(p_{A,\alpha\beta})=p_{A,\beta\alpha}.
 \label{eq:hermitian-real-involution}
\end{equation}
Its fixed subalgebra is the real coordinate algebra of the Hermitian form,
whose real points obey
$m_{\beta\alpha}=m_{\alpha\beta}^*$ and
$p_{A,\beta\alpha}=p_{A,\alpha\beta}^*$.  On the universal ring, $\rho$
commutes with every $\cD_A$.  After numerical specialization,
$\rho\cD_p\rho^{-1}=\cD_{p^\dagger}$; hence a Hermitian fixed fiber has a
$\rho$-stable translation kernel.

It is useful to distinguish this real structure from the complex-linear
orientation-reversal involution
\begin{equation}
 \kappa(m_{\alpha\beta})=m_{\beta\alpha},
 \qquad
 \kappa(p_{A,\alpha\beta})=p_{A,\beta\alpha},
 \qquad \kappa(c)=c.
 \label{eq:orientation-reversal-involution}
\end{equation}
For an oriented cycle
\begin{equation}
 Z_{\alpha_1\cdots\alpha_r}
 =m_{\alpha_1\alpha_2}m_{\alpha_2\alpha_3}\cdots
  m_{\alpha_r\alpha_1},
 \label{eq:generic-oriented-cycle}
\end{equation}
$\kappa$ reverses its orientation.  On the Hermitian locus it sends the cycle
value to its complex conjugate, so
\begin{equation}
 Z^{(+)}=\frac{Z+\kappa Z}{2},
 \qquad
 Z^{(-)}=\frac{Z-\kappa Z}{2\ii}
 \label{eq:cycle-even-odd-parts}
\end{equation}
are respectively even and odd under orientation reversal.  For real diagonal
CP-even spurions this is the usual algebraic CP grading of the mixing data.
More generally, $\kappa(\bm p)=\bm p$ is required for this grading to act
within one fixed fiber; otherwise $\kappa$ maps the fiber at $\bm p$ to the
fiber at $\kappa(\bm p)$.
It is not, by itself, the full neutrino-to-antineutrino transformation in a
material background, which also reverses the sign of the matter potential in
the propagation Hamiltonian.  The equivalence between vanishing odd cycle
data and the CP-conserving locus, including exceptional support graphs, is
therefore deferred to Sec.~VII.

\subsection{Support graphs, stabilizers, and quotient dimension}
\label{subsec:support-stabilizers}

For $M\in V_n$, let $G(M)$ be the underlying undirected graph on the $n$
flavors, with an edge $\{\alpha,\beta\}$ whenever at least one of
$m_{\alpha\beta}$ and $m_{\beta\alpha}$ is nonzero.  On the Hermitian locus
the two orientations vanish together.  If $G(M)$ has $c(M)$ connected
components, its rephasing stabilizer and orbit dimensions are as follows.
\begin{proposition}[Support-graph stabilizer]
\label{prop:support-graph-stabilizer}
\begin{equation}
 \dim\operatorname{Stab}_{\Tf}(M)=c(M)-1,
 \qquad
 \dim(\Tf\mathbin{\cdot}M)=n-c(M).
 \label{eq:support-graph-stabilizer}
\end{equation}
\end{proposition}
Indeed, stabilizing a nonzero edge requires $t_\alpha=t_\beta$.  Hence $t$ is
constant on each component, giving $c(M)$ constants before removing the
common diagonal scalar.  On the connected generic locus the torus action is
effective and its orbits have dimension $n-1$.  Removing an edge may destroy
a loop without changing this dimension; the stabilizer grows only when the
support graph disconnects.  A stabilizer jump is not automatically equivalent
to a singular point of the affine quotient, which requires a separate local
ring analysis in Sec.~VII.

For the simultaneous universal action on $(M,P_1,\ldots,P_k)$, the same proof
uses the union of the off-diagonal support graphs of all matrices in the tuple.
For a fixed diagonal-spurion fiber the spurions add no edges, so this union
reduces to $G(M)$.  For a frozen off-diagonal tuple, its spurion support first
determines $T_{\mathrm f,\bm p}$, and adding the support of $M$ can only reduce
the stabilizer of the full point.

The off-diagonal space has $n(n-1)$ complex coordinates.  Its weights occur
in opposite pairs $\pm(e_\alpha-e_\beta)$, so the zero-weight semigroup spans
the kernel of the weight-lattice map.  That kernel has rank
$n(n-1)-(n-1)$, and hence the finitely generated torus invariant ring has
\begin{equation}
 \dim\C[m_{\alpha\beta}\mid\alpha\ne\beta]^{\Tf}
 =n(n-1)-(n-1)=(n-1)^2.
 \label{eq:offdiagonal-quotient-dimension}
\end{equation}
Here and below, the dimension of a finitely generated invariant algebra means
Krull dimension; on the generic Hermitian locus it equals the real quotient
dimension.  For a diagonal traceless translation space,
$0\leq\ell\leq n-1$.  Its tangent directions
lie in the diagonal subspace and are independent of the generic torus orbit
directions.  Therefore
\begin{equation}
 \operatorname{trdeg}_{\C}\operatorname{Frac}\cI_{n,k}(\bm p)
 =(n^2-1)-\ell-(n-1)
 =n(n-1)-\ell.
 \label{eq:fixed-fiber-transcendence-degree}
\end{equation}
This proves the dimension statement announced in Sec.~\ref{sec:introduction};
Sec.~III supplies the full tensor-product presentation.

\subsection{Global generators, localizations, and rational coordinates}
\label{subsec:ring-versus-field}

We finish by fixing terminology that will be essential for the $3+1$ system.
A \emph{global polynomial presentation} is an isomorphism
\begin{equation}
 \cI\simeq\C[y_1,\ldots,y_N]/J.
 \label{eq:global-ring-presentation}
\end{equation}
The generators $y_i$ are regular on every stratum and all algebraic
redundancy is recorded by the ideal $J$.  A set
$u_1,\ldots,u_d\in\cI$ is only a \emph{transcendence basis} if it is
algebraically independent and $\operatorname{Frac}(\cI)$ is algebraic over
$\C(u_1,\ldots,u_d)$.  Even when $d=\dim\cI$, a discrete algebraic ambiguity
may remain.  Birational coordinates require the stronger equality
\begin{equation}
 \operatorname{Frac}(\cI)=\C(u_1,\ldots,u_d).
 \label{eq:birational-coordinate-condition}
\end{equation}

Finally, an expression involving division by $s\in\cI$ belongs only to the
localization $\cI_s=\cI[s^{-1}]$ and is valid on the principal open set
$D(s)=\{s\ne0\}$.  For example, the relation anticipated in
Eq.~\eqref{eq:intro-quadrangle-relation} implies
\begin{equation}
 C_{\alpha\beta\gamma\delta}
 =\frac{T_{\alpha\beta\gamma}T_{\alpha\gamma\delta}}
 {Q_{\alpha\gamma}}
 \label{eq:local-quadrangle-reconstruction}
\end{equation}
only after localizing at $Q_{\alpha\gamma}$.  It does not remove the
quadrangle from a global generating set.  Similarly, spectral formulas that
divide by Vandermonde factors live away from the discriminant and cannot by
themselves establish a global Hamiltonian identity.

The universal algebra retains covariance under changes of flavor frame and
is therefore the natural object for off-diagonal spurions.  Fixed diagonal
matter directions are simpler because the flavor torus fixes their
translation space pointwise.  Their quotient can consequently be taken in
either order: first remove the translated diagonal coordinates, then impose
weight balance on the remaining edge variables.  The next section makes this
decomposition explicit and proves the exact multi-parameter diagonal-flow
theorem.

\section{Exact multi-parameter diagonal matter flows}
\label{sec:diagonal-flows}

We now specialize to a fixed tuple of diagonal Hermitian matter directions.
This case includes standard charged-current matter, the relative
active--sterile neutral-current term, and diagonal flavor-nonuniversal
radiative corrections.  It is simpler than a fixed off-diagonal spurion for a
structural reason: the full flavor torus fixes every diagonal matrix.  The
translation and rephasing quotients can therefore be separated exactly.

\subsection{Flow data and adapted diagonal coordinates}
\label{subsec:diagonal-flow-data}

Decompose the complexified traceless Hamiltonian space as
\begin{equation}
 V_n=\mathfrak d_n^0\oplus\mathfrak o_n,
 \qquad
 \mathfrak d_n^0=\{\operatorname{diag}(h_1,\ldots,h_n):
                         \textstyle\sum_\alpha h_\alpha=0\},
 \label{eq:diagonal-offdiagonal-splitting}
\end{equation}
where $\mathfrak o_n$ consists of matrices with zero diagonal.  Let the raw
diagonal charge vectors be $\widehat q_A\in\C^n$.  Their trace-projected
directions are
\begin{equation}
 q_A=\widehat q_A-\frac{\bm 1^{\mathsf T}\widehat q_A}{n}\bm 1,
 \qquad
 p_A=\operatorname{diag}(q_A)\in\mathfrak d_n^0.
 \label{eq:projected-diagonal-charges}
\end{equation}
Define
\begin{equation}
 W=\Span_{\C}\{p_1,\ldots,p_k\}\subset\mathfrak d_n^0,
 \qquad
 \ell=\dim W\leq\min(k,n-1).
 \label{eq:diagonal-flow-space}
\end{equation}
Identity-valued directions and redundant parameter combinations do not
contribute to $\ell$.

Writing $h_\alpha=M_{\alpha\alpha}$ and
$z_{\alpha\beta}=M_{\alpha\beta}$ for $\alpha\ne\beta$, the entrywise solution
of the full affine family is
\begin{equation}
 h_\alpha(\bm a)=h_\alpha(0)+\sum_{A=1}^k a_Aq_{A\alpha},
 \qquad
 z_{\alpha\beta}(\bm a)=z_{\alpha\beta}(0).
 \label{eq:exact-diagonal-entry-flow}
\end{equation}
Choose traceless diagonal matrices $L_j$, $j=1,\ldots,n-1-\ell$, which form
a basis of the trace-orthogonal complement
$W^\perp\subset\mathfrak d_n^0$,
\begin{equation}
 \Tr(L_jp_A)=0\quad\text{for every }A,
 \qquad
 d_j(M)=\Tr(L_jM).
 \label{eq:annihilator-diagonal-invariants}
\end{equation}
Then $\cD_Ad_j=0$.  The particular $L_j$ are noncanonical; their polynomial
algebra is canonically
$\operatorname{Sym}((\mathfrak d_n^0/W)^*)$.

For a real Hermitian tuple, the construction can be made orthogonal.  If the
columns of $V_W$ form an orthonormal basis of $W\subset\bm1^\perp$ and the
columns of $U$ complete them to an orthonormal basis of $\bm1^\perp$, then
\begin{equation}
 d=U^{\mathsf T}h,
 \qquad s=V_W^{\mathsf T}h,
 \qquad
 d(\bm a)=d(0),\qquad
 s(\bm a)=s(0)+V_W^{\mathsf T}Q\bm a,
 \label{eq:adapted-diagonal-flow-coordinates}
\end{equation}
where $Q=(q_1\ \cdots\ q_k)$.  No inverse of $Q$ is required, so this formula
also covers redundant matter parameters.

\subsection{Complete polynomial invariant algebra}
\label{subsec:diagonal-ring-theorem}

Let
\begin{equation}
 \mathcal B_n
 =\C[z_{\alpha\beta}\mid\alpha\ne\beta]^{\Tf}
 \label{eq:general-cycle-algebra}
\end{equation}
denote the off-diagonal rephasing invariant ring.

\begin{theorem}[Complete algebra of diagonal matter flows]
\label{thm:complete-diagonal-flow-algebra}
For an $\ell$-dimensional space
$W\subset\mathfrak d_n^0$ of trace-projected diagonal directions,
\begin{align}
 \cI^{\rm diag}_{n,W}
 &\equiv
 \C[V_n]^{\Tf\times\Ga(W)} \nonumber\\
 &\simeq
 \operatorname{Sym}((\mathfrak d_n^0/W)^*)\otimes_{\C}\mathcal B_n
 \simeq
 \C[d_1,\ldots,d_{n-1-\ell}]\otimes_{\C}\mathcal B_n.
 \label{eq:complete-diagonal-flow-algebra}
\end{align}
Every global polynomial invariant under both the full $W$-translation family
and flavor rephasing is therefore a polynomial in the unaffected diagonal
linear forms and the rephasing-balanced off-diagonal monomials, subject to the
relations of $\mathcal B_n$.
\end{theorem}

\emph{Proof.}
The coordinate algebra splits as
\begin{equation}
 \C[V_n]
 \simeq
 \operatorname{Sym}((\mathfrak d_n^0)^*)
 \otimes_{\C}\C[z_{\alpha\beta}\mid\alpha\ne\beta].
 \label{eq:coordinate-ring-diagonal-splitting}
\end{equation}
Translations by $W$ act only on the first factor, while $\Tf$ acts trivially
on that factor and only on the second.  Proposition
\ref{prop:linear-translation-quotient} gives
\begin{equation}
 \operatorname{Sym}((\mathfrak d_n^0)^*)^{\Ga(W)}
 =\operatorname{Sym}((\mathfrak d_n^0/W)^*).
 \label{eq:diagonal-translation-kernel}
\end{equation}
Taking the torus invariants yields Eq.~\eqref{eq:complete-diagonal-flow-algebra}.
Equivalently, adapted diagonal coordinates turn the $\ell$ independent
derivations into ordinary partial derivatives; their joint kernel contains
exactly the remaining $n-1-\ell$ diagonal coordinates and all off-diagonal
coordinates.  This proves completeness among global polynomials. $\square$

This is a fixed-fiber theorem.  It neither asserts that universal
nonreductive invariant formation commutes with arbitrary specialization nor
classifies extra first integrals of one prescribed curve inside $W$.  It does
imply profile independence for the common algebra: for any differentiable
path $\bm a(x)$,
\begin{equation}
 \frac{d}{dx}I(M(\bm a(x)))
 =\sum_A\dot a_A(x)\cD_AI=0.
 \label{eq:profile-independent-local-invariant}
\end{equation}
Equation~\eqref{eq:profile-independent-local-invariant} concerns the local
Hamiltonian family, not automatically the path-ordered flavor-evolution
operator.

\subsection{Circulation semigroup, Hilbert basis, and normality}
\label{subsec:circulation-semigroup}

Let $E_n=\{(\alpha,\beta):\alpha\ne\beta\}$ be the directed edges of the
complete flavor graph and let $B$ be its incidence matrix, with one column
$e_\alpha-e_\beta$ for the edge $\alpha\to\beta$.  The zero-weight exponent
semigroup is
\begin{equation}
 S_n=\left\{u\in\mathbb N^{E_n}:Bu=0\right\}.
 \label{eq:circulation-semigroup}
\end{equation}
Equation~\eqref{eq:monomial-flow-weight} gives
\begin{equation}
 \mathcal B_n\simeq\C[S_n].
 \label{eq:cycle-semigroup-algebra}
\end{equation}

\begin{proposition}[Cycle Hilbert basis]
\label{prop:cycle-hilbert-basis}
The Hilbert basis of $S_n$ consists exactly of the incidence vectors of
directed simple cycles of lengths $2,\ldots,n$.  Consequently, if
$C=(\alpha_1\alpha_2\cdots\alpha_r)$, the monomials
\begin{equation}
 Z_C=z_{\alpha_1\alpha_2}z_{\alpha_2\alpha_3}\cdots
     z_{\alpha_r\alpha_1}
 \label{eq:simple-cycle-monomial}
\end{equation}
form the minimal monomial generating set of $\mathcal B_n$.
\end{proposition}

\emph{Proof.}
Starting from a positive edge of a nonzero nonnegative circulation, flow
conservation always supplies a positive outgoing edge at the next vertex.
Following such edges eventually repeats a vertex and produces a directed
simple cycle.  Subtracting the minimum edge multiplicity on that cycle
preserves nonnegativity and conservation; iteration decomposes the
circulation into simple cycles.  Conversely, a circulation supported on one
directed simple cycle has equal multiplicity on every edge.  Its unit
incidence vector cannot split into two nonzero circulations. $\square$

The number of length-$r$ Hilbert-basis elements is
\begin{equation}
 N_{n,r}=\binom nr(r-1)!=\frac{n!}{r(n-r)!},
 \qquad
 N_n=\sum_{r=2}^nN_{n,r}.
 \label{eq:number-simple-directed-cycles}
\end{equation}
Thus $N_3=3+2=5$, whereas $N_4=6+8+6=20$.  Introducing one variable $y_C$
for every directed simple cycle gives a surjection
\begin{equation}
 \Phi_n:\C[y_C\mid C\text{ a directed simple cycle}]
 \longrightarrow\mathcal B_n,
 \qquad y_C\longmapsto Z_C.
 \label{eq:cycle-presentation-map}
\end{equation}
Its kernel is a prime binomial toric ideal.  Hence the Hilbert basis is a
global generating set with relations, not a transcendence basis.  In
particular, for $n=4$ the quadrangles cannot be removed globally in favor of
edges and triangles.  Standard background on affine semigroup rings and
toric ideals can be found in Refs.~\cite{Sturmfels1996,CoxLittleSchenck2011}.

The ring is also normal.  The complete directed graph contains the reverse of
every edge, so $S_n$ generates the lattice $\ker_{\mathbb Z}B$.  Explicitly,
if $u=u_+-u_-$ is an integral circulation and $q$ is obtained by reversing
every edge of $u_+$ with the same multiplicity, then both $u_++q$ and
$u_-+q$ lie in $S_n$ and their difference is $u$.  Since
\begin{equation}
 S_n=\ker_{\mathbb Z}B\cap\mathbb N^{E_n},
 \label{eq:saturated-circulation-semigroup}
\end{equation}
it is saturated: if $ru\in S_n$ for an integer $r>0$ and
$u\in\ker_{\mathbb Z}B$, then $ru\geq0$ implies $u\geq0$.  Therefore
$\mathcal B_n$ is a finitely generated normal affine domain, and adjoining the
diagonal variables in Theorem~\ref{thm:complete-diagonal-flow-algebra}
preserves these properties.  Normality does not imply that the quotient is
smooth, factorial, or a complete intersection.

\subsection{Dimension and Hilbert series}
\label{subsec:diagonal-dimension-hilbert}

The incidence matrix $B$ has rank $n-1$.  Because the strictly positive
all-edge circulation lies in $S_n$, its cone is full dimensional in
$\ker_{\R}B$.  It follows that
\begin{equation}
 \dim\mathcal B_n=n(n-1)-(n-1)=(n-1)^2.
 \label{eq:general-cycle-ring-dimension}
\end{equation}
Combining this with Theorem~\ref{thm:complete-diagonal-flow-algebra} gives
\begin{corollary}[Dimension of the diagonal-flow quotient]
\label{cor:diagonal-flow-dimension}
\begin{equation}
 \dim\cI^{\rm diag}_{n,W}
 =(n-1-\ell)+(n-1)^2
 =n(n-1)-\ell.
 \label{eq:exact-diagonal-dimension}
\end{equation}
Here $\dim$ is Krull dimension.  On the generic Hermitian locus it equals the
real quotient dimension $(n^2-1)-\ell-(n-1)$.
\end{corollary}

With every matrix entry assigned degree one, the same weight projection gives
the constant-term Hilbert series, where $\operatorname{CT}$ extracts the
coefficient of $t_1^0\cdots t_{n-1}^0$,
\begin{align}
 H_{\mathcal B_n}(s)
 &=\operatorname{CT}_{t_1,\ldots,t_{n-1}}
   \prod_{\alpha\ne\beta}
   \frac{1}{1-s\,t_\alpha t_\beta^{-1}}\bigg|_{t_n=1},
 \label{eq:general-cycle-hilbert-series}\\
 H_{\cI^{\rm diag}_{n,W}}(s)
 &=\frac{H_{\mathcal B_n}(s)}{(1-s)^{n-1-\ell}}.
 \label{eq:diagonal-flow-hilbert-series}
\end{align}
This constant-term formula, rather than a product over cycle generators,
automatically includes all syzygies.  The explicit $n=4$ rational and
multigraded forms will be computed in Secs.~VI and VII.

For completeness, if a bidirected support graph has $m$ undirected edges and
$c$ connected components, including isolated vertices, its cycle quotient has dimension
$2m-n+c$.  On the Hermitian form this separates into $m$ edge magnitudes and
$m-n+c$ loop phases.  This statement describes the invariant ring after
edges are structurally removed.  A point where some edge values happen to
vanish is instead a stratum of the global quotient; it does not redefine the
global ring.

\subsection{Matter-basis changes and composition dependence}
\label{subsec:matter-basis-composition}

The algebra depends only on $W$.  If
\begin{equation}
 \widetilde p_B=\sum_AR_{BA}p_A,
 \qquad
 \widetilde\cD_B=\sum_AR_{BA}\cD_A,
 \label{eq:matter-basis-change}
\end{equation}
and the two tuples span the same subspace, their joint kernels coincide.  If
$W\subset W'$, then
\begin{equation}
 \cI^{\rm diag}_{n,W'}\subset\cI^{\rm diag}_{n,W}.
 \label{eq:enlarged-flow-subalgebra}
\end{equation}
Each added independent diagonal direction removes one diagonal polynomial
generator and leaves $\mathcal B_n$ unchanged.

A fixed-composition medium with $a_A=\xi_Aa$ follows the single direction
$p_{\rm eff}=\sum_A\xi_Ap_A$ and therefore has $\ell=1$ unless that direction
vanishes.  If the $a_A$ vary independently, $\ell$ is the rank of their full
projected span.  If both $a$ and the composition ratio $\eta$ vary in
$a(p_e+\eta p_s)$, the tangent directions at $a\ne0$ span $p_e$ and $p_s$;
the common polynomial algebra is therefore the $\ell=2$ algebra even though
the rank drops at the exceptional locus $a=0$.  Positivity of physical
densities does not change these polynomial conclusions, because a polynomial
identity valid on an open positive cone extends to its complex span.

\subsection{Spectral pullback and multi-parameter RGEs}
\label{subsec:diagonal-spectral-pullback}

The flavor-basis theorem is regular at spectral degeneracies, but on the
Hermitian real form it can be pulled back to the usual instantaneous spectral
variables.  On a simple-spectrum chart let
\begin{equation}
 M=V\Lambda V^\dagger,
 \qquad
 \Lambda=\operatorname{diag}(\lambda_1,\ldots,\lambda_n),
 \qquad \sum_i\lambda_i=0.
 \label{eq:instantaneous-spectral-decomposition}
\end{equation}
Then every diagonal generator and cycle has the exact representation
\begin{align}
 d_j
 &=\sum_i\lambda_i\sum_\alpha(L_j)_{\alpha\alpha}|V_{\alpha i}|^2,
 \label{eq:spectral-diagonal-invariant}\\
 Z_C
 &=\prod_{a=1}^{r}
   \left(\sum_i\lambda_i
   V_{\alpha_a i}V^*_{\alpha_{a+1}i}\right),
 \qquad \alpha_{r+1}=\alpha_1.
 \label{eq:spectral-cycle-invariant}
\end{align}
These are polynomial Hamiltonian identities and remain valid when eigenvalues
coalesce, even though labeled eigenvectors then cease to be unique.

On the simple-spectrum chart, differentiating
$V^\dagger MV=\Lambda$ gives the exact multi-parameter equations
\begin{align}
 \frac{\partial\lambda_i}{\partial a_A}
 &=\sum_\alpha q_{A\alpha}|V_{\alpha i}|^2,
 \label{eq:multiparameter-eigenvalue-flow}\\
 (\Omega_A)_{ij}
 &=\frac{\sum_\alpha q_{A\alpha}V^*_{\alpha i}V_{\alpha j}}
 {\lambda_j-\lambda_i},
 \qquad i\ne j,
 \qquad
 \Omega_A=V^\dagger\frac{\partial V}{\partial a_A}.
 \label{eq:multiparameter-eigenvector-flow}
\end{align}
The diagonal entries of $\Omega_A$ depend on the eigenvector phase gauge.
The poles in Eq.~\eqref{eq:multiparameter-eigenvector-flow} mark failure of the
labeled spectral chart, not a singularity of the flavor-basis invariant
algebra.  In a genuinely multi-parameter domain one should not assume that all
diagonal Berry connections can be set to zero globally.

\subsection{Three-flavor fibers}
\label{subsec:three-flavor-diagonal-fibers}

For standard three-flavor matter,
\begin{equation}
 p_e=|e\rangle\langle e|-\frac13\bm1
 =\frac13\operatorname{diag}(2,-1,-1),
 \qquad \ell=1,
 \label{eq:three-flavor-electron-direction}
\end{equation}
and the surviving diagonal invariant is
$d_{\mu\tau}=M_{\mu\mu}-M_{\tau\tau}$.  Define
\begin{equation}
 T=z_{e\mu}z_{\mu\tau}z_{\tau e},
 \qquad
 \widetilde T=z_{e\tau}z_{\tau\mu}z_{\mu e},
 \qquad
 Q_{\alpha\beta}=z_{\alpha\beta}z_{\beta\alpha}.
 \label{eq:three-flavor-cycle-generators}
\end{equation}
The theorem gives the global presentation
\begin{equation}
 \cI^{\rm diag}_{3,\langle p_e\rangle}
 \simeq
 \frac{\C[d_{\mu\tau},Q_{e\mu},Q_{e\tau},Q_{\mu\tau},T,\widetilde T]}
 {\left(T\widetilde T-Q_{e\mu}Q_{e\tau}Q_{\mu\tau}\right)}.
 \label{eq:complete-three-flavor-ring}
\end{equation}
It has six global generators, one relation, and Krull dimension five.  On the
Hermitian real form, $\widetilde T=T^*$ and
\begin{equation}
 (\operatorname{Re}T)^2+(\operatorname{Im}T)^2
 =Q_{e\mu}Q_{e\tau}Q_{\mu\tau}.
 \label{eq:three-flavor-real-relation}
\end{equation}
The quotient therefore has dimension five, matching the five independent
matter-invariant observables counted in Ref.~\cite{HarrisonScott2002}; its
global polynomial coordinate ring naturally uses six real generators with one
relation.

For two independent nonuniversal diagonal directions,
$W=\mathfrak d_3^0$ and $\ell=2$.  No diagonal invariant survives,
\begin{equation}
 \cI^{\rm diag}_{3,W}=\mathcal B_3,
 \qquad \dim\cI^{\rm diag}_{3,W}=4.
 \label{eq:three-flavor-two-direction-ring}
\end{equation}

\subsection{Standard $3+1$ matter compositions}
\label{subsec:standard-31-compositions}

At tree level in electrically neutral, unpolarized matter, the electron and
proton neutral-current contributions cancel and the active potentials are
$V_{CC}=\sqrt2G_FN_e$ and $V_{NC}=-G_FN_n/\sqrt2$.  Since a sterile state does
not interact, the unshifted potential is
\begin{equation}
 V_\nu=\operatorname{diag}(V_{CC}+V_{NC},V_{NC},V_{NC},0).
 \label{eq:unshifted-31-potential}
\end{equation}
Subtracting $V_{NC}\bm1$ gives the oscillation-equivalent positive sterile
entry $-V_{NC}$; this entry records the missing sterile neutral-current
interaction, not a sterile force.  In $M=2EH$ define
\begin{equation}
 a_e=2\sqrt2G_FN_eE,
 \qquad
 a_s=\sqrt2G_FN_nE,
 \qquad
 \eta=\frac{a_s}{a_e}=\frac{N_n}{2N_e}\geq0,
 \label{eq:31-matter-parameters}
\end{equation}
with strict positivity when the medium contains neutrons.  With
$\overline E_\alpha=|\alpha\rangle\langle\alpha|-\bm1/4$, the trace-free
matter term is $a_e\overline E_e+a_s\overline E_s$.

For fixed composition, $a_s=\eta a_e$ and there is one flow direction
\begin{equation}
 p_\eta=\overline E_e+\eta\overline E_s,
 \qquad \ell=1.
 \label{eq:fixed-composition-31-direction}
\end{equation}
A convenient annihilator basis is
\begin{align}
 d_{\mu\tau}&=M_{\mu\mu}-M_{\tau\tau},
 \label{eq:31-mu-tau-diagonal-invariant}\\
 d_\eta&=\eta M_{ee}
 +\frac{1-\eta}{2}(M_{\mu\mu}+M_{\tau\tau})-M_{ss}.
 \label{eq:31-eta-diagonal-invariant}
\end{align}
Therefore
\begin{equation}
 \cI^{\rm diag}_{4,\langle p_\eta\rangle}
 \simeq\C[d_{\mu\tau},d_\eta]\otimes\mathcal B_4,
 \qquad \dim\cI^{\rm diag}_{4,\langle p_\eta\rangle}=11.
 \label{eq:fixed-composition-31-ring}
\end{equation}
The Hilbert basis of $\mathcal B_4$ has six two-cycles, eight triangles, and
six quadrangles.  On the Hermitian form these organize into six real edge
moduli, four reverse-orientation triangle pairs (equivalently four complex
triangle generators), and three reverse-orientation quadrangle pairs
(equivalently three complex quadrangle generators).
Together with $d_{\mu\tau}$ and $d_\eta$, this gives 22 real global generators
with syzygies and an eleven-dimensional quotient, in agreement with the
fixed-$\eta$ matter flow of Ref.~\cite{WangZhou2026}.

If the electron and neutron densities vary independently, then
\begin{equation}
 W=\Span\{\overline E_e,\overline E_s\},
 \qquad \ell=2,
 \label{eq:independent-electron-sterile-span}
\end{equation}
and only $d_{\mu\tau}$ survives diagonally:
\begin{equation}
 \cI^{\rm diag}_{4,W}
 \simeq\C[d_{\mu\tau}]\otimes\mathcal B_4,
 \qquad \dim\cI^{\rm diag}_{4,W}=10.
 \label{eq:independent-composition-31-ring}
\end{equation}
The off-diagonal Hilbert basis and all of its syzygies are unchanged.
Antineutrinos reverse both matter coefficients and hence span the same $W$;
their polynomial invariant algebra is identical.

\subsection{Diagonal radiative directions}
\label{subsec:diagonal-radiative-directions}

Let $p_0$ and $r$ be trace-projected diagonal directions.  A correction tied
to the same density coefficient,
\begin{equation}
 M(a)=M_0+a(p_0+\epsilon r),
 \label{eq:fixed-radiative-composition}
\end{equation}
still has $W_\epsilon=\C(p_0+\epsilon r)$ and $\ell=1$.  Its $n-2$
diagonal invariants rotate to annihilate $p_0+\epsilon r$, but their number is
unchanged and every off-diagonal cycle remains exact.  For the illustrative
raw $3+1$ direction $(1,0,\epsilon,\eta)$ in the order
$(e,\mu,\tau,s)$, set
\begin{equation}
 x_e=M_{ee}-M_{\mu\mu},\quad
 x_\tau=M_{\tau\tau}-M_{\mu\mu},\quad
 x_s=M_{ss}-M_{\mu\mu}.
 \label{eq:radiative-difference-coordinates}
\end{equation}
Two exact diagonal invariants are
\begin{equation}
 d_\tau^{(\epsilon)}=x_\tau-\epsilon x_e,
 \qquad
 d_s^{(\eta)}=x_s-\eta x_e.
 \label{eq:fixed-radiative-diagonal-invariants}
\end{equation}

If the radiative coefficient is independently variable,
\begin{equation}
 M(a,b)=M_0+ap_0+br,
 \label{eq:independent-radiative-flow}
\end{equation}
then $W=\Span\{p_0,r\}$ and generically $\ell=2$; one additional diagonal
generator is lost.  More precisely, for $\epsilon_1\ne\epsilon_2$ and
independent $p_0,r$,
\begin{equation}
 \cI^{\rm diag}_{n,W_{\epsilon_1}}
 \cap_{\C[V_n]}\cI^{\rm diag}_{n,W_{\epsilon_2}}
 =\cI^{\rm diag}_{n,\Span\{p_0,r\}}.
 \label{eq:two-composition-intersection}
\end{equation}
If $\overline E_e$, $\overline E_s$, and a third direction spanning the
remaining diagonal axis of a $3+1$ system all vary independently, then
$\ell=3$, no diagonal generator survives, and the full invariant algebra is
$\mathcal B_4$ of dimension nine.

The flow-space rank $\ell$ and the effective matrix rank
$r_{\rm eff}$ of Sec.~IV are logically different.  A fixed radiative
correction can raise $r_{\rm eff}$ while leaving $\ell=1$.  Conversely, a
new coefficient lowers the invariant dimension only if its trace-projected
direction is linearly independent of the existing flow space.  Corrections
tied to several existing density coefficients may rotate that space without
increasing its dimension.  Section~VIII develops the corresponding controlled
breaking problem for off-diagonal NSI.

With the complete scalar algebra for diagonal translations established, we
next turn to common low-rank support, where exterior-power Krylov covariants
provide additional distinguished representatives.

\section{Rank-$r$ Krylov and exterior-power invariants}
\label{sec:krylov-exterior}

The diagonal-flow theorem classifies the full scalar invariant algebra when
the matter directions are diagonal in the flavor frame.  A different
structure becomes visible when one or several matter directions share a
low-dimensional support after an unobservable scalar shift.  On such a fixed
spurion fiber, block-Krylov rows evolve only by combinations of earlier rows.
Their exterior product is therefore exactly matter independent.  The result
does not replace the invariant-ring description of Sec.~III and does not claim
that Krylov minors generate that ring.  Rather, it identifies distinguished
determinantal representatives whose spectral form will be analyzed in
Sec.~V.

\subsection{Effective rank and common effective support}
\label{subsec:effective-rank-common-support}

For a Hermitian matter direction $P$, oscillations see only its scalar-shift
class.  Its effective rank is
\begin{equation}
 r_{\rm eff}(P)=\min_{c\in\R}\rank(P-c\bm1)
 =n-\max_{\lambda\in\operatorname{spec}P}
       \dim\ker(P-\lambda\bm1).
 \label{eq:effective-rank-spectral-multiplicity}
\end{equation}
The second equality follows by diagonalizing $P$: a minimizing $c$ is an
eigenvalue of largest multiplicity.  The minimizer and its complementary
support need not be unique if several eigenvalues have the same largest
multiplicity.  The case $r_{\rm eff}=0$ is a flavor-universal direction and is
discarded below.

For a tuple $\bm P=(P_1,\ldots,P_k)$, the relevant joint quantity is instead
\begin{equation}
 r_{\rm com}(\bm P)
 =\min_{\bm c\in\R^k}
 \dim\left(\sum_{A=1}^k\operatorname{im}
                    (P_A-c_A\bm1)\right).
 \label{eq:common-effective-support-rank}
\end{equation}
Choose a minimizing tuple $\bm c$ and an isometry
$S:\C^r\rightarrow\C^n$, $S^\dagger S=\bm1_r$, whose image is the sum in
Eq.~\eqref{eq:common-effective-support-rank}.  Hermiticity implies that each
$P_A-c_A\bm1$ annihilates the orthogonal complement of this image.  Hence
\begin{equation}
 P_A=c_A\bm1+S D_A S^\dagger,
 \qquad
 D_A=S^\dagger(P_A-c_A\bm1)S=D_A^\dagger.
 \label{eq:common-support-decomposition}
\end{equation}
The matrices $D_A$ need not be invertible and need not commute.  Only their
joint support is required to be minimal.  Minimality is useful for counting,
but the theorem below remains true for any representation of the form
\eqref{eq:common-support-decomposition}.  The rank $r_{\rm com}$ is unchanged
by an invertible change of matter-parameter basis and is logically distinct
from the translation-space dimension $\ell$ of Sec.~III.  For every $r$, the
$q=1$ covariant contains only the fixed support frame and is independent of
$M$.  If $r=n$, the condition $qr\leq n$ below permits no larger $q$; the
nontrivial hierarchy therefore requires $r<n$ and $q\geq2$.

We work on the fixed-spurion fiber
\begin{equation}
 M(\bm a)=M_0+
 \sum_{A=1}^k a_A(c_A\bm1+S D_A S^\dagger).
 \label{eq:common-support-affine-family}
\end{equation}
When the $P_A$ were first trace projected, the scalar pieces in
Eq.~\eqref{eq:common-support-affine-family} are precisely compensated by the
traces of $S D_A S^\dagger$; their temporary separation merely exposes a
low-rank representative of the same oscillation class.

\subsection{Exact block-Krylov row evolution}
\label{subsec:block-krylov-row-evolution}

For $p\geq0$ define the $r\times n$ block rows and the $r\times r$ moments
\begin{equation}
 R_p=S^\dagger M^p,
 \qquad
 G_p=S^\dagger M^pS,
 \qquad
 K_q=\begin{pmatrix}R_0\\R_1\\ \vdots\\R_{q-1}\end{pmatrix}
 \in\C^{qr\times n}.
 \label{eq:section4-block-krylov-definitions}
\end{equation}
Rows are ordered first by $p=0,\ldots,q-1$ and then by the chosen support
frame.  We write $\partial_A=\partial/\partial a_A$.

\begin{theorem}[Common-support Krylov invariance]
\label{thm:common-support-block-krylov}
For the family \eqref{eq:common-support-affine-family},
\begin{align}
 \partial_A R_0&=0,
 \label{eq:krylov-R0-flow}\\
 \partial_A R_p
 &=p c_A R_{p-1}
 +\sum_{j=0}^{p-1}G_jD_A R_{p-1-j},
 \qquad p\geq1.
 \label{eq:exact-block-row-flow}
\end{align}
Equivalently,
\begin{equation}
 \partial_AK_q=L_{A,q}K_q,
 \label{eq:block-krylov-matrix-flow}
\end{equation}
where, for block indices $p,\ell=0,\ldots,q-1$,
\begin{equation}
 (L_{A,q})_{p\ell}=
 \begin{cases}
  G_{p-1-\ell}D_A
  +p c_A\bm1_r\,\delta_{\ell,p-1},&\ell<p,\\
  0,&\ell\geq p.
 \end{cases}
 \label{eq:block-lower-krylov-generator}
\end{equation}
Thus $L_{A,q}$ is strictly block-lower triangular.

For every $q\geq1$, $\rank K_q$ is constant along each connected matter
orbit.  If $m=qr\leq n$, the exterior covariant
\begin{equation}
 \Omega_q(M;S)
 =\bigwedge_{p=0}^{q-1}\bigwedge_{\rho=1}^{r}
       (R_p)_{\rho,*}
 \in\bigwedge^m(\C^n)^*
 \label{eq:krylov-exterior-covariant}
\end{equation}
obeys
\begin{equation}
 \partial_A\Omega_q=0
 \qquad\text{for every }A.
 \label{eq:krylov-exterior-invariance}
\end{equation}
Equivalently, all of its Pl\"ucker coordinates
\begin{equation}
 \kappa_I^{(q)}(M;S)
 =\det K_q[:,I],
 \qquad I\subset\{1,\ldots,n\},\quad |I|=m,
 \label{eq:krylov-maximal-minors}
\end{equation}
are joint polynomial first integrals.  They are homogeneous in $M$ of degree
\begin{equation}
 \deg_M\kappa_I^{(q)}
 =r\sum_{p=0}^{q-1}p=\frac{rq(q-1)}2.
 \label{eq:krylov-minor-degree}
\end{equation}
When $qr=n$, Eq.~\eqref{eq:krylov-maximal-minors} consists of the single
square determinant $\det K_q$.
\end{theorem}

\emph{Proof.}
For $p\geq1$, differentiation of a matrix power gives
\begin{equation}
 \partial_AM^p
 =\sum_{j=0}^{p-1}M^jP_AM^{p-1-j}.
 \label{eq:matrix-power-matter-derivative}
\end{equation}
Inserting Eq.~\eqref{eq:common-support-decomposition} and multiplying from the
left by $S^\dagger$ yields $p$ identical scalar terms and the remaining
common-support terms,
\begin{equation}
 S^\dagger M^j(SD_AS^\dagger)M^{p-1-j}
 =G_jD_AR_{p-1-j},
 \label{eq:common-support-row-reduction}
\end{equation}
which proves Eqs.~\eqref{eq:krylov-R0-flow} and
\eqref{eq:exact-block-row-flow}.  Relabeling
$\ell=p-1-j$ gives Eq.~\eqref{eq:block-lower-krylov-generator}.

Along any path in matter-parameter space, Eq.~\eqref{eq:block-krylov-matrix-flow}
has the form
\begin{equation}
 K_q(\bm a)=U_q(\bm a)K_q(\bm0),
 \qquad \det U_q=1,
 \label{eq:finite-unipotent-krylov-flow}
\end{equation}
where $U_q$ is block-lower unitriangular.  This proves constancy of the row
rank.  Alternatively, applying the derivative to
Eq.~\eqref{eq:krylov-exterior-covariant} gives
\begin{equation}
 \partial_A\Omega_q=(\Tr L_{A,q})\Omega_q=0,
 \label{eq:exterior-trace-proof}
\end{equation}
because a replacement by an earlier row repeats a factor in the wedge.  This
argument does not divide by a minor and therefore remains valid on the
rank-deficient locus.  Expansion in the standard exterior basis gives
Eq.~\eqref{eq:krylov-maximal-minors}, and the degree follows from taking all
$r$ rows at every power $p$.  \hfill$\square$

When $qr<n$ and $\Omega_q\ne0$, the projective coordinates
$[\kappa_I^{(q)}]$ define a matter-constant point of the Grassmannian
$\operatorname{Gr}(qr,n)$ and consequently obey its quadratic Pl\"ucker
relations.  These relations express covariance of one exterior object; they
should not be confused with the toric cycle syzygies analyzed in Sec.~VII.

The affine matter derivatives commute even when the $D_A$ do not.  In terms
of the lower-triangular matrices this compatibility reads
\begin{equation}
 \left(\partial_AL_{B,q}-\partial_BL_{A,q}
 +L_{B,q}L_{A,q}-L_{A,q}L_{B,q}\right)K_q=0.
 \label{eq:krylov-connection-compatibility}
\end{equation}
The coefficient in parentheses itself vanishes on the full-row-rank locus,
where $K_q$ has a right inverse.  No unconditional matrix zero-curvature
claim is needed on rank-deficient strata.

The matrix $K_q$ is the block observability matrix of the pair
$(M,S^\dagger)$, in the framework introduced by Kalman~\cite{Kalman1960}.
The deformation $M\mapsto M+B S^\dagger$ is an output injection, whose
preservation of observability is classical~\cite{KarcaniasLimantsevaHalikias2021};
block rank and grade are standard Krylov notions~\cite{GutknechtSchmelzer2009}.
Here the additional scalar shift and simultaneous matter directions lead to
the exact exterior covariants on a neutrino Hamiltonian fiber.

\subsection{Trace shifts, covariance, and scalar invariants}
\label{subsec:krylov-covariance-scalars}

The construction is independent of the Hamiltonian representative in the
trace quotient.  For any scalar $b$,
\begin{equation}
 R_p(M+b\bm1)
 =\sum_{\ell=0}^{p}\binom{p}{\ell}b^{p-\ell}R_\ell(M).
 \label{eq:krylov-scalar-shift-binomial}
\end{equation}
Thus $K_q(M+b\bm1)$ is obtained from $K_q(M)$ by a block-Pascal matrix that is
unit lower triangular.  Every exterior coordinate
\eqref{eq:krylov-maximal-minors} is unchanged.

The support frame is not canonical.  For $U\in U(r)$,
\begin{equation}
 S\longmapsto SU,
 \qquad D_A\longmapsto U^\dagger D_AU,
 \qquad
 K_q\longmapsto(\bm1_q\otimes U^\dagger)K_q,
 \qquad
 \Omega_q\longmapsto(\det U^\dagger)^q\Omega_q.
 \label{eq:krylov-support-frame-covariance}
\end{equation}
Hence $\Omega_q$ is a determinant-line covariant and, when $\Omega_q\ne0$,
its projective Pl\"ucker point is support-frame independent.  Under a
simultaneous unitary flavor
change
\begin{equation}
 M\longmapsto FMF^\dagger,
 \qquad S\longmapsto FS,
 \qquad K_q\longmapsto K_qF^\dagger.
 \label{eq:krylov-flavor-frame-covariance}
\end{equation}
For diagonal rephasing, a component $\kappa_I^{(q)}$ therefore carries a
definite character.  More explicitly, if the support is the coordinate subset
$A_0$ and its canonical frame is held fixed, then for
$T=\operatorname{diag}(e^{\ii\phi_1},\ldots,e^{\ii\phi_n})$,
\begin{equation}
 \kappa_I^{(q)}\longmapsto
 \exp\!\left[\ii\left(q\sum_{\alpha\in A_0}\phi_\alpha
               -\sum_{\beta\in I}\phi_\beta\right)\right]
 \kappa_I^{(q)}.
 \label{eq:krylov-coordinate-rephasing-weight}
\end{equation}
In either convention an individual minor is generally a relative invariant,
not a scalar.

On the Hermitian real form, Cauchy--Binet supplies the canonical Hermitian norm
of the chosen exterior covariant,
\begin{align}
 \mathcal N_q(M;S)
 &\equiv\|\Omega_q(M;S)\|^2
 =\det(K_qK_q^\dagger) \nonumber\\
 &=\sum_{|I|=qr}|\kappa_I^{(q)}|^2
 =\det\!\left[G_{p+\ell}\right]_{p,\ell=0}^{q-1}.
 \label{eq:krylov-cauchy-binet-norm}
\end{align}
It is matter independent, invariant under unitary support-frame changes, and
invariant under simultaneous unitary flavor-frame changes.  Individual norm
squares $|\kappa_I^{(q)}|^2$ are already flavor-rephasing scalars; mixed
products require weight balance.  Nontriviality is separate from invariance:
if $\rank K_q<qr$, then $\Omega_q$ and all the displayed minors vanish.  The
rank-deficiency locus is nevertheless a union of complete matter orbits.

For a support extracted from a fixed numerical spurion, these are polynomials
in the entries of $M$ with fixed-spurion coefficients.  They are not
automatically universal polynomials in the entries of the $P_A$.  This is the
fixed-fiber distinction of Sec.~II in a determinantal form.

\subsection{Rank one and the three-flavor determinant}
\label{subsec:rank-one-krylov}

Let $r=1$, $q=n$, and write the normalized support vector as $s$.  On a
simple-spectrum chart,
\begin{equation}
 M=V\Lambda V^\dagger,
 \qquad
 x_i=(s^\dagger V)_i.
 \label{eq:rank-one-krylov-spectral-data}
\end{equation}
Then
\begin{equation}
 K_nV=
 \begin{pmatrix}
  x_1&\cdots&x_n\\
  \lambda_1x_1&\cdots&\lambda_nx_n\\
  \vdots&&\vdots\\
  \lambda_1^{n-1}x_1&\cdots&\lambda_n^{n-1}x_n
 \end{pmatrix},
 \label{eq:rank-one-krylov-vandermonde-matrix}
\end{equation}
and hence
\begin{align}
 (\det K_n)(\det V)
 &=\left(\prod_{i=1}^n x_i\right)
   \prod_{i<j}(\lambda_j-\lambda_i),
 \label{eq:rank-one-krylov-determinant}\\
 \mathcal N_n
 &=\left(\prod_{i=1}^n|x_i|^2\right)
   \prod_{i<j}(\lambda_j-\lambda_i)^2.
 \label{eq:rank-one-krylov-norm}
\end{align}
The first line is frame covariant; the second is a physical scalar.  The
ordinary Vandermonde factorization is the exceptional rank-one case of the
general Pl\"ucker expansion developed in Sec.~V.

For standard three-flavor matter take $s=|e\rangle$.  In the flavor order
$(e,\mu,\tau)$,
\begin{align}
 \det K_3
 &=M_{e\mu}(M^2)_{e\tau}-M_{e\tau}(M^2)_{e\mu}
 \nonumber\\
 &=M_{e\mu}^2M_{\mu\tau}-M_{e\tau}^2M_{\tau\mu}
 -d_{\mu\tau}M_{e\mu}M_{e\tau},
 \label{eq:three-flavor-krylov-flavor-form}\\
 |\det K_3|^2
 &=\prod_{i=1}^{3}|V_{ei}|^2
   \prod_{i<j}(\lambda_j-\lambda_i)^2.
 \label{eq:three-flavor-krylov-kty-norm}
\end{align}
This is the electron-row Vandermonde construction associated with the KTY
relations and their matter-invariant denominator
\cite{KimuraTakamuraYokomakura2002,WangZhou2026}.  The determinant itself is
phase covariant; its norm is the scalar invariant.

\subsection{The rank-two $3+1$ determinant}
\label{subsec:rank-two-31-krylov}

For the trace-free standard direction
$p_\eta=\overline E_e+\eta\overline E_s$, choose
\begin{equation}
 c_\eta=-\frac{1+\eta}{4},
 \qquad
 S=(|e\rangle,|s\rangle),
 \qquad
 D_\eta=\operatorname{diag}(1,\eta).
 \label{eq:31-krylov-support-data}
\end{equation}
For $\eta\ne0$ this has effective rank two.  With support-row order $(e,s)$,
block order $(R_0,R_1)$, and flavor-column order $(e,\mu,\tau,s)$,
\begin{equation}
 K_2=
 \begin{pmatrix}
  1&0&0&0\\
  0&0&0&1\\
  M_{ee}&M_{e\mu}&M_{e\tau}&M_{es}\\
  M_{se}&M_{s\mu}&M_{s\tau}&M_{ss}
 \end{pmatrix},
 \label{eq:31-square-krylov-matrix}
\end{equation}
so the ordering fixes the sign to be
\begin{equation}
 \mathcal D_{es}\equiv\det K_2
 =M_{e\mu}M_{s\tau}-M_{e\tau}M_{s\mu}.
 \label{eq:31-rank-two-krylov-determinant}
\end{equation}
On the Hermitian form define
\begin{equation}
 C_{e\mu s\tau}
 =M_{e\mu}M_{\mu s}M_{s\tau}M_{\tau e}.
 \label{eq:31-krylov-quadrangle}
\end{equation}
Then
\begin{equation}
 |\mathcal D_{es}|^2
 =Q_{e\mu}Q_{s\tau}+Q_{e\tau}Q_{s\mu}
 -2\operatorname{Re}C_{e\mu s\tau}.
 \label{eq:31-krylov-cycle-norm}
\end{equation}
This formula embeds the Krylov norm directly in the cycle algebra of
Sec.~III.  It is a distinguished polynomial combination, not a new
algebraically independent observable and not a multiplicative monomial of the
restricted spectral type excluded in Ref.~\cite{WangZhou2026}.

The same determinant is killed separately by the electron and sterile matter
derivatives.  Indeed, for independent $\overline E_e$ and $\overline E_s$ one
uses the same $S$ with singular support matrices
\begin{equation}
 D_e=\operatorname{diag}(1,0),
 \qquad
 D_s=\operatorname{diag}(0,1),
 \label{eq:independent-es-support-matrices}
\end{equation}
demonstrating that individual $D_A$ need not be invertible.  The proof of
Theorem~\ref{thm:common-support-block-krylov} separately shows that no
commutativity assumption is needed.

At $\eta=0$, the minimal support drops to rank one and additional rank-one
Krylov representatives become available; Eq.~\eqref{eq:31-krylov-cycle-norm}
remains invariant because a nonminimal common support is still allowed.  At
$\eta=1$, both the electron--sterile and the complementary $\mu$--$\tau$
supports minimize the effective rank.  Their square determinants are complex
conjugates up to the ordering sign and therefore have the same norm.

\subsection{Rank-deficient and noncommon-support boundaries}
\label{subsec:krylov-boundaries}

The theorem is polynomial and remains regular at spectral degeneracies.  For
a square block alternant, an eigenspace of multiplicity $d\leq r$ can have
linearly independent projected columns; pairwise degeneracy therefore does not
generally force a rank-$r>1$ determinant to vanish.  If their projected rank is
smaller than $d$, the square determinant vanishes; in particular, this is
forced when $d>r$.  For a nonsquare exterior covariant, this statement applies
in the mass-eigenstate column basis: block-alternant minors selecting more than
$r$ columns from that eigenspace vanish, whereas the flavor-coordinate minors
$\kappa_I^{(q)}$ are Cauchy--Binet sums and need not vanish individually.  An
$M$-invariant common support also gives
$\rank K_q\leq r$, so $\Omega_q=0$ for every $q>1$.  If $qr>n$, the exterior
power in Eq.~\eqref{eq:krylov-exterior-covariant} vanishes identically, though
the row-rank statement remains meaningful.  None of these cases is a failure
of matter invariance.

Common support must remain fixed along the flow.  If a matter direction has
an additional leakage term $E_A$ outside the factorization
\eqref{eq:common-support-decomposition}, Eq.~\eqref{eq:exact-block-row-flow}
acquires
\begin{equation}
 \sum_{j=0}^{p-1}S^\dagger M^jE_AM^{p-1-j},
 \label{eq:krylov-off-support-leakage}
\end{equation}
which need not lie in the span of earlier Krylov rows.  This is the precise
entry point for the controlled NSI breaking analysis of Sec.~VIII.  Likewise,
restricting a multiparameter family to a special path can lower
$r_{\rm com}$ and create path-specific covariants that are not common to the
full family.

Finally, on a spectral chart let
\begin{equation}
 M=V\Lambda V^\dagger,
 \qquad X=S^\dagger V.
 \label{eq:krylov-spectral-transition-data}
\end{equation}
Then
\begin{equation}
 K_q=
 \underbrace{\begin{pmatrix}
 X\\X\Lambda\\ \vdots\\X\Lambda^{q-1}
 \end{pmatrix}}_{\displaystyle\mathsf A_q(X,\Lambda)}V^\dagger.
 \label{eq:block-alternant-factorization}
\end{equation}
When $qr=n$, $\det K_q=(\det\mathsf A_q)(\det V^\dagger)$.  The next section
derives the general Pl\"ucker expansion and exact higher-rank no-separation
criteria in specified coefficient rings, followed by a complete factorization
classification for the physically central $4\times4$, rank-two case.

\section{Pl\"ucker expansion and factorization classification}
\label{sec:plucker-factorization}

Section~\ref{sec:krylov-exterior} established matter invariance in the flavor
frame, without introducing eigenvalues or eigenvectors.  We now analyze the
same exterior covariants on a simple-spectrum chart.  This change of
coordinates is used only to expose their factorization structure: every final
identity remains the pullback of a regular flavor-basis polynomial.  We first
give the Pl\"ucker expansion for all $qr\leq n$, then distinguish three
logically different statements: the exact rank-one Vandermonde factorization,
the absence of any generic higher-rank analogue of that particular product,
and the complete reducibility classification of the $4\times4$, rank-two
quadratic.

\subsection{Mass-basis block alternants}
\label{subsec:mass-basis-block-alternants}

Let
\begin{equation}
 M=V\Lambda V^\dagger,
 \qquad
 \Lambda=\operatorname{diag}(\lambda_1,\ldots,\lambda_n),
 \qquad
 X=S^\dagger V\in\C^{r\times n}.
 \label{eq:section5-spectral-data}
\end{equation}
For $m=qr\leq n$ define
\begin{equation}
 \mathsf A_q(X,\Lambda)
 =\begin{pmatrix}
   X\\X\Lambda\\ \vdots\\X\Lambda^{q-1}
  \end{pmatrix}
 \in\C^{m\times n}.
 \label{eq:section5-block-alternant}
\end{equation}
Equation~\eqref{eq:block-alternant-factorization} reads simply
$K_q=\mathsf A_qV^\dagger$.

For every $r$-element subset $J\subset[n]\equiv\{1,\ldots,n\}$, with its
elements in increasing order, write
\begin{equation}
 p_J(X)=\det X[:,J]
 \label{eq:plucker-coordinate-X}
\end{equation}
for the corresponding Pl\"ucker coordinate.  These homogeneous coordinates
satisfy the standard quadratic relations defining
$\operatorname{Gr}(r,n)$~\cite{Harris1992}.  For a subset $J$ and integer
$a\geq0$, also define
\begin{equation}
 \lambda_J^{[a]}=\prod_{j\in J}\lambda_j^a,
 \qquad \lambda_J^{[0]}=1.
 \label{eq:spectral-subset-monomial}
\end{equation}

Let $\mathfrak P_{r,q}(J)$ be the set of ordered partitions
$\bm J=(J_0,\ldots,J_{q-1})$ of a set $J$ into disjoint $r$-element subsets.
If every $J_a$ is internally increasing, let $\epsilon(\bm J)$ be the sign of
the shuffle that changes the increasing list of $J$ into the concatenated list
$(J_0,J_1,\ldots,J_{q-1})$.

\begin{theorem}[Krylov--Pl\"ucker expansion]
\label{thm:general-krylov-plucker-expansion}
For every $qr=m\leq n$ and every mass-index subset $J\subset[n]$ of size $m$,
the maximal minor of the block alternant is
\begin{equation}
 \Phi_J^{(q)}(X,\bm\lambda)
 \equiv\det\mathsf A_q[:,J]
 =\sum_{\bm J\in\mathfrak P_{r,q}(J)}
   \epsilon(\bm J)
   \prod_{a=0}^{q-1}
       p_{J_a}(X)\lambda_{J_a}^{[a]}.
 \label{eq:general-plucker-expansion}
\end{equation}
The flavor-frame Krylov coordinates of
Eq.~\eqref{eq:krylov-maximal-minors} are consequently
\begin{equation}
 \kappa_I^{(q)}
 =\sum_{\substack{J\subset[n]\\|J|=m}}
   \Phi_J^{(q)}(X,\bm\lambda)
   \det V^\dagger[J,I],
 \qquad |I|=m.
 \label{eq:flavor-minor-plucker-cauchy-binet}
\end{equation}
If $m=n$, there is one subset $J=[n]$ and
\begin{equation}
 \det K_q
 =F_{r,q}(p,\bm\lambda)\det V^\dagger,
 \qquad
 F_{r,q}\equiv\Phi_{[n]}^{(q)}=\det\mathsf A_q.
 \label{eq:square-plucker-alternant}
\end{equation}
\end{theorem}

\emph{Proof.}
Apply generalized Laplace expansion to the $q$ consecutive row blocks of
$\mathsf A_q[:,J]$.  Assigning the columns $J_a$ to row block $a$ produces
the shuffle sign $\epsilon(\bm J)$ and the block determinant
\begin{equation}
 \det\!\left(X[:,J_a]\operatorname{diag}
              (\lambda_j^a)_{j\in J_a}\right)
 =p_{J_a}(X)\lambda_{J_a}^{[a]}.
 \label{eq:block-alternant-laplace-factor}
\end{equation}
Multiplication over the row blocks and summation over all ordered partitions
gives Eq.~\eqref{eq:general-plucker-expansion}.  Applying Cauchy--Binet to the
minor of $K_q=\mathsf A_qV^\dagger$ gives
Eq.~\eqref{eq:flavor-minor-plucker-cauchy-binet}; the square case follows
immediately.  \hfill$\square$

For $m<n$, Eq.~\eqref{eq:flavor-minor-plucker-cauchy-binet} is essential:
the mass-basis minors $\Phi_J^{(q)}$ and the flavor-basis minors
$\kappa_I^{(q)}$ are related by a compound-unitary transformation, not by a
one-to-one relabeling.  This is why a degeneracy may force selected
$\Phi_J^{(q)}$ to vanish without forcing every flavor-coordinate minor to
vanish.

\subsection{Gauge, permutation, and affine spectral covariance}
\label{subsec:plucker-spectral-covariance}

Every term in $F_{r,q}$ contains each mass-eigenstate column exactly once.
Consequently an eigenvector phase change
\begin{equation}
 V_{*i}\longmapsto e^{\ii\theta_i}V_{*i}
 \label{eq:eigenvector-column-phase}
\end{equation}
multiplies $F_{r,q}$ by $e^{\ii\sum_i\theta_i}$ and multiplies
$\det V^\dagger$ by the inverse phase.  Their product in
Eq.~\eqref{eq:square-plucker-alternant} is gauge independent.  A simultaneous
permutation of the eigenvalues and eigenvector columns changes
$F_{r,q}$ by the sign of the column permutation and changes
$\det V^\dagger$ by the same inverse sign.  Thus the flavor determinant is
mass-label invariant.

The block-Pascal transformation of Eq.~\eqref{eq:krylov-scalar-shift-binomial}
and a row rescaling give
\begin{align}
 F_{r,q}(p,\bm\lambda+b\bm1)
 &=F_{r,q}(p,\bm\lambda),
 \label{eq:plucker-spectral-translation}\\
F_{r,q}(p,t\bm\lambda)
 &=t^{d_{r,q}}F_{r,q}(p,\bm\lambda),
 \qquad
 d_{r,q}=\frac{rq(q-1)}2.
 \label{eq:plucker-spectral-homogeneity}
\end{align}
For a support-frame transformation $G\in\operatorname{GL}(r,\C)$,
\begin{equation}
 F_{r,q}(GX,\bm\lambda)=(\det G)^qF_{r,q}(X,\bm\lambda).
 \label{eq:plucker-support-frame-covariance}
\end{equation}
Thus the alternant is homogeneous of Pl\"ucker degree $q$.  These properties
make clear that scalar eigenvalue shifts, support frames, eigenvector phases,
and mass labels are coordinate gauges rather than additional factorization
assumptions.

\subsection{Rank-one Vandermonde factorization}
\label{subsec:rank-one-vandermonde-factorization}

For $r=1$ and $q=n$, the Pl\"ucker coordinates are simply the entries
$p_{\{i\}}=x_i$ of the row $X$.  Ordered block partitions are permutations of
$[n]$, so Eq.~\eqref{eq:general-plucker-expansion} becomes the ordinary
alternant determinant.

\begin{corollary}[Rank-one factorization]
\label{cor:rank-one-vandermonde-factorization}
For $r=1$,
\begin{equation}
 F_{1,n}(x,\bm\lambda)
 =\left(\prod_{i=1}^n x_i\right)
  \Delta_n(\bm\lambda),
 \qquad
 \Delta_n(\bm\lambda)=\prod_{i<j}(\lambda_j-\lambda_i).
 \label{eq:rank-one-plucker-vandermonde}
\end{equation}
\end{corollary}

Here equality of any two eigenvalues makes two columns proportional, which is
exactly the divisibility mechanism producing the full Vandermonde.  This
mechanism is absent at generic higher rank: if $r>1$ and
$\lambda_i=\lambda_j$, the corresponding block columns are
\begin{equation}
 \begin{pmatrix}x_i\\\lambda_i x_i\\ \vdots\\
 \lambda_i^{q-1}x_i\end{pmatrix},
 \qquad
 \begin{pmatrix}x_j\\\lambda_i x_j\\ \vdots\\
 \lambda_i^{q-1}x_j\end{pmatrix},
 \label{eq:equal-eigenvalue-block-columns}
\end{equation}
and remain independent whenever the projected vectors $x_i$ and $x_j$ do.

\subsection{Universal divisors and higher-rank no-separation}
\label{subsec:generic-higher-rank-nonfactorization}

Factorization statements require a precise coefficient algebra.  Let
\begin{equation}
 \mathcal R_{r,n}^{\rm Pl}
 =\frac{\C[p_J\mid J\subset[n],\ |J|=r]}
        {I_{r,n}^{\rm Pl}}
 \label{eq:plucker-coordinate-ring}
\end{equation}
be the coordinate ring of the affine Pl\"ucker cone, and let
\begin{equation}
 \mathcal R_{r,n}^{\circ}
 =\mathcal R_{r,n}^{\rm Pl}
  \left[\left(\prod_{|J|=r}p_J\right)^{-1}\right],
 \qquad
 \mathcal K_{r,n}=\operatorname{Frac}(\mathcal R_{r,n}^{\rm Pl}).
 \label{eq:big-plucker-torus-ring}
\end{equation}
The first ring retains every Pl\"ucker boundary; the second is the
uniform-matroid open set.  Because $F_{r,q}$ has Pl\"ucker degree $q$, it is a
relative covariant on the affine cone, or equivalently a section of
$\mathcal O(q)$ on the projective Grassmannian.  Thus treating it as a
projective rational function requires a chart or a trivialization.  The rings
in Eqs.~\eqref{eq:plucker-coordinate-ring} and
\eqref{eq:big-plucker-torus-ring} make that choice explicit.

Set $\delta_i=\lambda_i-\lambda_n$ for $i<n$.  We call an identity
\begin{equation}
 F_{r,q}=G(p)H(\bm\delta),
 \qquad
 G\in\mathcal K_{r,n}^{\times},
 \quad
 H\in\C[\delta_1,\ldots,\delta_{n-1}]\setminus\C,
 \label{eq:separated-plucker-spectral-factorization}
\end{equation}
a \emph{separated Pl\"ucker--spectral factorization}.  It allows an arbitrary
Pl\"ucker rational factor, and is therefore stronger than testing only a
single Laurent monomial.

\begin{theorem}[Higher-rank factorization obstructions]
\label{thm:no-generic-higher-rank-vandermonde}
Let $n=qr$, with $q\geq2$ and $r\geq2$.
\begin{enumerate}
 \item The square alternant $F_{r,q}$ has no separated
 Pl\"ucker--spectral factorization of the form
 \eqref{eq:separated-plucker-spectral-factorization}.
 \item For every $i<j$,
 \begin{equation}
  \lambda_j-\lambda_i\nmid F_{r,q}
  \quad\text{in}\quad
  \mathcal R_{r,n}^{\circ}[\bm\lambda].
  \label{eq:no-pairwise-difference-divisor}
 \end{equation}
 \item No nonconstant $h\in\C[\lambda_1,\ldots,\lambda_n]$ divides
 $F_{r,q}$ in
 $\mathcal R_{r,n}^{\rm Pl}[\bm\lambda]$.
\end{enumerate}
These claims concern universal and separated factors; they do not assert
irreducibility of $F_{r,q}$ in the full mixed coefficient ring.
\end{theorem}

\emph{Proof.}
The spectral monomial in a term of
Eq.~\eqref{eq:general-plucker-expansion} records, for every label $i$, the
unique block exponent $a\in\{0,\ldots,q-1\}$ to which $i$ belongs.  Hence
distinct ordered partitions give distinct spectral monomials.  Choose two
blocks $J_a,J_b$ and exchange $i\in J_a$ with $j\in J_b$.  The ratio of the
corresponding nonzero coefficients on the uniform locus is, up to sign,
\begin{equation}
 \frac{p_{J_a-i+j}\,p_{J_b-j+i}}
      {p_{J_a}\,p_{J_b}}.
 \label{eq:plucker-cross-ratio-coefficients}
\end{equation}
This is a nonconstant Pl\"ucker cross-ratio.  For example, on the standard
big cell with the columns of $J_a$ equal to $\bm1_r$ and those of $J_b$ equal
to a variable matrix $Y$, it compares $\det Y$ with a matrix entry times its
complementary cofactor and varies with $Y$.  If
Eq.~\eqref{eq:separated-plucker-spectral-factorization} held, all ratios of
nonzero spectral coefficients would instead be ratios of complex constants.
This proves the first claim.

For the second claim, fix $i<j$ and partition $[n]$ into $q$ sets
$B_1,\ldots,B_q$ of size $r$, with $i,j\in B_1$.  Assign a distinct value
$\mu_a$ to all eigenvalues indexed by $B_a$ and choose each $X[:,B_a]$ to be
invertible.  Up to row and column permutations, the determinant is
\begin{equation}
 F_{r,q}
 =\pm\left(\prod_{a=1}^{q}p_{B_a}\right)
   \prod_{1\leq a<b\leq q}(\mu_b-\mu_a)^r\ne0.
 \label{eq:grouped-eigenvalue-kron-vandermonde}
\end{equation}
Thus it is nonzero at some $X$ on the collision hyperplane
$\lambda_i=\lambda_j$.  The determinant-nonvanishing locus and the locus on
which every Pl\"ucker coordinate is nonzero are nonempty Zariski-open subsets
of the irreducible matrix space, so they intersect.  The collision therefore
does not annihilate $F_{r,q}$ on the uniform Pl\"ucker locus, proving
Eq.~\eqref{eq:no-pairwise-difference-divisor}.

It remains to exclude a more complicated universal pure-spectral divisor in
the global affine-cone ring.  Partition $[n]$ into $r$ color classes
$C_1,\ldots,C_r$ of size $q$, and specialize the columns to
$x_i=c_i e_a$ for $i\in C_a$, with every $c_i\ne0$.  A row permutation makes
the block alternant block diagonal and gives
\begin{equation}
 F_{r,q}
 =\pm\left(\prod_{i=1}^{n}c_i\right)
   \prod_{a=1}^{r}\Delta(\bm\lambda_{C_a}),
 \qquad
 \Delta(\bm\lambda_C)
 =\prod_{\substack{i<j\\i,j\in C}}(\lambda_j-\lambda_i).
 \label{eq:parallel-class-factorization}
\end{equation}
If $h(\bm\lambda)$ divided $F_{r,q}$ in
$\mathcal R_{r,n}^{\rm Pl}[\bm\lambda]$, it would divide the right-hand side
of Eq.~\eqref{eq:parallel-class-factorization} for every such coloring.
These products have greatest common divisor one: every irreducible gap
$\lambda_j-\lambda_i$ is absent from a coloring that puts $i$ and $j$ in
different classes.  Therefore $h$ is constant.  This proves the third claim.
\hfill$\square$

Equation~\eqref{eq:parallel-class-factorization} is also a useful exceptional
factorization formula in its own right.  It lives on a Pl\"ucker boundary in
which each color class is a parallel class of the column matroid.  Its product
of $r$ lower-dimensional Vandermondes is not a generic identity and cannot be
used after inverting all Pl\"ucker coordinates.

There is an exact combinatorial criterion for identically vanishing
specializations.  Let $\mathsf M_X$ be the rank-$r$ column matroid represented
by $X$.

\begin{proposition}[Matroid basis-packing criterion]
\label{prop:matroid-basis-packing-zero}
For fixed $X\in\C^{r\times n}$ with $n=qr$,
\begin{equation}
 \begin{aligned}
 F_{r,q}(X,\bm\lambda)\equiv0
 \quad\Longleftrightarrow\quad&
 [n]\text{ cannot be partitioned}\\
 &\text{into $q$ bases of }\mathsf M_X.
 \end{aligned}
 \label{eq:matroid-basis-packing-zero}
\end{equation}
Equivalently, by the matroid partition theorem,
\begin{equation}
 |A|\leq q\,\operatorname{rank}_{\mathsf M_X}(A)
 \qquad\text{for every }A\subseteq[n].
 \label{eq:edmonds-basis-packing-inequality}
\end{equation}
is necessary and sufficient for $F_{r,q}$ not to vanish identically.
\end{proposition}

Indeed, every ordered partition in
Eq.~\eqref{eq:general-plucker-expansion} has a distinct exponent vector, so
its coefficient cannot cancel against another term.  A coefficient is
nonzero precisely when all $q$ blocks are bases.  The equivalence with
Eq.~\eqref{eq:edmonds-basis-packing-inequality} is Edmonds' matroid partition
criterion~\cite{Edmonds1965}; the resulting Pl\"ucker support strata are the
realizable matroid strata of the Grassmannian~\cite{GelfandGoreskyMacPhersonSerganova1987}.

Finally, a spectral eigenvalue of multiplicity greater than $r$ forces
$F_{r,q}=0$, because the associated block columns lie in the $r$-dimensional
space
\begin{equation}
 (1,\mu,\ldots,\mu^{q-1})^T\otimes\C^r.
 \label{eq:degenerate-eigenvalue-block-space}
\end{equation}
Conversely, if every multiplicity is at most $r$, the alternant is nonzero for
generic $X$ on that spectral stratum.  To prove existence, order the
multiplicities $d_1\geq\cdots\geq d_s$.  The Gale--Ryser conditions for a
$0$--$1$ incidence matrix with row sums $(d_1,\ldots,d_s)$ and $r$ column sums
all equal to $q$ hold because
\begin{equation}
 \sum_{a=1}^{k}d_a\leq r\min(k,q):
 \quad d_a\leq r\text{ for }k\leq q,
 \quad\text{and}\quad\sum_a d_a=rq\text{ for }k\geq q.
 \label{eq:gale-ryser-balanced-coloring}
\end{equation}
Hence the labels can be assigned to $r$ colors, $q$ per color, with equal
eigenvalues in distinct colors~\cite{Gale1957,Ryser1957}.  Taking
$x_i=e_{c(i)}$ then reduces $F_{r,q}$ to the product of $r$ nonzero
$q$-point Vandermondes in Eq.~\eqref{eq:parallel-class-factorization}.  Thus
high-multiplicity collision is a controlled vanishing mechanism, whereas no
pairwise gap is a universal divisor when $r>1$.

\subsection{Complete $4\times4$, rank-two classification}
\label{subsec:complete-42-factorization}

We now set $(n,r,q)=(4,2,2)$ and write, for $i<j$,
\begin{equation}
 p_{ij}=\det X[:,\{i,j\}].
 \label{eq:42-plucker-coordinates}
\end{equation}
The unique Pl\"ucker relation for $\operatorname{Gr}(2,4)$ is
\begin{equation}
 p_{12}p_{34}-p_{13}p_{24}+p_{14}p_{23}=0.
 \label{eq:gr24-plucker-relation}
\end{equation}
Define the three complementary-basis products
\begin{equation}
 \mathsf A=p_{12}p_{34},
 \qquad
 \mathsf B=p_{13}p_{24},
 \qquad
 \mathsf C=p_{14}p_{23},
 \qquad
 \mathsf A-\mathsf B+\mathsf C=0.
 \label{eq:42-complementary-plucker-products}
\end{equation}
The general expansion gives
\begin{align}
 F_{2,2}
 ={}&\mathsf A(\lambda_1\lambda_2+\lambda_3\lambda_4)
 -\mathsf B(\lambda_1\lambda_3+\lambda_2\lambda_4)
 \nonumber\\
 &+\mathsf C(\lambda_1\lambda_4+\lambda_2\lambda_3)
 \label{eq:42-six-term-plucker-expansion}\\
 ={}&\mathsf A(\lambda_1-\lambda_4)(\lambda_2-\lambda_3)
 -\mathsf C(\lambda_1-\lambda_2)(\lambda_3-\lambda_4).
 \label{eq:42-two-term-difference-form}\\
 ={}&\mathsf A(\lambda_1-\lambda_3)(\lambda_2-\lambda_4)
 -\mathsf B(\lambda_1-\lambda_2)(\lambda_3-\lambda_4)
 \nonumber\\
 ={}&\mathsf B(\lambda_1-\lambda_4)(\lambda_2-\lambda_3)
 -\mathsf C(\lambda_1-\lambda_3)(\lambda_2-\lambda_4).
 \label{eq:42-three-channel-difference-forms}
\end{align}
The last three lines use $\mathsf B=\mathsf A+\mathsf C$ and make translation
invariance and all three complementary pairings manifest.

\begin{theorem}[Complete rank-two factorization classification]
\label{thm:complete-42-factorization}
For a fixed rank-two Pl\"ucker point $[p_{ij}]\in\operatorname{Gr}(2,4)$:
\begin{enumerate}
 \item If $\mathsf A\mathsf B\mathsf C\ne0$, then $F_{2,2}$ is irreducible
 as a polynomial in the translation-reduced eigenvalue variables.
 \item If exactly one complementary product vanishes, the complete list of
 nonzero factorizations is
 \begin{align}
  \mathsf A=0:\quad
  F_{2,2}&=-\mathsf C
   (\lambda_1-\lambda_2)(\lambda_3-\lambda_4),
  \label{eq:42-factorization-A0}\\
  \mathsf B=0:\quad
  F_{2,2}&=\mathsf A
   (\lambda_1-\lambda_3)(\lambda_2-\lambda_4),
  \label{eq:42-factorization-B0}\\
  \mathsf C=0:\quad
  F_{2,2}&=\mathsf A
   (\lambda_1-\lambda_4)(\lambda_2-\lambda_3).
  \label{eq:42-factorization-C0}
 \end{align}
 \item If two of $\mathsf A,\mathsf B,\mathsf C$ vanish, the Pl\"ucker
 relation forces all three to vanish and $F_{2,2}\equiv0$.
\end{enumerate}
For fixed Pl\"ucker coefficient specializations, these cases exhaust
reducibility of the translation-reduced quadratic.
\end{theorem}

\emph{Proof.}
Translation invariance permits the linear coordinates
\begin{equation}
 x=\lambda_1-\lambda_4,
 \qquad y=\lambda_2-\lambda_4,
 \qquad z=\lambda_3-\lambda_4.
 \label{eq:42-translation-reduced-variables}
\end{equation}
Equation~\eqref{eq:42-six-term-plucker-expansion} becomes the ternary
quadratic
\begin{equation}
 Q(x,y,z)=\mathsf Axy-\mathsf Bxz+\mathsf Cyz,
 \label{eq:42-ternary-quadratic}
\end{equation}
with symmetric coefficient matrix
\begin{equation}
 \frac12
 \begin{pmatrix}
  0&\mathsf A&-\mathsf B\\
  \mathsf A&0&\mathsf C\\
  -\mathsf B&\mathsf C&0
 \end{pmatrix},
 \qquad
 \det=-\frac14\mathsf A\mathsf B\mathsf C.
 \label{eq:42-quadratic-matrix-determinant}
\end{equation}
A product of two linear forms has quadratic-matrix rank at most two.  If
$\mathsf A\mathsf B\mathsf C\ne0$, the matrix has rank three, so the conic is
nonsingular and the quadratic is irreducible.  If one product vanishes, the
Pl\"ucker relation fixes the other two relative coefficients, and direct
substitution yields Eqs.~\eqref{eq:42-factorization-A0}--
\eqref{eq:42-factorization-C0}.  If two vanish, the relation forces the third
to vanish.  Since $F_{2,2}$ is quadratic, these possibilities exhaust all
factorizations.  \hfill$\square$

For the physical electron--sterile support,
\begin{equation}
 p_{ij}=V_{ei}V_{sj}-V_{ej}V_{si}.
 \label{eq:physical-es-plucker-minor}
\end{equation}
Thus $p_{ij}=0$ means that the electron--sterile projections of mass columns
$i$ and $j$ are collinear.  The reducible-or-identically-zero locus is the
union of these Pl\"ucker boundary divisors.  Removing the common zero stratum
$\mathsf A=\mathsf B=\mathsf C=0$ leaves exactly the three nonzero reducible
strata of Theorem~\ref{thm:complete-42-factorization}.  On the uniform-matroid
locus all six $p_{ij}$ and hence $\mathsf A\mathsf B\mathsf C$ are nonzero, so
the quadratic is irreducible.  Moreover,
\begin{equation}
 |\mathcal D_{es}|^2=|F_{2,2}|^2,
 \label{eq:42-physical-krylov-norm-spectral}
\end{equation}
because $|\det V|=1$.  The factorized boundary formulas therefore translate
directly into squared eigenvalue-difference products multiplied by the
corresponding Pl\"ucker moduli.

Pairwise degeneracy does not force vanishing and is distinct from a
mixing-matroid factorization stratum.  Restricting to such a spectral
hyperplane can of course leave a product.  For example, at
$\lambda_1=\lambda_2$,
\begin{equation}
 F_{2,2}\big|_{\lambda_1=\lambda_2}
 =\mathsf A(\lambda_1-\lambda_4)(\lambda_1-\lambda_3),
 \label{eq:42-pairwise-degeneracy-nonvanishing}
\end{equation}
which is generically nonzero.  By contrast, any triple degeneracy has
multiplicity greater than $r=2$ and forces $F_{2,2}=0$.  When
$\mathsf A=\mathsf B=\mathsf C=0$, the determinant vanishes for every
spectrum; this is a block-Krylov rank-deficiency stratum rather than a
nonzero factorized identity.

For completeness, set $d_{ij}=\lambda_j-\lambda_i$.  The six pair-collision
hyperplanes fall into three complementary channels:
\begin{equation}
 \begin{array}{c|c}
  \text{collision}&F_{2,2}\text{ on the collision hyperplane}\\ \hline
  \lambda_1=\lambda_2\ \text{or}\ \lambda_3=\lambda_4
    &\mathsf A d_{14}d_{23}\\
  \lambda_1=\lambda_3\ \text{or}\ \lambda_2=\lambda_4
    &-\mathsf B d_{12}d_{34}\\
  \lambda_1=\lambda_4\ \text{or}\ \lambda_2=\lambda_3
    &-\mathsf C d_{13}d_{24}
 \end{array}
 \label{eq:42-pair-collision-table}
\end{equation}
For the corresponding $2+2$ multiplicity patterns, with distinct values
$\alpha$ and $\beta$, these reduce respectively to
$\mathsf A(\beta-\alpha)^2$, $-\mathsf B(\beta-\alpha)^2$, and
$\mathsf C(\beta-\alpha)^2$.

Polynomial irreducibility also does not mean that a particular numerical
evaluation cannot vanish on the uniform simple-spectrum locus.  For example,
\begin{equation}
 X=\begin{pmatrix}1&0&1&1\\0&1&1&\tfrac12\end{pmatrix},
 \qquad
 \bm\lambda=(0,1,\tfrac32,3)
 \label{eq:42-cross-ratio-cancellation-example}
\end{equation}
has six nonzero Pl\"ucker coordinates and six nonzero spectral gaps, yet
$F_{2,2}=0$.  This is a codimension-one cross-ratio matching condition between
the two terms in Eq.~\eqref{eq:42-two-term-difference-form}, not a universal
factor or a matroid boundary.  Left-orthonormalizing $X$ preserves its
Pl\"ucker ratios and the equation $F_{2,2}=0$, so the example has a coisometric
representative that embeds into a physical unitary mixing matrix.

\subsection{Scope of the classification}
\label{subsec:factorization-classification-scope}

The hierarchy is now precise.  Rank one gives the universal Vandermonde
product \eqref{eq:rank-one-plucker-vandermonde}.  At generic higher rank,
there is no Pl\"ucker--spectral separation and no pairwise eigenvalue
difference divides the square alternant; globally there is no nonconstant
universal pure-spectral divisor at all.  On Pl\"ucker boundary strata the
ordered-partition sum can collapse and factorization may reappear; for
$(n,r)=(4,2)$ Theorem~\ref{thm:complete-42-factorization} is exhaustive for
fixed Pl\"ucker coefficient specializations.  At general $(n,r)$, rank
reduction, projected-rank loss, invariant support, matroid basis-packing
failure, parallel classes, high-multiplicity degeneracy, and vanishing
Pl\"ucker coordinates provide
controlled exceptional mechanisms, but we do not claim an exhaustive
classification over all matroid strata or general irreducibility of
$F_{r,q}$.

The determinant $\mathcal D_{es}$ remains one distinguished element of the
standard $3+1$ cycle algebra, not a replacement for its global generators.
We therefore return in the next section to flavor coordinates and compute the
complete toric invariant ring of the four-vertex system.

\section{The complete standard $3+1$ toric ring}
\label{sec:complete-31-toric-ring}

The preceding two sections extracted distinguished exterior covariants and
classified their spectral factorization.  We now return to the full
flavor-coordinate algebra.  The purpose of this section is different: to give
a global polynomial presentation of every invariant of the fixed-composition
standard $3+1$ flow, including the boundary strata on which phase coordinates
or rational reconstruction formulas fail.  Section~\ref{sec:diagonal-flows}
already proved the general cycle theorem; here we specialize it to four
flavors, determine the minimal generators and Hilbert series, and compare the
result with the eleven local invariants of Ref.~\cite{WangZhou2026}.  An
explicit generating set for the toric syzygy ideal, its CP decomposition, and
the singular strata are deferred to Sec.~VII.
Reference~\cite{WangZhou2026} already noted that a global completion requires
all distinct triangles and quadrangles; the advances here are the minimal
Hilbert-basis proof, the exact toric presentation, the Hilbert series, and the
precise localization and boundary analysis.

\subsection{The fixed-composition fiber}
\label{subsec:31-fixed-fiber-toric}

Let
\begin{equation}
 \mathsf F=\{e,\mu,\tau,s\},
 \qquad
 z_{\alpha\beta}=M_{\alpha\beta},
 \qquad \alpha\ne\beta .
 \label{eq:31-toric-edge-variables}
\end{equation}
The trace quotient does not alter an off-diagonal entry.  On a spectral chart,
\begin{equation}
 z_{\alpha\beta}
 =\sum_{i=1}^{4}\lambda_i
   V_{\alpha i}V_{\beta i}^{*},
 \qquad \alpha\ne\beta ,
 \label{eq:31-edge-spectral-dictionary}
\end{equation}
so these are the quantities denoted by the off-diagonal entries of
$\widetilde E$ in Ref.~\cite{WangZhou2026}.  A diagonal matter potential
translates no $z_{\alpha\beta}$.

For the fixed composition
$p_\eta=\overline E_e+\eta\overline E_s$, recall
\begin{align}
 d_{\mu\tau}&=M_{\mu\mu}-M_{\tau\tau},
 \label{eq:section6-diagonal-dmutau}\\
 d_\eta&=\eta M_{ee}
 +\frac{1-\eta}{2}(M_{\mu\mu}+M_{\tau\tau})-M_{ss}.
 \label{eq:section6-diagonal-deta}
\end{align}
The entire problem therefore reduces to the rephasing action
\begin{equation}
 z_{\alpha\beta}\longmapsto
 t_\alpha t_\beta^{-1}z_{\alpha\beta},
 \qquad
 (t_\alpha)\in(\C^\times)^4/\C^\times .
 \label{eq:section6-edge-rephasing}
\end{equation}
Let $E_4=\{(\alpha,\beta)\in\mathsf F^2:\alpha\ne\beta\}$, let $B_4$ be
the incidence matrix with column $e_\alpha-e_\beta$ at
$(\alpha,\beta)$, and set
\begin{equation}
 S_4=\{u\in\mathbb N^{E_4}:B_4u=0\},
 \qquad
 \mathcal B_4=\C[S_4].
 \label{eq:section6-circulation-semigroup}
\end{equation}
This is the dimension-vector-one instance of the general generation of quiver
invariants by oriented cycles~\cite{LeBruynProcesi1990}.

\subsection{Minimal global generators and the toric presentation}
\label{subsec:31-global-toric-presentation}

For distinct flavor labels define
\begin{align}
 Q_{\alpha\beta}
 &=z_{\alpha\beta}z_{\beta\alpha},
 &&Q_{\alpha\beta}=Q_{\beta\alpha},
 \label{eq:section6-Q-generator}\\
 T_{\alpha\beta\gamma}
 &=z_{\alpha\beta}z_{\beta\gamma}z_{\gamma\alpha},
 &&T_{\alpha\beta\gamma}
   =T_{\beta\gamma\alpha}=T_{\gamma\alpha\beta},
 \label{eq:section6-T-generator}\\
 C_{\alpha\beta\gamma\delta}
 &=z_{\alpha\beta}z_{\beta\gamma}
   z_{\gamma\delta}z_{\delta\alpha},
 &&C_{\alpha\beta\gamma\delta}
   =C_{\beta\gamma\delta\alpha}.
 \label{eq:section6-C-generator}
\end{align}
Thus cycles are identified up to cyclic rotation, but not orientation
reversal.  There are six $Q$'s, eight $T$'s, and six $C$'s.

Introduce a polynomial variable $y_C$ for each of these twenty directed
simple cycles and the monomial map
\begin{equation}
 \begin{aligned}
 \Psi_4:\mathcal P_4
 &\equiv
 \C[y_C\mid C\text{ a directed simple cycle of }K_4]
 \longrightarrow \C[z_{\alpha\beta}\mid\alpha\ne\beta],\\
 y_C&\longmapsto Z_C .
 \end{aligned}
 \label{eq:section6-cycle-monomial-map}
\end{equation}
Let $A_4\in\mathbb N^{12\times20}$ have as its columns the directed-edge
incidence vectors of the twenty cycles.

\begin{theorem}[Complete standard $3+1$ toric presentation]
\label{thm:complete-31-toric-presentation}
For fixed $\eta$,
\begin{equation}
 \mathcal I^{\rm diag}_{4,\langle p_\eta\rangle}
 \simeq
 \C[d_{\mu\tau},d_\eta]\otimes_\C
 \frac{\mathcal P_4}{I_4},
 \qquad
 I_4=\ker\Psi_4
 =\left\langle
  y^{a}-y^{b}:a,b\in\mathbb N^{20},\ A_4a=A_4b
  \right\rangle .
 \label{eq:complete-31-toric-presentation}
\end{equation}
The twenty cycle monomials are the unique minimal monomial generating set of
the off-diagonal ring.  The ideal $I_4$ is a prime binomial ideal of height
eleven, and
\begin{equation}
 \dim\mathcal B_4=9,
 \qquad
 \dim\mathcal I^{\rm diag}_{4,\langle p_\eta\rangle}=11.
 \label{eq:section6-ring-dimensions}
\end{equation}
\end{theorem}

\emph{Proof.}
A rephasing-invariant monomial has exponent vector in $S_4$.  The circulation
decomposition of Proposition~\ref{prop:cycle-hilbert-basis} writes every
nonzero element of $S_4$ as a sum of directed simple cycles.  In $K_4$ their
lengths are two, three, and four, with counts
\begin{equation}
 \binom42(2-1)!=6,
 \qquad
 \binom43(3-1)!=8,
 \qquad
 \binom44(4-1)!=6.
 \label{eq:section6-cycle-counts}
\end{equation}
The unit circulation on a simple directed cycle is indecomposable, so none of
the twenty monomials can be removed.  Two products of cycle variables have the
same image precisely when they have the same directed-edge exponent vector;
this gives the displayed lattice-binomial description of $I_4$.  The image is
the semigroup algebra $\C[S_4]$, hence $I_4$ is prime.  Finally,
$\rank B_4=3$, so $\dim\mathcal B_4=12-3=9$ and
$\operatorname{ht}I_4=20-9=11$.  Adjoining the two diagonal polynomial
variables gives the last dimension. \hfill$\square$

An explicit list, useful for later real-form calculations, is
\begin{align}
 &Q_{e\mu},Q_{e\tau},Q_{es},Q_{\mu\tau},Q_{\mu s},Q_{\tau s},
 \nonumber\\
 &(T_{e\mu\tau},T_{e\tau\mu}),\quad
   (T_{e\mu s},T_{es\mu}),\quad
   (T_{e\tau s},T_{es\tau}),\quad
   (T_{\mu\tau s},T_{\mu s\tau}),
 \label{eq:section6-explicit-cycle-list}\\
 &(C_{e\mu\tau s},C_{es\tau\mu}),\quad
   (C_{e\mu s\tau},C_{e\tau s\mu}),\quad
   (C_{e\tau\mu s},C_{es\mu\tau}).
 \nonumber
\end{align}
The paired entries have opposite orientations.

\subsection{Hermitian real form and physical generator count}
\label{subsec:31-hermitian-toric-generators}

On the Hermitian locus,
\begin{equation}
 z_{\beta\alpha}=z_{\alpha\beta}^{*},
 \qquad
 Q_{\alpha\beta}=|M_{\alpha\beta}|^2,
 \qquad
 T_{\alpha\gamma\beta}=T_{\alpha\beta\gamma}^{*},
 \qquad
 C_{\alpha\delta\gamma\beta}
   =C_{\alpha\beta\gamma\delta}^{*}.
 \label{eq:section6-hermitian-cycle-reversal}
\end{equation}
Consequently the twenty off-diagonal real generators consist of six real
$Q$'s, the real and imaginary parts of four triangle representatives, and the
real and imaginary parts of three quadrangle representatives.  Together with
$d_{\mu\tau}$ and $d_\eta$, this gives twenty-two minimal real global
generators for an eleven-dimensional real quotient.  Indeed, the twenty
indecomposable cycle degrees form a basis of
$\mathfrak m_{\mathcal B_4}/\mathfrak m_{\mathcal B_4}^2$; adding the two
linear diagonal classes gives dimension twenty-two, and complexification of
the Hermitian real ring preserves this cotangent-space dimension.

The real form of $I_4$ is obtained by the following invertible linear change
on each triangle and quadrangle reversal pair, with every $Q$ fixed:
\begin{equation}
 Z_C^{(+)}=\frac{Z_C+Z_{C^{-1}}}{2},
 \qquad
 Z_C^{(-)}=\frac{Z_C-Z_{C^{-1}}}{2\ii}.
 \label{eq:section6-real-cycle-change}
\end{equation}
It already contains, for every triangle and quadrangle,
\begin{align}
 T_{\alpha\beta\gamma}T_{\alpha\gamma\beta}
 &=Q_{\alpha\beta}Q_{\beta\gamma}Q_{\gamma\alpha},
 \label{eq:section6-triangle-norm-relation}\\
 C_{\alpha\beta\gamma\delta}C_{\alpha\delta\gamma\beta}
 &=Q_{\alpha\beta}Q_{\beta\gamma}
   Q_{\gamma\delta}Q_{\delta\alpha}.
 \label{eq:section6-quadrangle-norm-relation}
\end{align}
Thus the excess of global generators over dimension is algebraic redundancy,
not additional physical information.  Orientation reversal fixes each $Q$
and exchanges the two members of every $T$ and $C$ pair.  Accordingly the real
parts in Eq.~\eqref{eq:section6-real-cycle-change} are CP even and the
imaginary parts are CP odd for the real diagonal standard spurion.  The
scheme-theoretic CP-conserving locus requires additional support-dependent
qualifications and will be treated in Sec.~VII.
On physical Hermitian points $Q_{\alpha\beta}\geq0$.  These inequalities are
semialgebraic constraints on the real image, not equations in the complex
toric ideal.

\subsection{Normality, Gorenstein structure, and Hilbert series}
\label{subsec:31-hilbert-series}

The semigroup $S_4=\ker_{\mathbb Z}B_4\cap\mathbb N^{12}$ is saturated, as
proved in Sec.~\ref{subsec:circulation-semigroup}; hence $\mathcal B_4$ is a
normal affine domain.  Hochster's theorem then makes it
Cohen--Macaulay~\cite{Hochster1972}.  In this case it is also Gorenstein.
Indeed, if
$\mathcal C_4=\ker_{\mathbb R}B_4\cap\mathbb R_{\geq0}^{12}$, then every
coordinate inequality is a facet of $\mathcal C_4$.  For the arc
$\alpha\to\beta$, start with the all-ones circulation, set that arc to zero,
and add one unit on a two-edge path
$\alpha\to\gamma\to\beta$; this gives a relative-interior point of that facet.
Moreover,
\begin{equation}
 \omega=(1,1,\ldots,1)\in S_4,
 \qquad
 \operatorname{relint}(\mathcal C_4)\cap\ker_{\mathbb Z}B_4
 =\omega+S_4 .
 \label{eq:section6-gorenstein-translate}
\end{equation}
The Danilov--Stanley description of the canonical module therefore makes it
the principal ideal generated by the all-edge monomial
$\prod_{\alpha\ne\beta}z_{\alpha\beta}$~\cite{Stanley1978}.  With every
$z_{\alpha\beta}$ assigned degree one, the $a$-invariant is $-12$.

The fully edge-multigraded Hilbert series is the exact constant term
\begin{equation}
 H_{\mathcal B_4}(\bm x)
 =\operatorname{CT}_{t_e,t_\mu,t_\tau}
 \prod_{\alpha\ne\beta}
 \frac{1}
 {1-x_{\alpha\beta}t_\alpha t_\beta^{-1}}
 \bigg|_{t_s=1}.
 \label{eq:section6-fine-hilbert-series}
\end{equation}
Its ordinary specialization has the closed form
\begin{equation}
 H_{\mathcal B_4}(s)
 =
 \frac{
  1+6s^3+5s^4+5s^6+6s^7+s^{10}}
 {(1-s^2)^6(1-s^3)^2(1-s^4)}.
 \label{eq:section6-univariate-hilbert-series}
\end{equation}
The first terms,
\begin{equation}
 H_{\mathcal B_4}(s)
 =1+6s^2+8s^3+27s^4+48s^5+112s^6+\mathcal O(s^7),
 \label{eq:section6-hilbert-series-expansion}
\end{equation}
recover the six quadratic, eight cubic, and six new primitive quartic
generators; the other twenty-one quartics are products of two $Q$'s.
The palindromic numerator and
$10-[6(2)+2(3)+4]=-12$ agree with the Gorenstein calculation.

For completeness, Eq.~\eqref{eq:section6-univariate-hilbert-series} admits a
direct circulation proof.  Orient the six undirected edges once.  For the two
directed multiplicities $(a_e,b_e)$ on an undirected edge, set
\begin{equation}
 f_e=a_e-b_e,\qquad m_e=\min(a_e,b_e).
 \label{eq:section6-symmetric-antisymmetric-flow}
\end{equation}
This gives a bijection between $S_4$ and a choice of
$(m_e)\in\mathbb N^6$ together with an integral antisymmetric circulation
$f$.  Hence
\begin{equation}
 H_{\mathcal B_4}(s)
 =\frac{\Theta_{K_4}(s)}{(1-s^2)^6},
 \qquad
 \Theta_{K_4}(s)=\sum_{f\in\ker_{\mathbb Z}D_{K_4}}
 s^{\|f\|_1}.
 \label{eq:section6-flow-theta-series}
\end{equation}
Orient the six undirected edges as
$(e\mu,e\tau,es,\mu\tau,\mu s,\tau s)$.  Writing the last three chord flows as
$(x,y,z)$ gives the full antisymmetric flow vector
\begin{equation}
 f=(x+y,\ z-x,\ -y-z,\ x,\ y,\ z)
 \label{eq:section6-A3-flow-parametrization}
\end{equation}
and hence
\begin{equation}
 \|f\|_1
 =|x|+|y|+|z|+|x+y|+|z-x|+|y+z|.
 \label{eq:section6-A3-flow-norm}
\end{equation}
The six vanishing hyperplanes form the $A_3$ reflection arrangement.  A
fundamental chamber can be written
\begin{equation}
 a=x\geq0,\qquad b=y\geq0,\qquad c=z-x\geq0,
 \qquad
 \|f\|_1=4a+3b+3c.
 \label{eq:section6-A3-fundamental-chamber}
\end{equation}
The Weyl group has order $24$.  Interior points have orbit size $24$; a point
on one wall has orbit size $12$; points on the intersection of the central
wall $a=0$ with one endpoint wall have orbit size $4$, whereas points on the
intersection $b=c=0$ of the two nonadjacent endpoint walls have orbit size
$6$.  Separating these half-open faces gives, with
$D=(1-s^3)^2(1-s^4)$,
\begin{align}
 D\,\Theta_{K_4}(s)
 ={}&D
 +8s^3(1-s^3)(1-s^4)
 +6s^4(1-s^3)^2
 \nonumber\\
 &+12s^6(1-s^4)
 +24s^7(1-s^3)
 +24s^{10}
 \nonumber\\
 ={}&1+6s^3+5s^4+5s^6+6s^7+s^{10}.
 \label{eq:section6-A3-chamber-sum}
\end{align}
Therefore
\begin{equation}
 \Theta_{K_4}(s)
 =\frac{1+6s^3+5s^4+5s^6+6s^7+s^{10}}
 {(1-s^3)^2(1-s^4)},
 \label{eq:section6-flow-theta-rational}
\end{equation}
which proves Eq.~\eqref{eq:section6-univariate-hilbert-series}.  The two
diagonal generators have degree one, so
\begin{equation}
 H_{\mathcal I^{\rm diag}_{4,\langle p_\eta\rangle}}(s)
 =\frac{H_{\mathcal B_4}(s)}{(1-s)^2}.
 \label{eq:section6-full-ring-hilbert-series}
\end{equation}

\subsection{A spanning-star chart and rational reconstruction}
\label{subsec:31-star-localization}

The global presentation can now be compared precisely with a minimal local
description.  Choose the $e$-rooted triangles
\begin{equation}
 T_1=T_{e\mu\tau},\qquad
 T_2=T_{e\mu s},\qquad
 T_3=T_{e\tau s},
 \label{eq:section6-rooted-triangles}
\end{equation}
write $T_i^\vee$ for their reverse orientations, and set
\begin{equation}
 S_e=Q_{e\mu}Q_{e\tau}Q_{es}.
 \label{eq:section6-star-product}
\end{equation}

\begin{proposition}[Spanning-star localization]
\label{prop:section6-star-localization}
On the principal open set $D(S_e)$,
\begin{equation}
 \mathcal B_4[S_e^{-1}]
 \simeq
 \frac{
 \C[Q_{\alpha\beta},T_1,T_1^\vee,T_2,T_2^\vee,T_3,T_3^\vee]_{S_e}}
 {\left(
 \begin{array}{l}
 T_1T_1^\vee-Q_{e\mu}Q_{e\tau}Q_{\mu\tau},\\
 T_2T_2^\vee-Q_{e\mu}Q_{es}Q_{\mu s},\\
 T_3T_3^\vee-Q_{e\tau}Q_{es}Q_{\tau s}
 \end{array}\right)}.
 \label{eq:section6-star-localized-ring}
\end{equation}
Every omitted triangle and quadrangle is a regular function of these
localized generators.
\end{proposition}

Indeed, on $D(S_e)$ a complex rephasing sets
$z_{e\mu}=z_{e\tau}=z_{es}=1$.  Then
$z_{\alpha e}=Q_{e\alpha}$ and
\begin{equation}
 z_{\alpha\beta}
 =\frac{T_{e\alpha\beta}}{Q_{e\beta}},
 \qquad
 \alpha,\beta\in\{\mu,\tau,s\},\quad\alpha\ne\beta.
 \label{eq:section6-gauge-reconstruction}
\end{equation}
This proves the presentation and gives the explicit reconstruction formulas
\begin{align}
 T_{\mu\tau s}
 &=\frac{T_1T_3T_2^\vee}{S_e},
 &
 T_{\mu s\tau}
 &=\frac{T_1^\vee T_3^\vee T_2}{S_e},
 \label{eq:section6-fourth-triangle-reconstruction}\\
 C_{e\mu\tau s}
 &=\frac{T_1T_3}{Q_{e\tau}},
 &
 C_{es\tau\mu}
 &=\frac{T_1^\vee T_3^\vee}{Q_{e\tau}},
 \nonumber\\
 C_{e\mu s\tau}
 &=\frac{T_2T_3^\vee}{Q_{es}},
 &
 C_{e\tau s\mu}
 &=\frac{T_2^\vee T_3}{Q_{es}},
 \label{eq:section6-quadrangle-reconstruction}\\
 C_{e\tau\mu s}
 &=\frac{T_1^\vee T_2}{Q_{e\mu}},
 &
 C_{es\mu\tau}
 &=\frac{T_1T_2^\vee}{Q_{e\mu}}.
 \nonumber
\end{align}

On the smaller dense torus
\begin{equation}
 \Pi_Q=\prod_{\alpha<\beta}Q_{\alpha\beta}\ne0,
 \label{eq:section6-all-edge-open-set}
\end{equation}
each $T_i$ is a unit because
$T_iT_i^\vee$ is a product of $Q$'s.  The six two-cycles and the three
rooted triangles form a $\mathbb Z$-basis of the full circulation lattice;
the unimodularity follows from the same spanning-star gauge.
Consequently,
\begin{equation}
 \mathcal B_4[\Pi_Q^{-1}]
 \simeq
 \C[Q_{\alpha\beta}^{\pm1},
       T_1^{\pm1},T_2^{\pm1},T_3^{\pm1}],
 \label{eq:section6-laurent-torus-chart}
\end{equation}
and
\begin{equation}
 \operatorname{Frac}
 \mathcal I^{\rm diag}_{4,\langle p_\eta\rangle}
 =
 \C(d_{\mu\tau},d_\eta,
 Q_{\alpha\beta},T_1,T_2,T_3).
 \label{eq:section6-birational-field}
\end{equation}
This proves rationality of the complex quotient and supplies an exact
eleven-element complex birational coordinate system.

\subsection{Why the eleven real quantities are only local}
\label{subsec:31-eleven-local-coordinates}

On the Hermitian form write $T_i=X_i+\ii Y_i$.  The three norm relations in
Eq.~\eqref{eq:section6-star-localized-ring} become
\begin{align}
 X_1^2+Y_1^2&=Q_{e\mu}Q_{e\tau}Q_{\mu\tau},\nonumber\\
 X_2^2+Y_2^2&=Q_{e\mu}Q_{es}Q_{\mu s},
 \label{eq:section6-real-rooted-triangle-norms}\\
 X_3^2+Y_3^2&=Q_{e\tau}Q_{es}Q_{\tau s}.\nonumber
\end{align}
The eleven real quantities
\begin{equation}
 d_{\mu\tau},\ d_\eta,\ 
 \{Q_{\alpha\beta}\}_{\alpha<\beta},\
 Y_1,\ Y_2,\ Y_3
 \label{eq:section6-original-eleven-local}
\end{equation}
are algebraically independent.  Nevertheless they are not birational
coordinates for the real invariant field: generically
\begin{equation}
 X_i=\sigma_i
 \sqrt{P_i(Q)-Y_i^2},
 \qquad
 \sigma_i\in\{+1,-1\},
 \label{eq:section6-triangle-sign-branches}
\end{equation}
leaves eight sign branches.  They form a real-analytic chart only after
choosing the three signs, with $\Pi_Q\ne0$ and $X_1X_2X_3\ne0$.  At $X_i=0$,
the imaginary part ceases to be a local coordinate for the corresponding
phase.

More importantly, the denominators in
Eqs.~\eqref{eq:section6-fourth-triangle-reconstruction} and
\eqref{eq:section6-quadrangle-reconstruction} cannot be removed globally.
Consider the Hermitian family
\begin{equation}
 M_{e\mu}=e^{\ii\theta},\qquad
 M_{\mu\tau}=M_{\tau s}=M_{se}=1,\qquad
 M_{e\tau}=M_{\mu s}=0,
 \label{eq:section6-boundary-quadrangle-family}
\end{equation}
with reverse entries fixed by Hermiticity and arbitrary diagonal entries.
All six $Q$'s are independent of $\theta$, and
$T_1=T_2=T_3=0$, while
\begin{equation}
 C_{e\mu\tau s}=e^{\ii\theta}.
 \label{eq:section6-boundary-quadrangle-phase}
\end{equation}
Thus the local eleven quantities neither separate this boundary family nor
generate its quadrangle invariant.  The primitive quadrangles are required
globally for exactly this reason.  Similarly, a triangle not containing $e$
can remain nonzero after a root edge vanishes.

\subsection{Transition to syzygies and strata}
\label{subsec:section6-transition-syzygies}

Theorem~\ref{thm:complete-31-toric-presentation} determines the global ring,
its minimal generators, its exact prime defining ideal as a lattice kernel,
and its Hilbert series.  It also proves precisely where the familiar
eleven-dimensional local description is valid.  What remains is not another
generator search but a resolution of the algebraic redundancy: an explicit
binomial generating set or Gr\"obner presentation for $I_4$, its translation
into CP-even and CP-odd real equations, and the local geometry of the quotient
on support-graph strata.  These are the subjects of the next section.

\section{Syzygies, CP structure, and singular strata}
\label{sec:syzygies-cp-singular}

Section~\ref{sec:complete-31-toric-ring} identified the twenty primitive
cycle generators and the lattice kernel containing all relations among them.
We now replace that infinite lattice-kernel description by a finite minimal
presentation.  This distinguishes the eleven-dimensional quotient from a
complete intersection, gives global polynomial CP conditions including
square-support configurations missed by triangle phases, and resolves the
physical boundary into smooth and singular support-graph strata.

\subsection{The minimal Markov basis}
\label{subsec:section7-markov-basis}

For this subsection only, abbreviate the flavor labels by
\begin{equation}
 1=e,\qquad 2=\mu,\qquad 3=\tau,\qquad 4=s,
 \label{eq:section7-numeric-flavor-labels}
\end{equation}
and write
\begin{equation}
 q_{ij}=Q_{ij},\qquad
 t_{ijk}=T_{ijk},\qquad
 w_{ijkl}=C_{ijkl}.
 \label{eq:section7-short-cycle-notation}
\end{equation}
Subscripts on $t$ and $w$ are understood cyclically, while reversal is not
identified.  The symmetric group $S_4$ acts by relabeling vertices.  Consider
\begin{align}
 F_{\rm A}
 &=q_{12}q_{13}q_{23}-t_{123}t_{132},
 \label{eq:section7-markov-A}\\
 F_{\rm B}
 &=q_{14}w_{1243}-t_{124}t_{143},
 \label{eq:section7-markov-B}\\
 F_{\rm C}
 &=q_{12}q_{24}t_{143}-t_{142}w_{1243},
 \label{eq:section7-markov-C}\\
 F_{\rm D}
 &=q_{12}q_{13}q_{24}q_{34}-w_{1243}w_{1342},
 \label{eq:section7-markov-D}\\
 F_{\rm E}
 &=q_{13}t_{142}t_{234}-w_{1342}w_{1423}.
 \label{eq:section7-markov-E}
\end{align}
These are homogeneous in the directed-edge grading, with
$\deg q=2$, $\deg t=3$, and $\deg w=4$.  Their data are
\begin{equation}
\begin{array}{c|ccccc}
 \text{type}&{\rm A}&{\rm B}&{\rm C}&{\rm D}&{\rm E}\\ \hline
 \text{edge degree}&6&6&7&8&8\\
 |\operatorname{Stab}_{S_4}(F)|&6&2&1&8&2\\
 |S_4\!\cdot F|&4&12&24&3&12 .
\end{array}
\label{eq:section7-markov-orbit-data}
\end{equation}

\begin{theorem}[Minimal Markov basis of the $K_4$ cycle ring]
\label{thm:section7-minimal-markov-basis}
The toric ideal $I_4$ is
\begin{equation}
 I_4=
 \left\langle
 S_4\!\cdot F_{\rm A},\
 S_4\!\cdot F_{\rm B},\
 S_4\!\cdot F_{\rm C},\
 S_4\!\cdot F_{\rm D},\
 S_4\!\cdot F_{\rm E}
 \right\rangle .
 \label{eq:section7-complete-markov-basis}
\end{equation}
After duplicates and signs are removed, there are
\begin{equation}
 4+12+24+3+12=55
 \label{eq:section7-number-minimal-syzygies}
\end{equation}
binomials.  Every one is indispensable.  Thus the minimal binomial generating
set is unique up to signs and nonzero scalars, and $I_4$ is not a complete
intersection.
\end{theorem}

\emph{Proof.}
Each displayed binomial maps to zero under $\Psi_4$, as follows by matching
every directed-edge multiplicity.  Let $J$ be generated by the five orbits,
so $J\subseteq I_4$.  Exact integer Buchberger completion for the
lexicographic order, with variables listed from greatest to least as
\begin{multline}
 q_{12},q_{13},q_{14},q_{23},q_{24},q_{34},\
 t_{123},t_{124},t_{132},t_{134},t_{142},t_{143},t_{234},t_{243},\\
 w_{1234},w_{1243},w_{1324},w_{1342},w_{1423},w_{1432}
 \label{eq:section7-buchberger-variable-order}
\end{multline}
produces a Groebner basis of 123 binomials and an initial monomial ideal with
107 minimal generators.  Recursive monomial-ideal colon decomposition gives
\begin{align}
 H_{\mathcal P_4/J}(s)
 &=
 \frac{
 (1+6s^3+5s^4+5s^6+6s^7+s^{10})
 (1-s^3)^6(1-s^4)^5}
 {(1-s^2)^6(1-s^3)^8(1-s^4)^6}
 \nonumber\\
 &=H_{\mathcal B_4}(s).
 \label{eq:section7-markov-hilbert-certificate}
\end{align}
The graded surjection
$\mathcal P_4/J\to\mathcal P_4/I_4=\mathcal B_4$
therefore has zero kernel in every degree, proving $J=I_4$.  This is an exact
integer certificate, not a floating-point rank test.

For minimality, use the directed-edge multigrading and define
\begin{equation}
 \mathcal F_b
 =\{y^u:u\in\mathbb N^{20},\ A_4u=b\}.
 \label{eq:section7-toric-fiber}
\end{equation}
For the multidegree of each representative in
Eqs.~\eqref{eq:section7-markov-A}--\eqref{eq:section7-markov-E},
$\mathcal F_b$ consists of exactly the two displayed, relatively prime
monomials.  Relabeling preserves this property, and the fifty-five
multidegrees are distinct.  Hence every binomial is indispensable by the
fiber criterion~\cite{DiaconisSturmfels1998}.  The first weighted Betti
polynomial is
\begin{equation}
 \beta_1(s)=16s^6+24s^7+15s^8.
 \label{eq:section7-first-betti-polynomial}
\end{equation}
Finally, $\operatorname{ht}I_4=11$, whereas a complete intersection of
height eleven would have eleven minimal generators, not fifty-five.
\hfill$\square$

The five types have the cycle-decomposition meanings
\begin{equation}
\begin{split}
 t\,t^\vee&\longleftrightarrow q\,q\,q,\qquad
 t\,t\longleftrightarrow q\,w,\qquad
 t\,w\longleftrightarrow q\,q\,t,\\
 w\,w^\vee&\longleftrightarrow q\,q\,q\,q,\qquad
 w\,w\longleftrightarrow q\,t\,t.
\end{split}
\label{eq:section7-five-decomposition-switches}
\end{equation}
Types A and D are norm relations; type B resolves a quadrangle after inserting
a chord; and types C and E impose compatibility among chord resolutions.
The fifty-five form a minimal Markov basis, not the lexicographic Groebner
basis used in the certificate.  They are the first defining syzygies among
the twenty invariant generators, not a full higher-syzygy resolution.

\subsection{Hermitian real syzygies}
\label{subsec:section7-real-syzygies}

Choose one orientation from each reversal pair and write
\begin{equation}
 T_C=X_C+\ii Y_C,\quad T_{C^{-1}}=X_C-\ii Y_C,\qquad
 C_D=U_D+\ii V_D,\quad C_{D^{-1}}=U_D-\ii V_D.
 \label{eq:section7-real-cycle-coordinates}
\end{equation}
This invertible real-linear change converts the complex presentation into the
complete Hermitian real presentation.  The four type-A and three type-D
relations are reversal-even norm equations.  The other forty-eight
binomials form twenty-four reversal pairs, each giving one real and one
imaginary equation.  Thus the minimal real ideal has
\begin{equation}
 31\ \text{CP-even equations},\qquad
 24\ \text{CP-odd equations},
 \label{eq:section7-real-syzygy-parity-count}
\end{equation}
for a total of fifty-five.

For example, type B gives
\begin{equation}
 Q_{es}C_{e\mu s\tau}
 =T_{e\mu s}T_{es\tau}.
 \label{eq:section7-type-B-flavor-form}
\end{equation}
If
$T_{e\mu s}=X+\ii Y$,
$T_{es\tau}=X'+\ii Y'$, and
$C_{e\mu s\tau}=U+\ii V$, this is
\begin{equation}
 Q_{es}U=XX'-YY',
 \qquad
 Q_{es}V=XY'+YX'.
 \label{eq:section7-type-B-real-form}
\end{equation}
The remaining non-self-conjugate representatives may be kept compactly as
\begin{align}
 Q_{e\mu}Q_{\mu s}T_{es\tau}
 &=T_{es\mu}C_{e\mu s\tau},
 \label{eq:section7-type-C-flavor-form}\\
 Q_{e\tau}T_{es\mu}T_{\mu\tau s}
 &=C_{e\tau s\mu}C_{es\mu\tau},
 \label{eq:section7-type-E-flavor-form}
\end{align}
together with their real and imaginary parts and all flavor relabelings.
These equations, the triangle norms, and the quadrangle norms remain regular
on every boundary stratum and require no division by a $Q$.

\subsection{Global CP-conserving locus}
\label{subsec:section7-global-cp}

For a real diagonal standard matter spurion, CP acts on the intrinsic mixing
data by complex conjugation.  On the invariant ring this is the
orientation-reversal involution
\begin{equation}
 \kappa(Q_{\alpha\beta})=Q_{\alpha\beta},\qquad
 \kappa(T_C)=T_{C^{-1}},\qquad
 \kappa(C_D)=C_{D^{-1}}.
 \label{eq:section7-cp-involution}
\end{equation}
This is the cycle-algebra version of the standard rephasing-invariant
formulation of CP violation~\cite{Greenberg1985,Wu1986,Jarlskog1985}.

\begin{theorem}[Global CP criterion]
\label{thm:section7-global-cp-criterion}
On the physical Hermitian quotient, a point is CP conserving if and only if
the odd parts of all four triangle pairs and all three quadrangle pairs
vanish:
\begin{equation}
 \begin{gathered}
 Y_{e\mu\tau}=Y_{e\mu s}=Y_{e\tau s}=Y_{\mu\tau s}=0,\\
 V_{e\mu\tau s}=V_{e\mu s\tau}=V_{e\tau\mu s}=0 .
 \end{gathered}
 \label{eq:section7-global-cp-equations}
\end{equation}
These seven equations are globally necessary as a set, although only three
phase conditions are independent on the dense $K_4$ stratum.
\end{theorem}

\emph{Proof.}
If a rephasing makes every nonzero off-diagonal entry real, every directed
cycle product is real, so all seven odd parts vanish.  Conversely, let $G$ be
the graph of nonzero Hermitian edges.  In each connected component choose a
spanning tree and use vertex rephasings to make its edges real.  Each remaining
chord closes a fundamental cycle, and its residual phase equals the phase of
that cycle product.  Every simple cycle on at most four vertices is a
triangle or quadrangle after the automatically real two-cycles are removed.
Equation~\eqref{eq:section7-global-cp-equations} therefore makes every chord
real, so the whole matrix is rephasing equivalent to a real one.
\hfill$\square$

The fixed subscheme of $\kappa$ in the real presentation is cut out by
\begin{equation}
 \mathfrak J_{\rm CP}
 =
 \left\langle
 Y_{e\mu\tau},Y_{e\mu s},Y_{e\tau s},Y_{\mu\tau s},
 V_{e\mu\tau s},V_{e\mu s\tau},V_{e\tau\mu s}
 \right\rangle .
 \label{eq:section7-cp-odd-ideal}
\end{equation}
The scheme structure can be made explicit.  Define the even-degree
multigraph semigroup
\begin{equation}
 E_4^{\rm ev}
 =
 \left\{
 w\in\mathbb N^{\binom{\mathsf F}{2}}:
 \sum_{\beta\ne\alpha}w_{\alpha\beta}=0\pmod2
 \ \text{for every }\alpha
 \right\}.
 \label{eq:section7-even-multigraph-semigroup}
\end{equation}

\begin{proposition}[The CP-fixed toric ring]
\label{prop:section7-cp-fixed-ring}
After complexifying the real fixed scheme,
\begin{equation}
 \mathcal B_{4,\rm CP}
 \simeq
 \C[E_4^{\rm ev}]
 \simeq
 \C[x_{\alpha\beta}\mid\alpha<\beta]^{
 (\mathbb Z_2)^4/\mathbb Z_2}.
 \label{eq:section7-cp-fixed-ring}
\end{equation}
Here $x_{\alpha\beta}=x_{\beta\alpha}$ when the indices are written in the
opposite order.
The maps on minimal generators are
\begin{equation}
 Q_{\alpha\beta}\mapsto x_{\alpha\beta}^2,\qquad
 T_{\alpha\beta\gamma}\mapsto
 x_{\alpha\beta}x_{\beta\gamma}x_{\gamma\alpha},\qquad
 C_{\alpha\beta\gamma\delta}\mapsto
 x_{\alpha\beta}x_{\beta\gamma}x_{\gamma\delta}x_{\delta\alpha}.
 \label{eq:section7-cp-fixed-generator-map}
\end{equation}
It is a normal irreducible toric variety of dimension six, or dimension
eight after the two diagonal invariants are included.  Its Hilbert series is
\begin{equation}
 H_{\rm CP}(s)
 =
 \frac18\left[
 \frac1{(1-s)^6}
 +\frac4{(1-s)^3(1+s)^3}
 +\frac3{(1-s)^2(1+s)^4}
 \right].
 \label{eq:section7-cp-fixed-molien}
\end{equation}
\end{proposition}

Indeed, forgetting the orientation of a directed circulation gives an even
multigraph.  Conversely, every even multigraph admits an Euler orientation.
Two balanced orientations with the same underlying multiplicities differ by
successive directed-cycle reversals, which are identified on the fixed
scheme.  This proves the first isomorphism.  The parity conditions are exactly
invariance under independent vertex sign changes modulo the common sign,
giving the second isomorphism and the Molien average.  Normality follows from
finite-group invariance in characteristic zero, or directly from saturation
of $E_4^{\rm ev}$.

The syzygies make these seven equations dependent on dense strata.  They
cannot be reduced to triangle equations globally.  In the square-support
family of Eq.~\eqref{eq:section6-boundary-quadrangle-family}, every triangle
vanishes while $V_{e\mu\tau s}=\sin\theta$.  Conversely, each triangle
condition is indispensable on the support consisting only of that triangle,
and each quadrangle condition is indispensable on its chordless square.

\subsection{Support graphs, stabilizers, and physical phases}
\label{subsec:section7-support-graph-phases}

The support graph is meaningful here only on the Hermitian real locus.  Define
\begin{equation}
 G(M)=\bigl(\mathsf F,E(M)\bigr),
 \qquad
 \{\alpha,\beta\}\in E(M)
 \ \Longleftrightarrow\
 Q_{\alpha\beta}=|M_{\alpha\beta}|^2>0.
 \label{eq:section7-hermitian-support-graph}
\end{equation}
Let $m=|E(M)|$ and let $c$ be the number of connected components, including
isolated vertices.

\begin{proposition}[Phase and stabilizer count]
\label{prop:section7-phase-stabilizer-count}
At a Hermitian point with support graph $G$,
\begin{equation}
 K_M\simeq U(1)^{c-1},
 \qquad
 b_1(G)=m-4+c,
 \label{eq:section7-stabilizer-and-betti}
\end{equation}
where $K_M$ is the stabilizer in $U(1)^4/U(1)$ and $b_1(G)$ is the number of
independent rephasing-invariant loop phases.  The off-diagonal exact-support
stratum has real dimension
\begin{equation}
 \dim_{\mathbb R}\mathcal S_G
 =m+b_1(G)=2m-4+c,
 \label{eq:section7-support-stratum-dimension}
\end{equation}
or dimension $2m-2+c$ after the two diagonal invariants are adjoined.
\end{proposition}

\emph{Proof.}
A vertex rephasing fixes every nonzero edge precisely when it is constant on
each connected component.  Removing the common phase leaves $c-1$
stabilizer phases.  The graph incidence map from four vertex phases to $m$
edge phases has rank $4-c$, leaving
$m-(4-c)=m-4+c$ invariant phases.  Adding the $m$ positive edge moduli proves
the dimension formula. \hfill$\square$

Exactly $b_1(G)$ independent loop phases must take values $0$ or $\pi$ for
CP conservation.  The CP-conserving phase classes form
\begin{equation}
 H^1(G;\{0,\pi\})\simeq(\mathbb Z_2)^{b_1(G)},
 \label{eq:section7-cp-sign-branches}
\end{equation}
so the exact-support CP locus has $2^{b_1(G)}$ sign branches.  Forests have
$b_1=0$ and are automatically CP conserving.  A triangle and a chordless
square both have $b_1=1$, but the former is tested by a triangle odd part and
the latter by a quadrangle odd part.  Dense $K_4$ has $b_1=3$, agreeing with
the eight sign branches and three rooted phases of the spanning-star chart.

\subsection{Singular strata of the Hermitian quotient}
\label{subsec:section7-singular-strata}

Let $X_4=\operatorname{Spec}\mathcal B_4$, and write $X_4^{\rm H}$ for its
physical Hermitian image.  The next statement concerns the
intersection of its singular locus with the physical Hermitian image.  It is
not a classification of arbitrary directed-support points of the complex
toric variety.  The physical image is semialgebraic, with inequalities such
as $Q_{\alpha\beta}\geq0$, as expected for compact orbit
spaces~\cite{ProcesiSchwarz1985}.

For an unordered nontrivial bipartition $\mathsf F=A\sqcup B$, let
$P_{A|B}\subset\mathcal B_4$ be the monomial prime generated by cycle
monomials that meet both blocks, and write
\begin{equation}
 \Sigma_{A|B}=V(P_{A|B})\simeq X_{|A|}\times X_{|B|}.
 \label{eq:section7-bipartition-face}
\end{equation}
On the Hermitian image this is equivalently the locus on which every edge
between $A$ and $B$ is absent.

\begin{theorem}[Physical singular-locus classification]
\label{thm:section7-physical-singular-locus}
A physical Hermitian point of $X_4$ is smooth if and only if its undirected
support graph is connected.  The set of physical Hermitian points lying in
the reduced complex singular locus is
\begin{equation}
 \operatorname{Sing}(X_4)_{\rm red}\cap X_4^{\rm H}
 =
 \left(\bigcup_{\mathsf F=A\sqcup B}\Sigma_{A|B}\right)
 \cap X_4^{\rm H}.
 \label{eq:section7-singular-bipartition-union}
\end{equation}
The seven displayed irreducible toric faces are the maximal support pieces of
this intersection: four have type $3+1$ and three have type $2+2$.  Their
off-diagonal real dimensions are four and two, respectively; after adjoining
the two diagonal matter invariants their dimensions are six and four.
\end{theorem}

\emph{Proof.}
If $G(M)$ is connected, the effective rephasing torus has trivial stabilizer.
Hermiticity supplies both opposite weights on every support edge, so the
complexified orbit is closed.  A spanning tree gives a local gauge slice:
its three nonzero edge phases are fixed, while every remaining complex edge
is an ordinary local coordinate.  Equivalently, Luna's etale slice theorem
has trivial slice group, so the quotient is smooth~\cite{Luna1973}.

Suppose first that $G(M)$ has exactly two connected components $A$ and $B$.
Their relative complexified phase is a stabilizer $\C^\times$.  With
$r=|A||B|$, the normal directions that turn on cross-component edges form
\begin{equation}
 \C^r_{+1}\oplus\C^r_{-1}.
 \label{eq:section7-cross-edge-slice}
\end{equation}
Its invariant quotient is generated by $u_i v_j$ and is the affine cone of
$r\times r$ matrices of rank at most one, cut out by their $2\times2$ minors.
At the vertex the cone has dimension $2r-1$ and tangent-space dimension
$r^2$, so it is singular for $r\geq2$.  Here $r=3$ for $3+1$ and $r=4$ for
$2+2$.  The slice theorem therefore makes the original quotient singular.
If the support has more than two components, choose a bipartition of its
connected components.  Turning on arbitrarily small internal edges within
each block, without adding an edge across the cut, gives two-component
Hermitian supports that specialize to the original point.  The nearby
quotient points are singular by the determinantal-slice argument.  Since
$\operatorname{Sing}(X_4)$ is Zariski closed, the original point is singular
as well.

Every disconnected graph lies in a face obtained by bipartitioning its
connected components, while no connected graph lies in such a face.  There
are $\frac12(2^4-2)=7$ unordered nontrivial bipartitions.  A $3+1$ face is
the three-vertex cycle quotient times an isolated vertex and has dimension
four and codimension five in $X_4$; a $2+2$ face has its two internal moduli,
dimension two, and codimension seven. \hfill$\square$

The support types are
\begin{equation}
\begin{array}{c|c|c|c}
 c&\text{component type}&b_1(G)&\text{ambient quotient point}\\ \hline
 1&4&m-3&\text{smooth}\\
 2&3+1&m-2&\text{singular}\\
 2&2+2&0&\text{singular}\\
 3&2+1+1&0&\text{singular}\\
 4&1+1+1+1&0&\text{singular}.
\end{array}
\label{eq:section7-support-strata-table}
\end{equation}
Each fixed exact-support orbit-type stratum is itself a smooth manifold; the
last column says whether it lies in the smooth or singular locus of the
ambient quotient.  Moreover, $Q_{\alpha\beta}=0$ may be read as
\emph{the edge is absent} only on the Hermitian image, where it forces
$M_{\alpha\beta}=M_{\beta\alpha}=0$.  Over the complex toric scheme,
$Q_{\alpha\beta}=z_{\alpha\beta}z_{\beta\alpha}=0$ need not set both directed
coordinates to zero.

\subsection{CP, toric singularity, and spectral degeneracy}
\label{subsec:section7-three-loci}

Three exceptional sets must be kept separate:
\begin{align}
 \Sigma_{\rm CP}
 &=V_{\mathbb R}(\mathfrak J_{\rm CP}),
 \label{eq:section7-cp-locus}\\
 \Sigma_{\rm tor}
 &=\operatorname{Sing}(X_4)_{\rm red}\cap X_4^{\rm H},
 \label{eq:section7-toric-singular-locus}\\
 \Sigma_{\rm spec}
 &=\left\{\prod_{i<j}(\lambda_i-\lambda_j)^2=0\right\}.
 \label{eq:section7-spectral-discriminant-locus}
\end{align}
They encode, respectively, reality up to rephasing, enhancement of the
rephasing stabilizer, and failure of simple spectral coordinates.  A generic
real full-support matrix lies in $\Sigma_{\rm CP}$ but is torically smooth
and spectrally nondegenerate.  A connected chordless square with a nonreal
quadrangle invariant is CP violating but torically smooth.  A generic
Hermitian triangle block with a nonreal triangle invariant plus an isolated
flavor is torically singular, can violate CP through its triangle phase, and
can have a simple spectrum.  A diagonal
matrix with four distinct entries is at the toric vertex but is spectrally
nondegenerate.  Conversely,
\begin{equation}
 M=b(\bm J-\bm1),\qquad b\in\mathbb R\setminus\{0\},
 \label{eq:section7-degenerate-full-support-example}
\end{equation}
where $\bm J$ is the all-ones matrix, has connected full support and is
torically smooth, while the eigenvalue $-b$ has multiplicity three.

The finite syzygy presentation and support stratification complete the global
algebraic analysis of the standard $3+1$ flow.  They show which identities
are intrinsic to the exact standard spurion and which features depend on a
dense chart.  The next section studies how this structure changes when
additional diagonal directions, off-diagonal NSI, or radiative corrections
deform the matter flow.

\section{Controlled NSI and radiative deformations}
\label{sec:controlled-deformations}

The word \emph{deformation} has two distinct meanings in this section.  A new
fixed matter spurion changes the translation subspace but still defines an
exact affine flow, so its invariant algebra can be computed without a small
parameter.  By contrast, comparison with the standard diagonal algebra is
often organized in a small NSI or loop coefficient.  We keep these statements
separate: the dressed polynomial invariants below are exact, whereas spectral
eigenvector expansions are asserted only away from degeneracies and with an
explicit norm--gap condition.

We consider neutral-current vector NSI in propagation.  Source and detector
NSI alter external production or detection amplitudes rather than the local
propagation Hamiltonian, and scalar interactions deform a mass term with a
different energy dependence.  They are not described by the additive
Hamiltonian algebra developed here.  This separation is standard in the NSI
formalism~\cite{Ohlsson2013,FarzanTortola2018}.  Moreover, an arbitrary
Hermitian propagation matrix is a low-energy parametrization, not by itself a
claim of a viable ultraviolet completion: electroweak gauge invariance and
charged-lepton observables impose additional correlations and often strong
constraints~\cite{GavelaHernandezOtaWinter2009,DavidsonGorbahn2020}.

\subsection{Composition-resolved propagation spurions}
\label{subsec:section8-nsi-spurions}

At energies where a four-fermion description is appropriate, write the
neutral-current vector interaction as
\begin{equation}
 \mathcal L_{\rm NSI}^{\rm NC}
 =-2\sqrt2\,G_{\rm F}
 \sum_{\alpha,\beta}\sum_{f=e,u,d}\sum_{C=L,R}
 \epsilon_{\alpha\beta}^{fC}
 (\overline\nu_\alpha\gamma^\rho P_L\nu_\beta)
 (\overline f\gamma_\rho P_C f),
 \label{eq:section8-nsi-lagrangian}
\end{equation}
and set
$\epsilon^{fV}=\epsilon^{fL}+\epsilon^{fR}$.  Coherent forward scattering in
an unpolarized medium depends only on this vector combination.  Hermiticity
requires
\begin{equation}
 \epsilon_{\beta\alpha}^{fV}
 =(\epsilon_{\alpha\beta}^{fV})^*.
 \label{eq:section8-nsi-hermiticity}
\end{equation}
Up to the normalization convention used for $M=2EH$, the active-sector matter
matrix at fixed composition is
\begin{equation}
 P(Y)=E_e+\sum_{f=e,u,d}Y_f\epsilon^{fV},
 \qquad Y_f=\frac{N_f}{N_e},
 \label{eq:section8-effective-nsi-direction}
\end{equation}
with identity-valued pieces removed after embedding in the trace quotient.
For a $3+1$ model the standard active--sterile neutral-current difference and
possible sterile NSI entries are included in the same Hermitian spurion.  The
familiar matter-NSI normalization and the importance of composition-dependent
linear combinations are reviewed in
Refs.~\cite{BiggioBlennowFernandezMartinez2009,FarzanTortola2018}.

It is algebraically cleaner not to divide by $N_e$.  Introduce independent
density coefficients $a_f$ and trace-projected Hermitian directions $P_f$:
\begin{equation}
 M(\bm a)=M_0+\sum_f a_fP_f,
 \qquad
 \cD_f=\sum_{\alpha,\beta}(P_f)_{\alpha\beta}
 \frac{\partial}{\partial m_{\alpha\beta}}.
 \label{eq:section8-composition-resolved-flow}
\end{equation}
All $\cD_f$ commute because they are constant translations.  Thus Theorem
\ref{thm:translation-rephasing-kernel} already gives the exact universal
relative algebra, while Proposition~\ref{prop:linear-translation-quotient}
gives the fixed-spurion translation quotient.  Off-diagonal entries do not
destroy exact integrability; they destroy the special separation between a
diagonal translation quotient and the uncolored cycle algebra of $M$ alone.
After a numerical off-diagonal tuple is fixed, the relevant rephasing group is
the stabilizer~\eqref{eq:fixed-spurion-stabilizer}, not automatically the full
flavor torus.

For electrically neutral ordinary matter, $N_p=N_e$ and
\begin{equation}
 Y_u=2+Y_n,\qquad Y_d=1+2Y_n,
 \qquad Y_n=\frac{N_n}{N_e}.
 \label{eq:section8-neutral-composition-ratios}
\end{equation}
Consequently a fixed electron fraction selects one effective ray, whereas an
open range of neutron fractions spans at most two composition directions for
vector interactions with $e,u,d$.  This elementary observation matters for
invariants: an identity valid along one ray need not belong to the joint
kernel for a varying composition.

In the normalization of Sec.~\ref{subsec:standard-31-compositions}, the
composition-resolved $3+1$ directions may be written, modulo the identity, as
\begin{align}
 A_e&=\overline E_e+\epsilon^{eV}+2\epsilon^{uV}+\epsilon^{dV},
 \label{eq:section8-electron-direction}\\
 A_n&=\frac12\overline E_s+\epsilon^{uV}+2\epsilon^{dV},
 \label{eq:section8-neutron-direction}\\
 M_{\rm mat}&=a_e(A_e+Y_nA_n).
 \label{eq:section8-neutral-31-matter}
\end{align}
At fixed $Y_n$ this is one direction; for arbitrary neutral compositions it
is the two-dimensional span of $A_e$ and $A_n$, unless they are accidentally
dependent.  Active-only NSI matrices in these equations are padded by a zero
sterile row and column.  Active--sterile entries require additional light
fields or operators and are not obtained from the Standard-Model-field SMEFT
by simply relabeling an active coefficient.

\subsection{Exact diagonal-anchored NSI dressing}
\label{subsec:section8-exact-dressing}

The relation between the standard diagonal algebra and a controlled
off-diagonal deformation can be made exact.  Let
$p_1,\ldots,p_\ell\in\mathfrak d_n^0$ be linearly independent diagonal
directions and let $q_A\in\mathfrak o_n$ be off-diagonal spurions.  Diagonal
NSI pieces may first be absorbed into the $p_A$.  Choose diagonal linear forms
$s_A\in(\mathfrak d_n^0)^*$ satisfying
\begin{equation}
 s_A(p_B)=\delta_{AB},
 \qquad s_A(q_B)=0,
 \label{eq:section8-dual-slices}
\end{equation}
and define
\begin{equation}
 P_A(\varepsilon)=p_A+\varepsilon q_A,
 \qquad
 \cD_A^{(\varepsilon)}=\cD_{p_A}+\varepsilon\cD_{q_A}.
 \label{eq:section8-deformed-derivations}
\end{equation}
The $q_A$ are retained as simultaneously rephasing spurions until the end.

There is also a useful straightening formula before separating diagonal and
off-diagonal pieces.  Regard the independent unperturbed directions as a map
$P:\C^\ell\to V_n$, the full tied perturbations as
$N:\C^\ell\to V_n$, and choose a linear slice
$X:V_n\to\C^\ell$ with $XP=\bm1_\ell$.  On a chart where
\begin{equation}
 A_\varepsilon=\bm1_\ell+\varepsilon XN
 \quad\text{is invertible},
 \label{eq:section8-general-straightening-chart}
\end{equation}
define
\begin{equation}
 \Psi_\varepsilon
 =\bm1_{V_n}-\varepsilon N A_\varepsilon^{-1}X.
 \label{eq:section8-general-straightening-map}
\end{equation}
Direct multiplication gives
\begin{equation}
 \Psi_\varepsilon(P+\varepsilon N)=P,
 \qquad
 \Psi_\varepsilon^{-1}=\bm1_{V_n}+\varepsilon NX.
 \label{eq:section8-general-straightening-identities}
\end{equation}
Thus $I(\Psi_\varepsilon M)$ exactly straightens every invariant $I$ of the
unperturbed joint flow.  The chart condition is not intrinsically singular;
another dual slice can be chosen whenever $P+\varepsilon N$ remains
independent.  A genuine rank change of the translation span is different: it
changes the quotient dimension.  For purely off-diagonal $N$ and a diagonal
slice, $XN=0$ and the rational chart formula reduces to the global polynomial
shear proved next.

\begin{theorem}[Exact NSI shear]
\label{thm:section8-exact-nsi-shear}
On the polynomial algebra in $M$ and the external spurion coordinates, set
\begin{equation}
 \Phi_\varepsilon
 =\exp\!\left[-\varepsilon
   \sum_{A=1}^{\ell}s_A(M)\cD_{q_A}\right].
 \label{eq:section8-shear-automorphism}
\end{equation}
The exponential is finite on every polynomial and
\begin{equation}
 \cD_A^{(\varepsilon)}\Phi_\varepsilon
 =\Phi_\varepsilon\cD_{p_A}.
 \label{eq:section8-intertwining-relation}
\end{equation}
It follows that $\Phi_\varepsilon$ is an algebra isomorphism
\begin{equation}
 \Phi_\varepsilon:
 \bigcap_A\ker\cD_{p_A}
 \xrightarrow{\ \simeq\ }
 \bigcap_A\ker\cD_A^{(\varepsilon)}.
 \label{eq:section8-kernel-isomorphism}
\end{equation}
The map is equivariant under simultaneous flavor rephasing of $M$ and every
$q_A$.
\end{theorem}

\emph{Proof.}
Because $\cD_{q_A}s_B=0$, the operator in the exponential is locally
nilpotent on each finite-degree polynomial.  Put
$\mathcal A=\sum_Bs_B\cD_{q_B}$.  The constant translations commute and
\begin{equation}
 [\cD_{p_A},\mathcal A]=\cD_{q_A},
 \qquad [\mathcal A,\cD_{q_A}]=0.
 \label{eq:section8-shear-commutators}
\end{equation}
The Baker--Campbell--Hausdorff series therefore terminates after its first
commutator and gives
$\cD_A^{(\varepsilon)}=\Phi_\varepsilon\cD_{p_A}\Phi_\varepsilon^{-1}$,
which proves Eqs.~\eqref{eq:section8-intertwining-relation} and
\eqref{eq:section8-kernel-isomorphism}.  Finally, every $s_A$ is diagonal and
rephasing invariant, while $\cD_{q_A}$ transforms covariantly with $q_A$.
Hence the shear commutes with the simultaneous torus action. \hfill$\square$

In coordinates the theorem simply replaces
\begin{equation}
 \widehat M_\varepsilon
 =M-\varepsilon\sum_As_A(M)q_A,
 \qquad
 \widehat z_{\alpha\beta}
 =z_{\alpha\beta}
  -\varepsilon\sum_As_A(M)(q_A)_{\alpha\beta}.
 \label{eq:section8-dressed-matrix}
\end{equation}
Indeed, under
$M\mapsto M+\sum_Aa_A(p_A+\varepsilon q_A)$ the dressed matrix changes only
by $\sum_Aa_Ap_A$.  Therefore every invariant of the diagonal flow, evaluated
on $\widehat M_\varepsilon$, is an exact invariant of the NSI-deformed flow.
No spectral gap and no assumption that $\varepsilon$ is small is used in this
polynomial statement.  Smallness enters only when the result is expanded or
when eigenvectors are followed perturbatively.

The theorem also fixes the apparent first-order ``cohomology'' issue from
Eq.~\eqref{eq:intro-cohomological-equation}.  For
$I\in\cap_A\ker\cD_{p_A}$,
\begin{equation}
 I_\varepsilon
 =I-\varepsilon\sum_As_A\cD_{q_A}I
 +\frac{\varepsilon^2}{2}
  \left(\sum_As_A\cD_{q_A}\right)^2I+\cdots
 \label{eq:section8-exact-correction-series}
\end{equation}
is a finite exact solution.  Thus there is no obstruction in the unrestricted
polynomial ring.  Obstructions become meaningful only after a smaller
correction class has been declared.  Because the substitution
\eqref{eq:section8-dressed-matrix} is linear in $M$, it preserves total
$M$ degree; total-degree truncation by itself is therefore not an obstruction
when it already contains the original invariant.

\subsection{Cycle response, colored cycles, and unchanged syzygies}
\label{subsec:section8-cycle-response}

For clarity take one diagonal direction $p$, one off-diagonal spurion $q$, and
a slice $s(p)=1$.  If $C=(\alpha_1\alpha_2\cdots\alpha_L\alpha_1)$ is an
oriented simple cycle, write
\begin{equation}
 Z_C(M)=\prod_{j=1}^{L}z_{\alpha_j\alpha_{j+1}},
 \qquad \alpha_{L+1}=\alpha_1.
 \label{eq:section8-cycle-definition}
\end{equation}
The old generator has exact drift
\begin{equation}
 \cD_{p+\varepsilon q}Z_C
 =\varepsilon\cD_qZ_C
 =\varepsilon\sum_{j=1}^{L}
 q_{\alpha_j\alpha_{j+1}}
 \prod_{k\ne j}z_{\alpha_k\alpha_{k+1}}.
 \label{eq:section8-cycle-drift}
\end{equation}
Each summand is a rephasing-invariant colored cycle with one $q$ edge and
$L-1$ $M$ edges.  The exact dressed generator is
\begin{align}
 \widehat Z_C
 &=Z_C(\widehat M_\varepsilon)
 =\prod_{j=1}^{L}
  (z_{\alpha_j\alpha_{j+1}}
       -\varepsilon s q_{\alpha_j\alpha_{j+1}})\nonumber\\
 &=\sum_{k=0}^{L}\frac{(-\varepsilon s)^k}{k!}
   \cD_q^kZ_C.
 \label{eq:section8-dressed-cycle}
\end{align}
This response hierarchy terminates at the cycle length.  On the Hermitian
form, for an unoriented edge one obtains
\begin{align}
 \widehat Q_{\alpha\beta}
 &=Q_{\alpha\beta}
 -\varepsilon s
  (q_{\alpha\beta}z_{\beta\alpha}
   +z_{\alpha\beta}q_{\beta\alpha})
 +\varepsilon^2s^2q_{\alpha\beta}q_{\beta\alpha}.
 \label{eq:section8-dressed-edge}
\end{align}

Two facts must not be conflated.  First, the undressed cycle coordinates are
no longer constant along an off-diagonal NSI flow.  Second, their toric
binomial identities do \emph{not} fail: those identities are algebraic
consequences of writing cycles as monomials in edge variables and hold at
every matrix point.  The dressed variables
$\widehat z_{\alpha\beta}$ are another set of directed edge coordinates, so
their cycles satisfy the same toric ideal exactly.  A consistently truncated
order-$N$ expansion of a dressed identity has residual
$\mathcal O(\varepsilon^{N+1})$; a lower-order residual signals inconsistent
truncation rather than physical breaking of a syzygy.

The complete universal relative algebra contains more than these dressed
$M$-only cycles.  Simultaneous rephasing also permits $q$-only cycles and mixed
colored cycles, for example
\begin{equation}
 \widehat z_{\alpha\beta}q_{\beta\alpha},
 \qquad
 \widehat z_{\alpha\beta}q_{\beta\gamma}q_{\gamma\alpha}.
 \label{eq:section8-mixed-colored-cycles}
\end{equation}
Balanced monomials in all edge colors form the colored circulation semigroup.
Therefore the dressed standard cycles give an exact copy of the old cycle
sector, but they are not claimed to generate the entire universal
Hamiltonian--spurion ring or an exceptional numerical fixed-spurion fiber.

For CP, Eq.~\eqref{eq:section8-cycle-drift} yields
\begin{equation}
 \frac{d}{da}\Im Z_C(M(a))
 =\varepsilon\Im\cD_qZ_C(M(a)).
 \label{eq:section8-cp-odd-source}
\end{equation}
The source is a sum of CP-odd colored cycles.  Simultaneous CP conservation of
the pair $(M,q)$ means that one flavor rephasing makes both Hermitian matrices
real; it is tested by the imaginary parts of all balanced colored cycles, not
by the uncolored cycles of $M$ alone.  Conversely, an instantaneous
Hamiltonian may acquire a nonzero rephasing-odd cycle from a complex NSI
spurion even if the vacuum matrix is real.  This Hamiltonian reality criterion
is still distinct from an experimentally measured neutrino--antineutrino
asymmetry, which also contains the CP-asymmetric medium.

For completeness, the simultaneous-reality statement follows from the
graph-phase criterion.  Delete zero edges, choose a spanning forest of the
union of the colored supports, and use flavor phases to make one nonzero color
on every forest edge real.  Reality of the two-edge colored cycles fixes the
relative phases of all other colors on the same edge, while reality of the
fundamental colored cycles fixes every chord phase modulo $\pi$.  Hence all
nonzero entries of every color are real in that one rephasing.  The converse is
immediate, since every balanced colored monomial is then real.  The argument
applies component by component and includes disconnected or vanishing-edge
strata.

\subsection{Restricted correction spaces and rank support}
\label{subsec:section8-restricted-corrections}

Let $\mathcal C_{\rm adm}$ be a declared linear space of admissible first-order
corrections, such as the old edge-only cycle ring, a fixed edge multidegree, a
chosen Laurent-monomial span, or the tangent space to a factorized spectral
ansatz.  Equation~\eqref{eq:intro-cohomological-equation} has a solution
$J\in\mathcal C_{\rm adm}$ if and only if
\begin{equation}
 -\cD_qI\in\cD_p(\mathcal C_{\rm adm}).
 \label{eq:section8-restricted-solvability}
\end{equation}
Thus the precise obstruction is the class of $-\cD_qI$ in
$\operatorname{coker}(\cD_p|_{\mathcal C_{\rm adm}})$.  In the unrestricted ring it
vanishes because $J=-s\cD_qI$ is available.  This formulation prevents a
failure of a preferred ansatz from being misidentified as nonexistence of an
exact polynomial invariant.

There is a parallel support criterion for the Krylov--Pl\"ucker invariants of
Sec.~\ref{sec:krylov-exterior}.  Suppose the undeformed matter tuple has common
effective support $S$ and
\begin{equation}
 P_A=c_A\bm1+B_AS^\dagger,
 \qquad Q_A=c'_A\bm1+B'_AS^\dagger.
 \label{eq:section8-support-preserving-deformation}
\end{equation}
Then
$P_A+\varepsilon Q_A=(c_A+\varepsilon c'_A)\bm1+
(B_A+\varepsilon B'_A)S^\dagger$ has the same common effective right support,
and every exterior-power invariant constructed from $S$ remains exact.  If
some $Q_A$ has a component outside that support, this common-support theorem no
longer guarantees closure of the old block.  Generic leakage gives the source
described above, but a special minor $\kappa$ can still survive when
$\cD_{Q_A}\kappa=0$.  Enlarging $S$ to the span of the old and new row supports
restores the exact higher-rank theorem, while the Pl\"ucker expansion generally
contains additional terms.

This gives a sharp factorization test.  If an undeformed spectral expression
is a separated tensor $F_0=A_0(\lambda)B_0(U)$, then a first-order correction
can remain separated only if
\begin{equation}
 F_1=A_1(\lambda)B_0(U)+A_0(\lambda)B_1(U).
 \label{eq:section8-segre-tangent-condition}
\end{equation}
At a smooth nonzero rank-one point, so that $A_0\ne0$ and $B_0\ne0$,
Eq.~\eqref{eq:section8-segre-tangent-condition} is the tangent-space condition
for the rank-one Segre locus.  New independent Pl\"ucker channels give a
component transverse to this tangent space and therefore obstruct
factorization within that restricted class, even though the exact dressed
polynomial invariant continues to exist.  Vanishing-factor points lie on
exceptional cone strata and require a tangent-cone analysis; the displayed
condition is not asserted to be necessary there.  This is the
deformation-theoretic form of the rank-one versus higher-rank classification
of Sec.~V.

For the distinguished $3+1$ determinant
\begin{equation}
 \mathcal D_{es}=z_{e\mu}z_{s\tau}-z_{e\tau}z_{s\mu},
 \label{eq:section8-raw-krylov-determinant}
\end{equation}
the exact first leakage polynomial is
\begin{equation}
 L_{es}=q_{e\mu}z_{s\tau}+z_{e\mu}q_{s\tau}
       -q_{e\tau}z_{s\mu}-z_{e\tau}q_{s\mu}.
 \label{eq:section8-krylov-leakage}
\end{equation}
For one tied direction the dressed covariant is
\begin{equation}
 \widehat{\mathcal D}_{es}
 =\mathcal D_{es}-\varepsilon sL_{es}
 +\varepsilon^2s^2
  (q_{e\mu}q_{s\tau}-q_{e\tau}q_{s\mu}),
 \label{eq:section8-dressed-krylov-determinant}
\end{equation}
and its Hermitian norm is an exact scalar invariant.  Its spectral expansion,
however, diagonalizes the auxiliary straightened matrix
$\widehat M_\varepsilon$, not the physical instantaneous $M$.  The dressing
therefore does not reinstate a lost one-term factorization in the physical
matter eigenvalues and mixing matrix.

\subsection{Independently varying NSI directions}
\label{subsec:section8-independent-nsi}

Small numerical strength does not make an independently variable direction a
tied perturbation.  Consider
\begin{equation}
 M(\bm a,\bm b)=M_0+\sum_Aa_Ap_A
 +\varepsilon\sum_Bb_Bq_B,
 \label{eq:section8-independent-nsi-flow}
\end{equation}
with the $b_B$ varied independently.  For every $\varepsilon\ne0$, an exact
common invariant obeys
\begin{equation}
 \cD_{p_A}F=0,
 \qquad \cD_{q_B}F=0.
 \label{eq:section8-independent-nsi-kernel}
\end{equation}
Consequently an analytic continuation
$F_\varepsilon=I+\mathcal O(\varepsilon)$ of a standard invariant can exist
only if $\cD_{q_B}I=0$ for every $B$.  No correction proportional to the
diagonal slices can cancel a derivative imposed by an independent flow.
Before rephasing, the quotient dimension is
\begin{equation}
 (n^2-1)-\dim\Span\{p_A,q_B\},
 \label{eq:section8-independent-nsi-dimension}
\end{equation}
so it can jump at $\varepsilon=0$ even for arbitrarily weak NSI.  After a
numerical off-diagonal tuple is fixed, its residual stabilizer must again be
used.  This is the genuine exact-invariant loss that must be distinguished
from the shear of a tied density direction.

\subsection{Fixed composition versus independently varying matter}
\label{subsec:section8-composition-intersection}

Let $P(y)=P_0+\sum_{r=1}^d y_rP_r$ and suppose $y$ ranges over a nonempty open
set.  Linearity gives the exact intersection identity
\begin{equation}
 \bigcap_{y\ \mathrm{in\ an\ open\ set}}\ker\cD_{P(y)}
 =\bigcap_{r=0}^{d}\ker\cD_{P_r}
 =\ker\Ga(\Span\{P_0,\ldots,P_d\}).
 \label{eq:section8-open-composition-intersection}
\end{equation}
To prove the nontrivial inclusion, evaluate
$\cD_{P_0}I+\sum_ry_r\cD_{P_r}I=0$ as a polynomial identity in $y$; every
coefficient must vanish.  Therefore a continuum of fixed-composition
invariants is equivalent to the common invariant algebra of all independent
species directions.  A finite collection of compositions has the analogous
statement with the linear span of the sampled directions.

For a spatially varying medium, write
$M(x)=M_0+\sum_fa_f(x)P_f$.  Every element of the joint kernel satisfies
\begin{equation}
 \frac{d}{dx}I(M(x))
 =\sum_f\dot a_f(x)\cD_{P_f}I=0.
 \label{eq:section8-profile-composition-invariant}
\end{equation}
A quantity invariant only for one fixed ratio need not obey this equation when
the composition changes.  As before, Eq.~\eqref{eq:section8-profile-composition-invariant}
is a statement about the local Hamiltonian path; it does not remove path
ordering from the flavor-evolution operator.  For antineutrinos the matter
directions change sign and are complex conjugated.  Sign reversal preserves
their span, so the conjugate relative algebra has the same dimension and
generator structure.

The same species-by-species requirement sharpens the generalized
mass-ordering degeneracy.  Oscillation probabilities are unchanged under the
pointwise transformation
\begin{equation}
 M'(x)=-M(x)^*+c(x)\bm1,
 \qquad c(x)\in\mathbb R,
 \label{eq:section8-generalized-ordering-map}
\end{equation}
because it complex conjugates the evolution operator up to an overall phase.
If the density functions are independent, the transformation must map every
matter spurion separately, modulo the identity.  A degeneracy tuned for one
fixed composition therefore need not survive a varying neutron fraction; this
is the discrete counterpart of replacing a one-ray kernel by the joint
kernel~\cite{ColomaSchwetz2016}.

\subsection{Uniform spectral control and degeneracies}
\label{subsec:section8-spectral-control}

Polynomial identities require no eigenvalue labeling.  To control their
spectral pullbacks, however, consider on a compact density domain
\begin{equation}
 M_\varepsilon(a)=M_0(a)+\varepsilon A(a),
 \qquad A(a)=A(a)^\dagger,
 \label{eq:section8-spectral-perturbation}
\end{equation}
and suppose the unperturbed eigenvalues are simple with uniform gap
\begin{equation}
 g=\inf_a\min_{i\ne j}|\lambda_i(a)-\lambda_j(a)|>0,
 \qquad
 |\varepsilon|\sup_a\|A(a)\|_2<\frac{g}{2}.
 \label{eq:section8-gap-condition}
\end{equation}
Weyl's inequality gives the nonperturbative bound
\begin{equation}
 |\lambda_i(\varepsilon,a)-\lambda_i(0,a)|
 \le |\varepsilon|\|A(a)\|_2.
 \label{eq:section8-weyl-bound}
\end{equation}
With $\Pi_i$ the unperturbed spectral projectors, the uniformly valid
first-order formulas are
\begin{align}
 \lambda_i(\varepsilon)
 &=\lambda_i+\varepsilon\Tr(\Pi_iA)+\mathcal O(\varepsilon^2),
 \label{eq:section8-eigenvalue-response}\\
 \Pi_i(\varepsilon)
 &=\Pi_i+\varepsilon\sum_{j\ne i}
 \frac{\Pi_jA\Pi_i+\Pi_iA\Pi_j}{\lambda_i-\lambda_j}
 +\mathcal O(\varepsilon^2).
 \label{eq:section8-projector-response}
\end{align}
The remainder is controlled by powers of
$|\varepsilon|\|A\|_2/g$.  At a degeneracy, individual eigenvectors and
one-dimensional projectors need not be analytic; separated cluster projectors
remain the appropriate variables.  The polynomial shear invariants and toric
relations remain exact across the degeneracy.  This is another reason not to
identify a spectral-coordinate singularity with failure of the underlying
invariant algebra.  Uniformity follows by representing each separated
projector as a Riesz contour integral of the resolvent and expanding
$(z-M_0-\varepsilon A)^{-1}$ in its uniformly convergent Neumann series on
contours whose distance from the unperturbed spectrum is fixed below by the
gap.  The first resolvent coefficient gives
Eq.~\eqref{eq:section8-projector-response}, and the remaining geometric tail
gives the stated powers of $|\varepsilon|\|A\|_2/g$.

Control of the full path-ordered evolution uses a different small parameter.
Let $S_\varepsilon(L)$ and $S_0(L)$ be the two unitary evolution operators and
\begin{equation}
 \delta H(x)=\frac{\varepsilon}{2E}
  \sum_Aa_A(x)q_A,
 \qquad
 \eta_L=\int_0^Ldx\,\|\delta H(x)\|_2.
 \label{eq:section8-duhamel-control-parameter}
\end{equation}
Duhamel's formula and unitarity imply
\begin{equation}
 \|S_\varepsilon(L)-S_0(L)\|_2\le\eta_L,
 \qquad
 |P_{\alpha\beta}^{(\varepsilon)}-P_{\alpha\beta}^{(0)}|
 \le2\eta_L.
 \label{eq:section8-duhamel-bounds}
\end{equation}
Thus $\eta_L\ll1$ controls oscillation probabilities without constant-density
or adiabatic assumptions, whereas Eq.~\eqref{eq:section8-gap-condition}
controls instantaneous eigenprojectors.  Neither condition is needed for the
exact polynomial identities.

\subsection{Radiative matter potentials and RG closure}
\label{subsec:section8-radiative-rg}

Electroweak loops provide a controlled Standard Model deformation of the
matter potential.  In ordinary unpolarized matter and in the charged-lepton
mass basis, the coherent Standard Model correction is flavor diagonal after
the universal identity part is removed.  It contains both a correction to the
electron charged-current potential and a small flavor-nonuniversal
neutral-current contribution.  The classic flavor-dependent calculation is
Ref.~\cite{BotellaLimMarciano1987}.
Accordingly, for each composition direction one may write
\begin{equation}
 p_f^{\rm 1\ell}
 =p_f^{(0)}+\kappa r_f+\mathcal O(\kappa^2),
 \qquad \kappa\sim\frac{\alpha}{4\pi},
 \qquad r_f\in\mathfrak d_n^0.
 \label{eq:section8-one-loop-diagonal-direction}
\end{equation}
The numerical correction to a conventionally normalized potential can be at
the percent level, with its quoted percentage depending on the chosen input
scheme.  Recent long-baseline studies therefore include it when extracting
small NSI effects~\cite{HuangOhlssonVihonenZhou2025}.

If the composition is fixed, Eq.~\eqref{eq:section8-one-loop-diagonal-direction}
defines one rotated diagonal ray.  The number of diagonal invariants is
unchanged and every off-diagonal edge, triangle, and quadrangle generator
remains exactly constant.  If the coefficients of $p_f^{(0)}$ and $r_f$ vary
independently, the relevant space is
\begin{equation}
 W_{\rm rad}=\Span\{p_f^{(0)},r_f\}_{f},
 \label{eq:section8-radiative-span}
\end{equation}
and Theorem~\ref{thm:complete-diagonal-flow-algebra} applies with
$\ell=\dim W_{\rm rad}$.  Hence radiative corrections lower the diagonal
invariant count only when they add an independently variable direction.  They
can nevertheless raise the effective matrix support rank of Sec.~IV and spoil
a special single-term spectral factorization even when the density flow is
still one dimensional.

For the physical neutral-composition $3+1$ family, a complete one-loop
calculation still multiplies only the two independent electron- and
neutron-density profiles.  It need not, however, preserve the tree
normalization inside the $(x_e,x_\tau,x_s)$ diagonal quotient.  Write the two
renormalized directions as
\begin{equation}
 \bm p_e=(u_e,u_\tau,u_s),
 \qquad
 \bm p_n=(v_e,v_\tau,v_s),
 \label{eq:section8-general-radiative-plane-vectors}
\end{equation}
so that, with $a_n=a_eY_n$,
\begin{equation}
 M_{\rm mat}^{\rm 1\ell}
 =a_e\sum_{i=e,\tau,s}u_i\overline E_i
  +a_n\sum_{i=e,\tau,s}v_i\overline E_i.
 \label{eq:section8-general-radiative-31}
\end{equation}
The coefficients include charged-current rescaling and the
flavor-universal active neutral-current correction, which cannot be discarded
in active--sterile oscillations.  When $\bm p_e$ and $\bm p_n$ are independent,
their exact common diagonal invariant is, up to normalization,
\begin{align}
 d_{\rm rad}^{\rm full}
 &=(\bm p_e\mathbin\times\bm p_n)\mathbin\cdot(x_e,x_\tau,x_s)\nonumber\\
 &=(u_\tau v_s-u_sv_\tau)x_e
 +(u_sv_e-u_ev_s)x_\tau
 +(u_ev_\tau-u_\tau v_e)x_s.
 \label{eq:section8-general-radiative-invariant}
\end{align}
Thus the complete physical correction rotates the existing
electron--neutron plane but does not introduce a third independent density in
ordinary matter.  The coefficients in
Eqs.~\eqref{eq:section8-general-radiative-plane-vectors}--
\eqref{eq:section8-general-radiative-invariant} must be taken from a specified
one-loop input scheme.

It is useful to display separately the leading flavor-nonuniversal
$\tau-\mu$ approximation used in the matter-flow literature
\cite{XingZhu2022}.  Keep the tree normalization of the electron and sterile
directions and absorb the remaining flavor-universal and electron-potential
renormalizations into the density coefficients.  At this stated level of
approximation,
\begin{equation}
 \bm p_e^{\rm lead}=(1,r_e,0),
 \qquad
 \bm p_n^{\rm lead}=\left(0,r_n,\frac12\right),
 \label{eq:section8-leading-radiative-plane-vectors}
\end{equation}
where
\begin{align}
 r(Y_n)&=r_e+r_nY_n
 =\kappa\left[(1+Y_n)L-1-\frac23Y_n\right],
 \label{eq:section8-leading-radiative-ratio}\\
 \kappa&=\frac{3G_{\rm F}m_\tau^2}{2\sqrt2\pi^2},
 \qquad L=\log\frac{m_W^2}{m_\tau^2}.
 \label{eq:section8-leading-radiative-coefficients}
\end{align}
Equivalently,
\begin{equation}
 M_{\rm mat}^{\rm lead}
 =a_e(\overline E_e+r_e\overline E_\tau)
  +a_n\left(\frac12\overline E_s+r_n\overline E_\tau\right),
 \qquad a_n=a_eY_n,
 \label{eq:section8-physical-radiative-31}
\end{equation}
and the cross-product formula reduces to
\begin{equation}
 d_{\rm rad}^{\rm lead}=x_\tau-r_ex_e-2r_nx_s.
 \label{eq:section8-physical-radiative-invariant}
\end{equation}
Equations~\eqref{eq:section8-physical-radiative-31} and
\eqref{eq:section8-physical-radiative-invariant} are therefore a controlled
leading normal form, not the complete one-loop potential.  All
$\mathcal B_4$ cycles and all 55 Markov relations remain unchanged in either
description.  At fixed $\eta_n=a_n/a_e=Y_n$, a convenient pair for the leading
normal form is
\begin{equation}
 d_\tau^{\rm lead}=x_\tau-(r_e+\eta_nr_n)x_e,
 \qquad d_s^{(\eta_n)}=x_s-\frac{\eta_n}{2}x_e.
 \label{eq:section8-fixed-composition-radiative-invariants}
\end{equation}

The effective support rank is not perturbatively continuous.  For the raw
diagonal direction
\begin{equation}
 P_{\rm rad}=\operatorname{diag}(1,0,\rho,\eta)
 \quad\text{modulo }\bm1,
 \label{eq:section8-radiative-rank-example}
\end{equation}
one has
\begin{equation}
 r_{\rm eff}=4-\max_{\lambda\in\{1,0,\rho,\eta\}}
 \operatorname{mult}_{P_{\rm rad}}(\lambda).
 \label{eq:section8-radiative-effective-rank}
\end{equation}
For generic $\eta\notin\{0,1\}$, the tree value $\rho=0$ gives
$r_{\rm eff}=2$, while any nonzero
$\rho\notin\{1,\eta\}$ gives $r_{\rm eff}=3$.  For $n=4$, the full radiative
spurion then has $qr_{\rm eff}>4$ for $q\ge2$, so the common-support theorem
does not certify a nontrivial exterior object of that order.  This is not a
nonexistence claim: accidental lower-support representatives and every
polynomial $\kappa$ satisfying $\cD_{P_{\rm rad}}\kappa=0$ survive.  In
particular, the coordinate-support determinant $\mathcal D_{es}$ in
Eq.~\eqref{eq:section8-raw-krylov-determinant} uses only off-diagonal entries
and remains exact under any diagonal potential, as do all off-diagonal cycles.
The discontinuity affects the general common-support certification and its
special spectral factorization, not the polynomial cycle algebra.

Finite one-loop matching of the Standard Model potential is distinct from
renormalization-group running of NSI Wilson coefficients.  Let
$P_A(\mu)$ be trace-projected matter spurions and let
$\mathcal R_\mu=\mu d/d\mu$ act on their coefficients.  Since
$\mathcal R_\mu$ does not translate $M$,
\begin{equation}
 [\mathcal R_\mu,\cD_A]=\cD_{\mathcal R_\mu P_A}.
 \label{eq:section8-rg-translation-commutator}
\end{equation}
This gives an exact closure criterion.

\begin{proposition}[RG closure]
\label{prop:section8-rg-closure}
Suppose, in the trace quotient,
\begin{equation}
 \mathcal R_\mu P_A=\Gamma_A{}^BP_B.
 \label{eq:section8-rg-span-closure}
\end{equation}
Then the joint kernel $\cap_A\ker\cD_A$ is preserved by $\mathcal R_\mu$.  If
instead
$\mathcal R_\mu P_A=\Gamma_A{}^BP_B+R_A^\perp$, the failure is sourced by the new
directions:
\begin{equation}
 \cD_A(\mathcal R_\mu I)=-\cD_{R_A^\perp}I
 \qquad (\cD_BI=0\ \text{for all }B).
 \label{eq:section8-rg-nonclosure-source}
\end{equation}
\end{proposition}

\emph{Proof.}
From Eq.~\eqref{eq:section8-rg-translation-commutator},
$\cD_A\mathcal R_\mu I=\mathcal R_\mu\cD_AI-[\mathcal R_\mu,\cD_A]I$.
The first term vanishes.  Under Eq.~\eqref{eq:section8-rg-span-closure} the
second is a linear combination of $\cD_BI$; with a transverse component only
Eq.~\eqref{eq:section8-rg-nonclosure-source} remains. \hfill$\square$

In SMEFT and LEFT language the coefficients obey equations of the form
\begin{equation}
 \mu\frac{dC_i}{d\mu}
 =\frac{1}{16\pi^2}\gamma_{ij}C_j.
 \label{eq:section8-wilson-rge}
\end{equation}
Operator mixing can therefore rotate a closed matter-spurion space or generate
a genuinely new direction.  Complete one-loop running across SMEFT and LEFT
has been implemented in NSI analyses~\cite{TerolCalvoTortolaVicente2020}.
The proposition isolates the algebraic content of that running; it does not
replace matching, threshold corrections, or phenomenological bounds.

\subsection{A dressed standard $3+1$ template}
\label{subsec:section8-dressed-31}

Return to the raw fixed-composition $3+1$ direction
$(1,0,0,\eta)$ in the order $(e,\mu,\tau,s)$ and let $q=q^\dagger$ be purely
off diagonal.  With the difference coordinates of
Eq.~\eqref{eq:radiative-difference-coordinates}, choose
\begin{equation}
 s=x_e,
 \qquad
 -d_{\mu\tau}=x_\tau,
 \qquad
 d_s^{(\eta)}=x_s-\eta x_e.
 \label{eq:section8-31-slice}
\end{equation}
For
$M(a)=M_0+a[p_0+\varepsilon q]$, the two diagonal forms remain exact and the
dressed edges are
\begin{equation}
 \widehat z_{\alpha\beta}
 =z_{\alpha\beta}-\varepsilon x_eq_{\alpha\beta}.
 \label{eq:section8-31-dressed-edges}
\end{equation}
Define
\begin{align}
 \widehat Q_{\alpha\beta}
 &=\widehat z_{\alpha\beta}\widehat z_{\beta\alpha},
 \label{eq:section8-31-dressed-q}\\
 \widehat T_{\alpha\beta\gamma}
 &=\widehat z_{\alpha\beta}\widehat z_{\beta\gamma}
   \widehat z_{\gamma\alpha},
 \label{eq:section8-31-dressed-t}\\
 \widehat C_{\alpha\beta\gamma\delta}
 &=\widehat z_{\alpha\beta}\widehat z_{\beta\gamma}
   \widehat z_{\gamma\delta}\widehat z_{\delta\alpha}.
 \label{eq:section8-31-dressed-c}
\end{align}
They are constant along the deformed flow and obey, without approximation,
the complete 55-binomial Markov basis of
Sec.~\ref{subsec:section7-markov-basis}.
Consequently the abstract standard $3+1$ toric presentation survives as a
dressed relative algebra, even though none of the undressed edge cycles that
contains a translated NSI edge need remain constant.

The statement is covariant under simultaneous rephasing of $(M,q)$.  If a
numerical $q$ with connected support is held fixed, its pointwise rephasing
stabilizer can be trivial; one must not reinterpret the dressed relative
generators as invariants of the full torus acting on $M$ alone.  This example
therefore realizes all three levels of the framework: an exact affine
translation quotient, an exact spurion-covariant dressed toric algebra, and a
controlled expansion relative to the standard diagonal generators.

\subsection{Hierarchy of survival statements}
\label{subsec:section8-survival-hierarchy}

The results of this section can be summarized as follows.
\begin{enumerate}
 \item A fixed or independently varying set of Hermitian NSI directions always
 defines commuting exact translations.  Before rephasing, its invariant ring
 is the coordinate ring of the linear quotient by their span.
 \item Diagonal corrections leave the entire off-diagonal cycle algebra
 unchanged.  They change the diagonal invariant count only by increasing the
 dimension of the independently varying diagonal span.
 \item An off-diagonal NSI tied to independent diagonal anchor directions has
 an exact dressed algebra obtained by the shear of
 Theorem~\ref{thm:section8-exact-nsi-shear}.  The old cycles drift at order
 $\varepsilon$, while their toric syzygies remain algebraically exact.
 \item An NSI coefficient varied independently of every diagonal anchor is a
 new translation direction even when numerically small.  Exact survivors must
 lie in its kernel already, and the quotient dimension can jump at zero
 coupling.
 \item Perturbations that preserve the common support retain the same
 rank-$r$ exterior-power invariants.  Enlarging the support removes this
 certificate and generally destroys a one-term spectral factorization.
 Special representatives survive exactly when $\cD_Q\kappa=0$; minors such as
 $\mathcal D_{es}$ explicitly survive diagonal deformations.
 \item Polynomial identities cross spectral degeneracies without difficulty;
 perturbative eigenprojectors require a separated spectral cluster and a norm
 bound such as Eq.~\eqref{eq:section8-gap-condition}.
 \item Radiative or EFT running preserves the matter-invariant bundle exactly
 when the running spurions close on their current span.  Operator mixing outside
 that span supplies the controlled breaking term in
 Eq.~\eqref{eq:section8-rg-nonclosure-source}.
\end{enumerate}

These distinctions turn ``NSI breaking'' into a set of testable algebraic
questions: whether a direction is diagonal, whether several densities are
independent, whether a common Krylov support survives, whether the requested
correction class is large enough, and whether radiative running closes on the
same spurion span.  Section~IX supplies exact tests across generic and
exceptional support strata, together with numerical flows at generic
benchmarks.

\section{Numerical and symbolic verification}
\label{sec:numerical-symbolic-verification}

The proofs in the preceding sections are independent of numerical sampling.
This section supplies a second layer of validation aimed at signs, ordering
conventions, boundary strata, and reproducibility.  The accompanying source
bundle separates exact rational certificates from floating-point stress
tests.  An exact failure raises an exception immediately; a floating test
records a reported absolute or scale-normalized residual and fails at its
declared tolerance.

\subsection{Reproducibility protocol}
\label{subsec:section9-reproducibility}

The complete suite is launched by
\texttt{python verification/verify\_all.py}.  It records the Python, NumPy,
and SciPy versions, fixed seeds, script hashes, declared tolerances, exact
certificate counts, the reported maxima, and sample accounting for the listed
flow grids in machine-readable JSON files.  The release run used Python
3.12.13, NumPy 2.3.5, and SciPy 1.17.0.  The exact-certificate polynomial,
determinant, rank, fiber, and Buchberger calculations use integers or
\texttt{fractions.Fraction};
NumPy and SciPy are confined to explicitly numerical blocks.  The principal
floating-point seed is 20260903.

The $K_4$ calculation is deliberately duplicated.  One program constructs the
graph and cycle columns, expands the five displayed representatives under
$S_4$, and validates their orbits, toric fibers, and Jacobians.  A second exact
binomial Buchberger implementation uses the
fixed lexicographic variable order of
Eq.~\eqref{eq:section7-buchberger-variable-order}.  The two routes share no
stored list of the 55 relations.  Direct circulation enumeration, the
$A_3$ flow-theta parametrization, and the rational Hilbert series are also
compared coefficient by coefficient through edge degree 18.

\subsection{Exact certificates}
\label{subsec:section9-exact-certificates}

Table~\ref{tab:section9-exact-ledger} summarizes the exact layer.  The
common-support test uses $n=7$, $r=2$, $q=3$, and $k=2$, with a noncoordinate support, two
noncommuting support matrices, and nonzero scalar shifts.  All seven maximal
minors of the $6\times7$ Krylov matrix are evaluated at three rational
two-parameter matter points.  The insertion formula for
$\partial_AM^p$ supplies an independent entrywise check of
Eq.~\eqref{eq:exact-block-row-flow}.  Square tests at $(6,2,3)$ and $(5,1,5)$
verify the higher-rank and rank-one determinant cases.

For the Pl\"ucker test, every four-column minor in a $(5,2,2)$ example is
computed both as a direct alternant determinant and as the signed ordered-
partition sum of Theorem~\ref{thm:general-krylov-plucker-expansion}.  A
rational orthogonal flavor transformation then tests every Cauchy--Binet
pullback minor independently.  The rank-one Vandermonde product is exact.
At rank two, an exact pair-collision example has $F_{2,2}=2$, demonstrating
that one spectral gap is not a universal divisor.  The three Pl\"ucker
boundary factorizations, a uniform-matroid cancellation point, and the
pair-versus-triple degeneracy distinction are checked separately.

\begin{table}[t]
\centering
\small
\caption{Exact verification ledger.  ``Exact'' means equality over the
integers, rationals, or Gaussian rationals; no numerical tolerance is used.}
\label{tab:section9-exact-ledger}
\begin{tabular}{@{}p{0.06\linewidth}p{0.25\linewidth}p{0.42\linewidth}p{0.15\linewidth}@{}}
\toprule
ID & Object & Certificate & Result\\
\midrule
D & diagonal flows & $n=5$, rank-two annihilators and fixed-versus-independent $3+1$ nullspaces & exact equality\\
K & rank-$r$ Krylov & six row-flow identities, seven nonsquare minors, and two square determinants & exact constancy\\
P & Pl\"ucker pullback & five ordered-partition identities and five Cauchy--Binet identities & exact equality\\
F & factorization strata & rank-one product, three rank-two boundary products, collision and cancellation witnesses & exact equality\\
T & $K_4$ toric ideal & 20 cycles, 55 two-monomial fibers, orbit sizes $4,12,24,3,12$ & exact counts\\
H & Hilbert certificate & direct circulations, $A_3$ theta series, and initial ideal through degree 18 & identical coefficients\\
S & physical strata & exact $55\times20$ Jacobian on all 64 support graphs & $38$ smooth, $26$ singular\\
N & NSI and radiation & general straightening, tied shear, independent rank jump, and diagonal survivors & exact equality\\
\bottomrule
\end{tabular}
\end{table}

The toric programs reconstruct the cycle counts $(6,8,6)$, the exact ranks
\begin{equation}
 \rank B_4=3,
 \qquad \rank A_4=9,
 \qquad \operatorname{ht}I_4=11,
 \label{eq:section9-k4-rank-certificate}
\end{equation}
and the 55 indispensable fibers in weighted counts $16,24,15$ at degrees
$6,7,8$.  The archived Buchberger order gives 123 binomials and 107 minimal
initial monomial generators.  A second insertion order gives 161 unreduced
binomials but the same 107 minimal initial generators and the same Hilbert
numerator, as expected because an unreduced completion size is not canonical.
The three Hilbert routes give
\begin{multline}
 (h_0,\ldots,h_{18})=
 (1,0,6,8,27,48,112,192,378,624,1092,1728,2802,\nonumber\\
 4248,6516,9528,13983,19824,28090),
 \label{eq:section9-hilbert-coefficients}
\end{multline}
which agrees with Eq.~\eqref{eq:section6-univariate-hilbert-series}.

For the singular-locus test, every present Hermitian edge receives a distinct
positive integer, all twenty cycle coordinates are evaluated, and the
Jacobian of the 55 binomials is row reduced over the rationals.  All 38
connected labeled graphs have rank 11.  The 26 disconnected graphs have
ranks $0$, $2$, or $7$, never 11.  This exhaustive support check reproduces
Theorem~\ref{thm:section7-physical-singular-locus}; the theorem's slice proof,
rather than this finite list of representative points, controls every point
of each stratum.  Exact CP tests include a rephased-real full support, an
isolated CP-odd triangle, and a chordless square.  The last has vanishing triangle
odd parts and a nonzero quadrangle odd part, confirming that the seven global
conditions in Eq.~\eqref{eq:section7-global-cp-equations} cannot be replaced
globally by triangle conditions alone.

\subsection{Physics-shaped floating-point stress tests}
\label{subsec:section9-floating-tests}

The numerical benchmark is an illustrative algebraic stress point, not a
fit.  It uses
\begin{equation}
 (m_1^2,m_2^2,m_3^2,m_4^2)
 =(0,7.49\times10^{-5},2.513\times10^{-3},1)\ \mathrm{eV}^2,
 \label{eq:section9-benchmark-masses}
\end{equation}
The mixing convention is fixed by
\begin{equation}
 U=R_{34}\widetilde R_{24}(\theta_{24},\delta_{24})
 \widetilde R_{14}(\theta_{14},\delta_{14})R_{23}
 \widetilde R_{13}(\theta_{13},\delta_{13})R_{12},
 \label{eq:section9-mixing-convention}
\end{equation}
where each complex rotation is the identity off the $(i,j)$ subspace and
\begin{equation}
 \left.\widetilde R_{ij}(\theta_{ij},\delta_{ij})\right|_{ij}
 =\begin{pmatrix}
 c_{ij}&s_{ij}e^{-\ii\delta_{ij}}\\
 -s_{ij}e^{\ii\delta_{ij}}&c_{ij}
 \end{pmatrix},
 \qquad R_{ij}=\widetilde R_{ij}(\theta_{ij},0).
 \label{eq:section9-complex-rotation}
\end{equation}
with
$\theta_{12}=33.45^\circ$, $\theta_{13}=8.62^\circ$,
$\theta_{23}=48.6^\circ$, $(\theta_{14},\theta_{24},\theta_{34})=(8^\circ,
6^\circ,10^\circ)$, and nonzero phases
$(\delta_{13},\delta_{14},\delta_{24})=(230^\circ,70^\circ,140^\circ)$.
It produces a unitary, full-support, CP-violating $M_0$ and avoids accidental
matroid boundaries.

For a constant-density scale we use
\begin{equation}
 a_e=1.52\times10^{-4}\ \mathrm{eV}^2
 \left(\frac{\rho}{{\rm g\,cm}^{-3}}\right)Y_e
 \left(\frac{E}{{\rm GeV}}\right),
 \qquad a_n=a_eY_n,
 \qquad a_s=\frac{a_n}{2}.
 \label{eq:section9-matter-normalization}
\end{equation}
At $\rho=2.85\ \mathrm{g\,cm}^{-3}$, $Y_e=0.5$, $E=2.5$ GeV, and
$Y_n=1$, this gives $a_e=a_n=5.415\times10^{-4}\ \mathrm{eV}^2$ and
$a_s=2.7075\times10^{-4}\ \mathrm{eV}^2$.  An exact rational
$(a_e,a_s)$ grid separately checks the joint kernel: $d_{\mu\tau}$ remains
constant, whereas $d_s^{(\eta)}$ remains at its vacuum value only on the ray
$a_s=\eta a_e$.  This separates fixed composition from independently varying
densities without identifying the algebraic grid with the single physical
benchmark point.

The tied propagation-NSI stress test uses active-block Hermitian spurions of
unit spectral norm and $\varepsilon=0.03$.  Along 101 points of the ray
$a_n=a_e$, the raw triangle and quadrangle drift by $3.49\%$ and $2.70\%$,
while their exactly sheared versions remain constant at relative residuals
$4.5\times10^{-16}$ and $4.3\times10^{-16}$.  Varying the NSI coefficients
independently instead gives the exact translation-rank jump $2\to4$ and is
tested against the joint-kernel condition, not against the tied shear.

For radiative matter, the leading normal-form coefficients are
\begin{equation}
 r_e=2.62136\times10^{-5},
 \qquad r_n=2.75328\times10^{-5}.
 \label{eq:section9-leading-radiative-numbers}
\end{equation}
Both the general cross-product invariant
Eq.~\eqref{eq:section8-general-radiative-invariant} and the leading invariant
Eq.~\eqref{eq:section8-physical-radiative-invariant} are checked on independent
two-density grids.  The full-plane coefficient vectors used by the code are
deterministic generic algebra-test directions, not a complete one-loop matching
result in a specified input scheme.  Only the quoted $(r_e,r_n)$ pair
instantiates the qualified leading $\tau-\mu$ normal form.  The test retains
the distinction between these
coefficients and a scheme-dependent percent-level renormalization of the
electron potential.

\begin{table}[t]
\centering
\small
\caption{Floating-point regression results.  Residuals are maxima over the
stated grids; the exact polynomial statements are certified separately in
Table~\ref{tab:section9-exact-ledger}.}
\label{tab:section9-floating-ledger}
\begin{tabular}{@{}p{0.33\linewidth}p{0.20\linewidth}p{0.18\linewidth}p{0.16\linewidth}@{}}
\toprule
Check & Sample & Maximum result & Acceptance\\
\midrule
spectral Krylov/Pl\"ucker pullback & two 7-point square flows; one nonsquare matrix (5 minors) & $2.58\times10^{-14}$ & $\leq10^{-10}$\\
mass-label and phase covariance & 14 transformations & $5.50\times10^{-15}$ & $\leq10^{-10}$\\
tied-NSI dressed triangle & 101 matter points & $4.5\times10^{-16}$ relative & $\leq10^{-12}$\\
tied-NSI dressed quadrangle & 101 matter points & $4.3\times10^{-16}$ relative & $\leq10^{-12}$\\
Duhamel probability bound & $L=1300$ km, $E=2.5$ GeV & $5.80\times10^{-3}$ & $\leq2\eta_L+10^{-12}$\\
\bottomrule
\end{tabular}
\end{table}

The path-ordered evolution test uses the same propagation NSI at constant
density.  It finds $\eta_L=3.472\times10^{-2}$,
$\max_{\alpha\beta}|\Delta P_{\alpha\beta}|=5.792\times10^{-3}$, and hence
the channel changes lie below the bound $2\eta_L=6.944\times10^{-2}$.  This
single-energy calculation tests Eq.~\eqref{eq:section8-duhamel-bounds}; it is
not an experimental sensitivity forecast.

The largest numerical residual in the algebraic spectral suite is
$2.58\times10^{-14}$.  The exact pair- and triple-degeneracy certificates are
evaluated in flavor coordinates.  The floating spectral tests use absolute
residuals at generic nondegenerate points; the tied-NSI entries use relative
residuals only at benchmark cycles verified to be nonzero.

\subsection{What the computations establish}
\label{subsec:section9-interpretation}

The exact computations reproduce the finite certificates used in the toric
presentation and independently check representative instances of the
principal general identities listed in Table~\ref{tab:section9-exact-ledger}.
The floating tests verify the spectral dictionary,
covariance conventions, and perturbative control estimates.  The all-rank,
all-degree, and all-stratum conclusions continue to follow from the proofs:
finite sampling illustrates those conclusions and detects implementation
errors, while exact algebra establishes them.

\section{Conclusions}
\label{sec:conclusions}

Matter dependence in neutrino oscillations admits a unified algebraic
description as commuting affine translations of the flavor-basis Hamiltonian
together with flavor rephasing.  The physical polynomial invariants are the
joint kernel of these actions, or equivalently the coordinate algebra of the
corresponding quotient.  This formulation distinguishes constants along a
chosen density path, globally regular polynomial generators, and the more
restrictive property of a one-term spectral factorization.

For arbitrary independent diagonal matter directions, the invariant ring
splits into the polynomial algebra of diagonal annihilators and the toric
algebra of balanced off-diagonal cycles, with dimension
$n(n-1)-\ell$.  When the matter spurions have common effective rank $r$, the
block-Krylov rows evolve by a strictly lower-block-triangular system.  If
$qr\leq n$, the resulting maximal minors are exact exterior covariants, and
their spectral pullback is
the complete Krylov--Pl\"ucker expansion.  Rank one yields the universal
Vandermonde product.  In the square case $n=qr$ with $q,r\geq2$, generic
higher rank has no analogous separated Pl\"ucker--spectral factor and no
universal pairwise gap divisor.  The exhaustive $4\times4$, rank-two
reducibility classification at fixed Pl\"ucker point identifies precisely the
three nonzero Pl\"ucker-boundary factorizations and the identically vanishing
stratum.

For the standard fixed-composition $3+1$ fiber, the global off-diagonal algebra is generated
minimally by six two-cycles, eight directed triangles, and six directed
quadrangles.  Its prime toric ideal has height eleven and a unique minimal
Markov basis of 55 indispensable binomials.  On the Hermitian form these give
31 CP-even and 24 CP-odd real equations.  The seven triangle and quadrangle
odd parts provide a global CP criterion, including chordless-square supports
that are invisible to triangle phases.  The physical quotient is smooth
exactly on connected support graphs; its singular intersection is the union
of the seven maximal bipartition faces.  These global results explain why the
familiar eleven quantities work on dense sign-resolved charts but do not form
a global polynomial generating set.

Additional matter physics fits naturally into the same hierarchy.  Several
composition profiles require the joint kernel of the corresponding species
directions.  In the universal simultaneous-spurion algebra, off-diagonal
propagation NSI tied to diagonal anchors admit an exact spurion-covariant shear
and hence an exactly dressed cycle algebra before specialization to a fixed
spurion.  Independently varying NSI shrink the quotient discontinuously at
nonzero coupling precisely when their directions enlarge the translation
span.  Diagonal radiative corrections preserve every off-diagonal cycle while
rotating the diagonal annihilator subspace and possibly lifting the effective
Krylov support rank; the displayed radiative coefficients are a leading normal
form, not a complete one-loop or EFT matching calculation.  Spectral
perturbation and full evolution are controlled by distinct
parameters: a uniform eigengap for instantaneous projectors and the Duhamel
integral for path-ordered probabilities.

The resulting framework supplies algebraic inputs and exact consistency
relations that can be incorporated into reconstruction, density and energy
averaging, and NSI analyses.  Natural extensions are automated invariant generation for
larger flavor graphs, factorization over general realizable matroid strata,
SMEFT--LEFT matching of the matter-spurion space across thresholds, and
propagation algebras for non-Hermitian, open-system, or neutrino
self-interaction Hamiltonians.  The present work does not perform a global
sterile/NSI fit or sensitivity forecast, treat source/detector NSI or establish
a UV completion, remove path ordering, or by itself resolve generalized
ordering degeneracies.  In this form, the nonlinear response of
effective masses and mixing parameters is organized by a low-dimensional
affine geometry in the flavor basis.

\bibliographystyle{unsrt}
\bibliography{low_rank_matter_flows}

\end{document}